\documentclass[11pt,a4paper]{article}
\pdfoutput=1

\usepackage{amssymb}
\usepackage{amsmath}
\usepackage{amsfonts}
\usepackage{bbm}
\usepackage{amsthm}
\usepackage{mathrsfs}
\usepackage{hyperref}
\usepackage{color}
\usepackage[margin=2.41cm]{geometry}
\usepackage[all,cmtip]{xy}
\usepackage[utf8]{inputenc}
\usepackage{graphicx}
\usepackage{varwidth}
\usepackage{comment}
\usepackage[shortlabels]{enumitem}

\usepackage{mathtools}

\usepackage{upgreek}
\usepackage{rotating}

\usepackage{tikz}
\usetikzlibrary{shapes.geometric}

\usepackage{tikz-cd}
\usetikzlibrary{matrix,arrows,calc,decorations.pathmorphing,fit,shapes.geometric,decorations.pathreplacing,positioning}

\usepackage{amssymb}
\tikzset{curve/.style={settings={#1},to path={(\tikztostart)
    .. controls ($(\tikztostart)!\pv{pos}!(\tikztotarget)!\pv{height}!270:(\tikztotarget)$)
    and ($(\tikztostart)!1-\pv{pos}!(\tikztotarget)!\pv{height}!270:(\tikztotarget)$)
    .. (\tikztotarget)\tikztonodes}},
    settings/.code={\tikzset{quiver/.cd,#1}
        \def\pv##1{\pgfkeysvalueof{/tikz/quiver/##1}}},
    quiver/.cd,pos/.initial=0.35,height/.initial=0}
\definecolor{darkred}{rgb}{0.8,0.1,0.1}
\hypersetup{
	colorlinks=true,         % false: boxed links; true: colored links,false is default
	linkcolor=darkred,
	citecolor=blue,
}

\theoremstyle{plain}
\newtheorem{theo}{Theorem}[section]
\newtheorem{lem}[theo]{Lemma}
\newtheorem{propo}[theo]{Proposition}
\newtheorem{cor}[theo]{Corollary}

\theoremstyle{definition}
\newtheorem{defi}[theo]{Definition}

\newtheorem{assu}[theo]{Assumption}

\newtheorem{exx}[theo]{Example}
\newenvironment{ex}
{\pushQED{\qed}\exx}
{\popQED\endexx}

\newtheorem{remm}[theo]{Remark}
\newenvironment{rem}
{\pushQED{\qed}\remm}
{\popQED\endremm}

\newtheorem{constrr}[theo]{Construction}
\newenvironment{constr}
{\pushQED{\qed}\constrr}
{\popQED\endconstrr}

\numberwithin{equation}{section}

\renewcommand{\emptyset}{\varnothing}

\def\nn{\nonumber}

\def\bbK{\mathbb{K}}

\def\bbC{\mathbb{C}}

\def\ii{{\,{\rm i}\,}}

\def\Hom{\mathrm{Hom}}

\def\id{\mathrm{id}}
\def\Id{\mathrm{Id}}

\def\vol{\mathrm{vol}}
\def\cl{\mathrm{cl}}

\def\cc{\mathrm{c}}

\def\1{I}
\def\oone{\mathbbm{1}}
\def\op{\mathrm{op}}

\def\pr{\mathrm{pr}}

\def\Mor{\operatorname{Mor}}

\def\Loc{\mathbf{Loc}}

\def\Lan{\operatorname{Lan}}
\def\lan{\operatorname{lan}}

\def\COpen{\mathbf{COpen}}
\def\RC{\mathbf{RC}}

\def\Alg{\mathbf{Alg}}

\def\Vec{\mathbf{Vec}}

\def\CC{\mathbf{C}}
\def\DD{\mathbf{D}}

\def\EE{\mathbf{E}}
\def\FF{\mathbf{F}}

\def\TT{\mathbf{T}}
\def\Cat{\mathbf{Cat}}
\def\CAT{\mathbf{CAT}}

\def\AQFT{\mathbf{AQFT}}

\def\Pr{\mathbf{Pr}}
\def\Fun{\mathbf{Fun}}

\def\AAA{\mathfrak{A}}
\def\LLL{\mathfrak{L}}
\def\BBB{\mathfrak{B}}

\def\RRR{\mathfrak{R}}

\def\K{\mathcal{K}}

\def\O{\mathcal{O}}

\def\U{\mathcal{U}}
\def\V{\mathcal{V}}

\def\KG{\mathrm{KG}}

\def\HK{\mathsf{HK}}
\def\calHK{\mathcal{HK}}

\def\colim{\mathrm{colim}}

\def\bilim{\mathrm{bilim}}
\def\bicolim{\mathrm{bicolim}}

\newcommand\und[1]{\underline{#1}}
\newcommand\ovr[1]{\overline{#1}}

\DeclareMathOperator*{\Motimes}{\text{\raisebox{0.25ex}{\scalebox{0.8}{$\bigotimes$}}}}

\def\sk{\vspace{2mm}}

\makeatletter
\let\@fnsymbol\@alph
\makeatother

\title{%
Haag-Kastler stacks
}

\author{%
Marco Benini$^{1,2,a}$, Alastair Grant-Stuart$^{3,b}$\ and\ Alexander Schenkel$^{3,c}$\vspace{4mm}\\
{\small ${}^1$ Dipartimento di Matematica, Dipartimento di Eccellenza 2023-27, Universit\`a di Genova,}\\
{\small Via Dodecaneso 35, 16146 Genova, Italy.}\vspace{2mm}\\
{\small ${}^2$ INFN, Sezione di Genova,}\\
{\small Via Dodecaneso 33, 16146 Genova, Italy.}\vspace{2mm}\\
{\small ${}^3$ School of Mathematical Sciences, University of Nottingham,}\\
{\small University Park, Nottingham NG7 2RD, United Kingdom.}\vspace{4mm}\\
{\small \begin{tabular}{ll}
Email: & ${}^a$~\href{mailto:marco.benini@unige.it}{\texttt{marco.benini@unige.it}}\\
&${}^b$~\href{mailto:alastair.grant-stuart@nottingham.ac.uk}{\texttt{alastair.grant-stuart@nottingham.ac.uk}}\\
& ${}^c$~\href{mailto:alexander.schenkel@nottingham.ac.uk}{\texttt{alexander.schenkel@nottingham.ac.uk}}
\vspace{2mm}
\end{tabular}
}
}

\date{November 2025}

\begin{document}

\maketitle

\begin{abstract}
\noindent This paper provides an alternative implementation of the principle of general local covariance for algebraic quantum field theories (AQFTs) which is more flexible than the original one by Brunetti, Fredenhagen and Verch. This is realized by considering the $2$-functor $\mathsf{HK} : \mathbf{Loc}^\mathrm{op} \to \mathbf{CAT}$ which assigns to each Lorentzian manifold $M$ the category $\mathsf{HK}(M)$ of Haag-Kastler-style AQFTs over $M$ and to each embedding $f:M\to N$ a pullback functor $f^\ast = \mathsf{HK}(f) : \mathsf{HK}(N) \to \mathsf{HK}(M)$ restricting theories from $N$ to $M$. Locally covariant AQFTs are recovered as the points of the $2$-functor $\mathsf{HK}$. The main advantages of this new perspective are: 1.)~It leads to technical simplifications, in particular with regard to the time-slice axiom, since global problems on $\mathbf{Loc}$ become families of simpler local problems on individual Lorentzian manifolds. 2.)~Some aspects of the Haag-Kastler framework which previously got lost in locally covariant AQFT, such as a relative compactness condition on the open subsets in a Lorentzian manifold $M$, are reintroduced. 3.)~It provides a radically new perspective on descent conditions in AQFT, i.e.\ local-to-global conditions which allow one to recover a global AQFT on a Lorentzian manifold $M$ from its local data in an open cover $\{U_i \subseteq M\}$.
\end{abstract}
\vspace{-1mm}

\vspace{-2mm}

\vspace{-2mm}

\tableofcontents

%%%%%%%%%%%%%%%%%%%%%%%%%%%%%%%%%%%%%%%%%%%%%%%%
%%%%%%%%%%%%%%%%%%%%%%%%%%%%%%%%%%%%%%%%%%%%%%%%

\section{\label{sec:intro}Introduction and summary}
In its traditional form as proposed by Haag and Kastler \cite{HaagKastler},
algebraic quantum field theory (AQFT) studies nets of operator algebras
which are defined on suitable open subsets of the Minkowski spacetime
and are endowed with an action of the Poincar{\'e} group, i.e.\ the group of automorphisms
of the Minkowski spacetime. 
This axiomatic approach therefore combines the key principles of quantum theory and special relativity, 
leading to a powerful framework in which one can prove a variety of 
model-independent results about quantum field theories, such as the CPT 
and spin-statistics theorems, as well as general results about
scattering theory, see e.g.\ Haag's monograph \cite{Haag}.
Much of this power however gets lost when one replaces the Minkowski spacetime
by an arbitrary (oriented, time-oriented and globally hyperbolic) Lorentzian manifold $M$,
because the latter will in general have no non-trivial geometric automorphisms.
This issue led Brunetti, Fredenhagen and Verch to propose their
principle of general local covariance for AQFTs \cite{BFV}. This principle was implemented 
by replacing the Haag-Kastler point of view, i.e.\ that an AQFT is
formulated with respect to open subsets in an individual Lorentzian manifold,
with the proposal that a locally covariant AQFT is a functor $\AAA : \Loc\to \Alg$ 
out of the category of \textit{all}
(oriented, time-oriented and globally hyperbolic) Lorentzian manifolds $M$
with morphisms given by (isometric, causal and open) embeddings $f:M\to N$,
which is subject to some physically motivated conditions such as causality and the time-slice axiom.
See also \cite{FewsterVerch} for a more modern presentation of these ideas.
Through this enhanced functoriality, locally covariant AQFT regains much of 
the power of Haag-Kastler AQFT on the Minkowski spacetime, leading to a variety
of model-independent results for AQFTs on Lorentzian manifolds 
which extend the traditional ones on the Minkowski spacetime, see e.g.\ 
\cite{FewsterVerchSPASS,FewsterSpinStatistics,Fewster}. Furthermore, 
the principle of general local covariance is the key to developing
renormalization techniques for AQFTs on Lorentzian manifolds
\cite{Hollands1,Hollands2,Hollands3,Khavkine}, see also Rejzner's monograph \cite{Rejzner},
and it plays a pivotal role in applications to physics, 
see e.g.\ \cite{Hack,FewsterVerchMeasurement}.
\sk

The aim of this paper is to provide an alternative and 
more flexible implementation of the 
principle of general local covariance for AQFTs which 
is compatible with, but considerably enhances, the original 
Haag-Kastler viewpoint that AQFTs are defined on
suitable open subsets of a Lorentzian manifold $M\in\Loc$.
Starting from the latter viewpoint, there is a category $\HK(M)$ whose objects are all
Haag-Kastler-style AQFTs over $M$ and whose morphisms are natural transformations.
The key idea is to consider not only one of these categories, but rather
the whole family of categories $\{\HK(M):M\in\Loc\}$,
for all Lorentzian manifolds $M\in\Loc$, and endow it with additional
structure describing the behavior of Haag-Kastler-style AQFTs under $\Loc$-morphisms $f:M\to N$.
We assemble these structures into a (contravariant) $2$-functor $\HK : \Loc^\op\to \CAT$
from the category $\Loc$ of Lorentzian manifolds and their embeddings
to the $2$-category $\CAT$ of categories, functors and natural transformations.
This $2$-functor assigns to each Lorentzian manifold $M\in \Loc$ the category
$\HK(M)$ of all Haag-Kastler-style AQFTs over $M$ and to each 
embedding $f:M\to N$ in $\Loc$ a pullback functor $f^\ast = \HK(f) : \HK(N)\to \HK(M)$
which describes the restriction along $f$ of AQFTs over $N$ to AQFTs over $M$.
See Definition \ref{def:HK2functor} for the precise definition of this $2$-functor.
Our approach identifies 
the locally covariant AQFTs from \cite{BFV,FewsterVerch} with the points of the $2$-functor $\HK$,
i.e.\ pseudo-natural transformations $\AAA : \Delta\mathbf{1}\Rightarrow \HK$
from the constant $2$-functor $\Delta\mathbf{1} : \Loc^\op\to \CAT$
assigning the one-object category $\mathbf{1}\in\CAT$, 
see Theorem \ref{theo:LCQFTvsHK} and Corollary \ref{cor:LCQFTvsHKW}. This identification
gives a precise meaning to the following intuitive slogan: ``A locally covariant AQFT is the same datum
as a natural family of Haag-Kastler-style AQFTs over all $M\in\Loc$.''
It also shows that our framework is richer than locally covariant AQFT,
since the $2$-functor $\HK$ contains more structure than its category of points 
$\Delta\mathbf{1}\Rightarrow \HK$, see e.g.\ Remark \ref{rem:(op)laxpoints} for an example.
\sk

Even though our approach might superficially look more complicated 
than locally covariant AQFT, owing to its use of some
$2$-categorical concepts and techniques, it actually leads to 
certain technical simplifications. Loosely speaking, the origin
of these simplifications lies in the fact that our approach turns 
complicated global problems for the category $\Loc$ into 
families of simpler local problems for the categories of causally convex
open subsets $\COpen(M)$ in all individual Lorentzian manifolds $M\in\Loc$.
An example for this is given by the time-slice axiom, which from a structural point of view
corresponds to a localization of the relevant spacetime category at all Cauchy morphisms,
see e.g.\ \cite{BSWoperad} and \cite{BSreview,BSchapter} for reviews.
The localization of the category $\Loc$ at all Cauchy morphisms 
is in general notoriously difficult to understand and describe, with notable
exceptions given by the very special cases of $1$-dimensional Lorentzian manifolds \cite{BCStimeslice}
and of $2$-dimensional conformal Lorentzian manifolds \cite{BGS},
while the localizations of the categories $\COpen(M)$ at all Cauchy morphisms
admit very simple and intuitive models, see \cite{BDSboundary} 
and also Example \ref{ex:localizations} and Appendix \ref{app:localization}.
When combined with the identification from Corollary \ref{cor:LCQFTvsHKW}
between locally covariant AQFTs and points of the $2$-functor $\HK$, 
this leads to a technically 
useful perspective on the time-slice axiom in locally covariant AQFT.
These novel ideas and techniques have recently found a nontrivial
application in the context of the $\infty$-categorical equivalence
problem between AQFTs and factorization algebras  \cite{tPFAAQFTinfty}.
\sk

A notable feature of our framework is that it allows us to reintroduce
some aspects of Haag-Kastler-style AQFTs which previously got lost in
the generalization of \cite{BFV} to locally covariant AQFT.
For instance, in the original framework \cite{HaagKastler}
one assigns algebras only to \textit{relatively compact} open subsets of the Minkowski spacetime,
but there is no remnant of this restriction in locally covariant AQFT \cite{BFV}.
In our approach one can easily include such relative compactness conditions
by considering the $2$-functor $\HK^{\mathrm{rc}} : \Loc^\op\to \CAT$
that assigns to each Lorentzian manifold $M\in\Loc$ the category
$\HK^{\mathrm{rc}}(M)$ of all Haag-Kastler-style AQFTs on $M$ which are modeled
on the full subcategory $\RC(M)\subseteq \COpen(M)$ of relatively compact causally convex opens in $M$.
See Definition \ref{def:RCHK2functor} for the precise definition of this $2$-functor.
One can then show that relative compactness does interplay well with the time-slice axiom,
leading to very simple and intuitive models for the localizations of the categories $\RC(M)$
at all Cauchy morphisms, see Example \ref{ex:localizations} and Appendix \ref{app:localization}.
An interesting observation is that the points $\Delta\mathbf{1} \Rightarrow \HK^{\mathrm{rc}}$ 
of the relatively compact $2$-functor $\HK^{\mathrm{rc}}$ are related to, but not exactly the same as
additive locally covariant AQFTs, see Corollaries \ref{cor:additiveLCQFT} and \ref{cor:WadditiveLCQFT}.
This indicates that the relatively compact $2$-functor $\HK^{\mathrm{rc}}$ 
implements a weakened variant of the additivity property of locally covariant AQFTs.
\sk

Another significant novelty of our framework is that it provides
a radically new approach to the issue of descent conditions in AQFT,
i.e.\ local-to-global conditions which allow one to recover 
the global datum of an AQFT on a complicated Lorentzian manifold $M$
from its local data in simpler regions. Descent conditions have been 
relatively little studied in the context of AQFT, even though they are 
technically powerful and also conceptually interesting as they provide mathematical
realizations of the slogan that ``quantum field theory should be local''.
One of the reasons why descent conditions did not receive much attention in AQFT 
could be that in their most basic form, which is given by a cosheaf
condition with respect to causally convex open covers $\{U_i\subseteq M\}$
for the underlying functors $\AAA : \Loc\to \Alg$,
they are \textit{not} satisfied even by the simplest examples, see \cite[Appendix A]{BSreview}.
An alternative to such cover-type descent conditions is given by Fredenhagen's universal
algebra \cite{Fredenhagen1,Fredenhagen3,Fredenhagen2}, which can be realized 
in terms of a left Kan extension of AQFTs along the functor $\Loc_{\diamond}\to \Loc$
from contractible Lorentzian manifolds to all Lorentzian manifolds,
see \cite{Lang} for the latter point of view and also \cite{BSWoperad} for a refinement using operad theory.
Descent conditions which are based on Fredenhagen's universal algebra are 
however conceptually very different to cover-type
descent conditions: Instead of reconstructing the global datum of an AQFT on $M$
from local data in \textit{any choice} of cover $\{U_i\subseteq M\}$, one must exhaust $M$ 
by all possible embeddings $D\to M$ of contractibles and glue together the local data on all these subsets.
This means that these descent conditions are less flexible and  practical, and hence also less powerful, than 
cover-type descent conditions.
\sk

In our framework there are two different layers of cover-type 
descent conditions, which are related to each other, as we shall explain below. 
The more abstract top layer is to ask whether or not the $2$-functor
$\HK : \Loc^\op\to \CAT$ satisfies the descent conditions of a 
\textit{stack of categories}, which means whether or not the 
category $\HK(M)$ of global Haag-Kastler-style AQFTs on $M$ can be recovered 
from the categories of local Haag-Kastler-style AQFTs on the regions
of a causally
convex open cover $\{U_i\subseteq M\}$. See Definition \ref{def:stack} for the 
precise definition of a stack. We will show in Propositions \ref{prop:HKnotastack}, 
\ref{prop:HKtimeslicenotastack}, \ref{prop:RCHKnotastack}
and \ref{prop:RCHKtimeslicenotastack} that these descent conditions are \textit{not} 
automatically satisfied by the Haag-Kastler $2$-functor $\HK : \Loc^\op\to \CAT$
and all of its variations arising from relative compactness and/or the time-slice axiom.
Using more sophisticated technology from the theory of locally presentable
categories, we are able to pin down the origin of this violation
and provide an interpretation in terms of `bad objects' in the categories
$\HK(M)$ which do not have appropriately local behavior, 
see Theorems \ref{theo:precostack} and \ref{theo:Wprecostack}.
Discarding these `bad objects', we introduce improvements of the 
Haag-Kastler $2$-functors, see Definitions \ref{def:improvedHKpseudofunctor} 
and \ref{def:WimprovedHKpseudofunctor} for the details. The selection criterion 
(see Definitions \ref{def:improvedHK} and \ref{def:WimprovedHK}) 
for the full subcategories $\calHK(M)\subseteq \HK(M)$ of 
`good objects' which are assigned by the improved pseudo-functors 
can be understood as a second layer of descent conditions. The latter demand that
the selected `good' AQFTs over $M$
satisfy suitable cover-type descent conditions,
i.e.\ they are recovered by gluing their local data on any cover $\{U_i\subseteq M\}$.
Such a gluing construction for AQFTs,
which has been recently considered also in \cite{AB}, 
can be exploited also more constructively to build new `good' AQFTs on $M$ by gluing simpler `good' AQFTs
on the regions of a causally convex open cover $\{U_i\subseteq M\}$,
see Propositions \ref{prop:caldescentmapadjoint} and \ref{prop:Wcaldescentmapadjoint}. 
We will show in Theorems \ref{theo:RCHKstack} and \ref{theo:WRCHKstack}
that our improvement construction yields stacks $\calHK^{\mathrm{rc}}$
and $\calHK^{\mathrm{rc},W}$ of relatively compact Haag-Kastler-style AQFTs,
with or without the time-slice axiom, but we currently do not know if similar
results hold true without relative compactness. The question of
whether or not our AQFT-inspired `descent improvement' construction
realizes an abstract stackification, defined through a universal property, is currently
open, see Remark \ref{rem:stackification} for more detailed comments.
In Subsection \ref{subsec:KGexample}, we verify that the usual examples of free 
(i.e.\ non-interacting) AQFTs, such as the Klein-Gordon quantum field, satisfy our descent conditions,
hence they define points of the stacks $\calHK^{\mathrm{rc}}$ and $\calHK^{\mathrm{rc},W}$.
This is in stark contrast to the more elementary cosheaf-type descent conditions
discussed above, which are violated for such examples, see \cite[Appendix A]{BSreview}.
\sk

The framework and results of this paper suggest various avenues for future research.
On the one hand, the focus of our present paper is on the
case where the collection of all AQFTs over $M$ assembles into a $1$-category $\HK(M)$,
which is however inadequate for gauge theories as these are 
described by higher-categorical AQFTs which assemble into an $\infty$-category, see
e.g.\ \cite{BSWhomotopy,BPSW,Yau,Carmona} and also the reviews \cite{BSreview,BSchapter}.
It would be interesting to generalize our framework and results to this $\infty$-categorical
context and explore if the resulting higher-categorical descent conditions
are satisfied by examples of free quantum gauge theories as in \cite{BBS,BMS}.
The first steps of these developments, which correspond
to an $\infty$-categorical generalization of the main content of Section \ref{sec:HK2functor}, 
have been carried out recently in \cite{tPFAAQFTinfty}.
On the other hand, the key idea to consider $2$-functors $\HK : \Loc^\op\to \CAT$ 
which assign to each Lorentzian manifold $M\in\Loc$ the category of all AQFTs over $M$ is transferable
to other axiomatizations of quantum field theory, in particular to 
(pre)factorization algebras \cite{CG1,CG2} which are usually described
on the open subsets of a fixed manifold $M$,
similar in perspective to Haag-Kastler-style AQFTs.
It would be interesting to understand if the $2$-functor which assigns
to each manifold $M$ (potentially endowed with geometry) its category of prefactorization algebras over $M$
can be improved to a stack by adapting our constructions from Subsection \ref{subsec:HKcostackification}.
This would provide an alternative to the Weiss cover descent conditions from \cite{CG1,CG2}.
It is worthwhile to mention that locally constant prefactorization algebras, which describe
topological quantum field theories, automatically assemble into an $\infty$-stack \cite{Matsuoka,Scheimbauer}. 
However, we expect this to be a special feature of the topological nature 
of such theories, and that prefactorization algebras over manifolds endowed with 
(Riemannian, complex or Lorentzian) geometry will require an improvement construction
as in Subsection \ref{subsec:HKcostackification}.
\sk

The outline of the remainder of this paper is as follows:
In Section \ref{sec:prelim} we recall some preliminaries
which are needed to state and prove the results of the present paper.
Subsection \ref{subsec:AQFT} covers some basic aspects of 
the theory of orthogonal categories and AQFTs. More details can be
found in \cite{BSWoperad}, see also the reviews \cite{BSreview} and \cite{BSchapter}.
Subsection \ref{subsec:stacks} gives a brief introduction to pseudo-functors
and stacks, focusing mainly on the classes of examples appearing in our work. 
Subsection \ref{subsec:Pr} recalls some well-known aspects
of the theory of locally presentable categories which are needed for 
developing our Haag-Kastler stacks in the last Section \ref{sec:HKstack}.
Readers who are mostly 
interested in our more elementary Haag-Kastler $2$-functors from Section \ref{sec:HK2functor}
can skip this  subsection.
In Section \ref{sec:HK2functor} we develop the concept of Haag-Kastler $2$-functors
and prove the results announced in the paragraphs above, except
the ones concerning descent which will be presented afterwards 
in Section \ref{sec:HKstack}. This section is split
into Subsections \ref{subsec:HK2functorCOpen} and \ref{subsec:HK2functorRC} which cover
separately the case of Haag-Kastler-style AQFTs modeled over all causally convex opens and
the case modeled over all relatively compact causally convex opens. In Section \ref{sec:HKstack}
we study descent conditions of the Haag-Kastler $2$-functors from Section \ref{sec:HK2functor}.
Subsection \ref{subsec:HKnotations} starts with some general observations and constructions
to streamline the presentation. Subsection \ref{subsec:HKcostack}
uses techniques from the theory of locally presentable categories
to prove that suitable adjoints of our Haag-Kastler $2$-functors automatically
satisfy the weaker codescent conditions of a precostack, see Theorems \ref{theo:precostack}
and \ref{theo:Wprecostack} for the main results. Subsection \ref{subsec:HKcostackification}
develops our improvement construction for the Haag-Kastler $2$-functors
which consists of selecting suitable full subcategories of AQFTs that
satisfy cover-type descent conditions, see Definitions \ref{def:improvedHK} and \ref{def:WimprovedHK}.
We then prove that, under certain additional hypotheses, this construction yields
stacks, see Theorems \ref{theo:HKstacks} and \ref{theo:WHKstacks}. It is shown
in Theorems \ref{theo:RCHKstack} and \ref{theo:WRCHKstack} that these additional
hypotheses hold true for the case of the relatively compact Haag-Kastler
$2$-functor, with or without the time-slice axiom, which yields stacks $\calHK^{\mathrm{rc}}$ 
and $\calHK^{\mathrm{rc},W}$. Subsection \ref{subsec:KGexample} verifies
our descent conditions for the typical examples of free (i.e.\ non-interacting) AQFTs,
leading to the main result (see Theorem \ref{theo:KGpoints}) 
of this subsection that the Klein-Gordon quantum field
defines a point in both stacks $\calHK^{\mathrm{rc}}$ and $\calHK^{\mathrm{rc},W}$.
This paper includes four appendices. Appendix \ref{app:operadicLKE} recalls
some aspects of operadic left Kan extensions which are needed for some of our proofs.
Appendix \ref{app:localization} develops explicit
models for the localizations at all Cauchy morphisms
of the categories of causally convex opens $\COpen(M)$ and of relatively compact causally convex opens 
$\RC(M)$ in any Lorentzian manifold $M\in\Loc$.
Appendix \ref{app:bilimit} provides the details for the computation of 
a bicategorical limit which is needed in Subsection \ref{subsec:HKcostackification}.
Appendix \ref{app:Lorentz_geometry_details} proves some results about
Cauchy development stable covers which are needed for the proof of Theorem \ref{theo:WRCHKstack}.

\paragraph{Remarks about AQFTs taking values in $C^\ast$-algebras:} In this paper
we consider by default AQFTs taking values in the category $\Alg_{\mathsf{uAs}}(\TT)$
of unital associative algebras in some cocomplete closed symmetric monoidal category $\TT$.
This includes the cases of plain algebras for $\TT=\Vec$ 
the category of vector spaces, Banach algebras for $\TT = \mathbf{Ban}$ the category Banach spaces,
and (complete) bornological algebras for $\TT=\mathbf{(c)Born}$ the category
of (complete) bornological vector spaces.
The main reason for this choice is that it makes available to us 
powerful operadic technology \cite{BSWoperad} to describe universal constructions for AQFTs, 
which we will use frequently in our constructions and proofs. It is worthwhile
to highlight that the category of $C^\ast$-algebras $C^\ast\Alg$, which is often used
to model analytic features of AQFTs, is \textit{not} of this type, hence
$C^\ast$-algebraic AQFTs are a priori \textit{not} covered by this paper.
\sk

Disentangling precisely the results of our paper which do or do not adapt 
to $C^\ast$-algebraic AQFTs is difficult, because we are freely using 
operadic technology whenever it is convenient to carry out constructions or proofs. 
However, we can offer two general remarks which we believe are valuable
for readers interested in $C^\ast$-algebraic AQFTs: 
1.)~The main result of Section \ref{sec:HK2functor} is Theorem \ref{theo:LCQFTvsHK}
where it is proven that locally covariant AQFTs can be described equivalently
in terms of points of the Haag-Kastler $2$-functor. This result carries over ad verbum
to $C^\ast$-algebraic AQFTs, see Remark \ref{rem:CastHKvLCQFT}.
2.)~The key technical hypothesis for our Haag-Kastler stack constructions
in Section \ref{sec:HKstack} is that the Haag-Kastler $2$-functors 
take values in the $2$-category $\Pr^R$ of locally presentable categories, 
right adjoint functors and natural transformations. This is also true
for $C^\ast$-algebraic AQFTs, see Remark \ref{rem:Castpresentability},
hence there are no fundamental obstructions which prevent the adaption
of these ideas to this case. The actual complications are more 
of a technical nature since in the $C^\ast$-algebraic context 
one loses a priori the practically very useful possibility 
to model the relevant left adjoint functors in terms of operadic left Kan extensions.
This means that further work is required to find out 
whether or not the results of Section \ref{sec:HKstack} generalize to $C^\ast$-algebraic AQFTs.

%%%%%%%%%%%%%%%%%%%%%%%%%%%%%%%%%%%%%%%%%%%%%%%%
%%%%%%%%%%%%%%%%%%%%%%%%%%%%%%%%%%%%%%%%%%%%%%%%

\section{\label{sec:prelim}Preliminaries} 
In this section we recollect some background material
which is needed to state and prove the results of this paper.
We try to be as concise as possible and provide the reader
with references in which additional information and more details can be found.

\subsection{\label{subsec:AQFT}Orthogonal categories and AQFTs}
Orthogonal categories are an abstraction of categories of spacetimes
which are endowed with a notion of independent pairs
of subspacetimes $f_1 : M_1\rightarrow N \leftarrow M_2 :f_2$.
The following definition originated in \cite{BSWoperad}, see also \cite{Grant-Stuart}
for subsequent developments.
\begin{defi}\label{def:OCat}
\begin{itemize}
\item[(a)] An \textit{orthogonal category} is a pair $\ovr{\CC}= (\CC,\perp_\CC^{})$
consisting of a small category $\CC$ and a subset
${\perp_\CC^{}} \subseteq \Mor \CC\,{}_{\mathsf{t}}{\times}_{\mathsf{t}}\Mor \CC $
(called orthogonality relation) of the set of pairs of morphisms to a common target, 
which satisfies the following conditions:
\begin{itemize}
\item[(i)] Symmetry: $(f_2,f_1)\in{\perp_\CC^{}}$ for all $(f_1,f_2)\in{\perp_\CC^{}}$.
\item[(ii)] Composition stability: $(g\,f_1\,h_1, g\,f_2\,h_2)\in{\perp_{\CC}^{}}$
for all $(f_1,f_2)\in{\perp_\CC^{}}$ and all composable $\CC$-morphisms $g,h_1,h_2$.
\end{itemize}
We often write $f_1\perp_\CC^{} f_2$ to denote orthogonal pairs $(f_1,f_2)\in{\perp_\CC^{}}$.

\item[(b)] An \textit{orthogonal functor} $F : \ovr{\CC}\to \ovr{\DD}$
is a functor $F : \CC\to\DD$ between the underlying categories which
preserves orthogonal pairs, i.e.\ $F(f_1)\perp_\DD^{}F(f_2)$ for all $f_1\perp_\CC^{} f_2$.

\item[(c)] We denote by $\Cat^\perp$ the $2$-category whose objects
are all orthogonal categories, $1$-morphisms are all
orthogonal functors and $2$-morphisms are all 
natural transformations between orthogonal functors.

\end{itemize}
\end{defi}

\begin{ex}\label{ex:OCat}
The following orthogonal categories and functors are pivotal for our work:
\begin{itemize}
\item[(1)] Denote by $\Loc$ the category whose objects are oriented, time-oriented
and globally hyperbolic Lorentzian manifolds $M$ of a fixed 
dimension\footnote{More pedantically, one should write $\Loc_m$
to make explicit the dimension $m$ of the manifolds.
Since $m$ will be fixed but arbitrary throughout 
our whole paper, we ease notation by simply writing $\Loc$.}
$m\geq 1$ and whose morphisms $f:M\to N$ are orientation and time-orientation preserving isometric embeddings
with causally convex and open image $f(M)\subseteq N$. The orthogonal category $\ovr{\Loc}$
is defined by equipping $\Loc$ with the following orthogonality relation:
$(f_1:M_1\to N)\perp(f_2:M_2\to N)$ if and only if the images $f_1(M_1)\subseteq N$ and
$f_2(M_2)\subseteq N$ are causally disjoint in $N$. This orthogonal category features in locally covariant
AQFT \cite{BFV,FewsterVerch,BSWoperad}.

\item[(2)] Choose any oriented, time-oriented and globally hyperbolic Lorentzian
manifold $M\in\Loc$. Denote by $\COpen(M)$ the category whose objects are
all non-empty causally convex open subsets $U\subseteq M$ and whose morphisms $U\to V$ are
subset inclusions $U\subseteq V$. The orthogonal category $\ovr{\COpen(M)}$
is defined by equipping $\COpen(M)$ with the following orthogonality relation:
$(U_1\subseteq V)\perp (U_2\subseteq V)$ if and only if $U_1$ and $U_2$ are causally disjoint 
in $V$, or equivalently in $M$. Restricting to causally convex opens 
$U\subseteq M$ that are relatively compact, i.e.\ the closure $\mathrm{cl}(U)\subseteq M$
is a compact subset of $M$, defines a full orthogonal subcategory which we denote by
$\ovr{\RC(M)}\subseteq \ovr{\COpen(M)}$. These 
orthogonal categories feature in Haag-Kastler-style AQFT \cite{HaagKastler} on a fixed $M\in\Loc$,
where the relative compactness condition generalizes the concept of bounded regions in Minkowski spacetime.

\item[(3)] Consider the functor $k_M : \COpen(M)\to \Loc$ which is given by assigning to 
an object $U\subseteq M$ the object $U\in\Loc$ (with orientation, time-orientation
and metric induced by restricting those of $M\in\Loc$) and to a morphism
$U\subseteq V$ the canonical inclusion morphism $\iota_U^V : U\to V$ in $\Loc$.
This defines an orthogonal functor $k_M : \ovr{\COpen(M)}\to \ovr{\Loc}$ with respect to 
the orthogonality relations defined in items (1) and (2) above. The restriction
to relatively compact subsets defines an orthogonal functor which we denote with a slight
abuse of notation by the same symbol $k_M : \ovr{\RC(M)}\to \ovr{\Loc}$. \qedhere
\end{itemize}
\end{ex}

Associated to every orthogonal category $\ovr{\CC}$ is a concept of
AQFTs over $\ovr{\CC}$. These admit a concise and powerful 
description in terms of algebras over the AQFT operads $\O_{\ovr{\CC}}$
from \cite{BSWoperad}, see also Definition \ref{def:AQFToperad} and
\cite{BSreview,BSchapter} for reviews. In particular, this operadic perspective is crucial
to prove the key results in Propositions \ref{propo:operadicLKE} and \ref{propo:Ffullorthogonal}
below. In order to simplify
the presentation of our present paper, we will not recall this operadic
approach to AQFT and we provide instead an equivalent, but more elementary definition.
For this we fix any cocomplete closed symmetric monoidal category $\TT$
and denote by $\Alg_\mathsf{uAs}(\TT)$ the category of unital associative algebras 
in $\TT$.\footnote{In applications, one often chooses $\TT=\Vec_\bbK$ to be 
the cocomplete closed symmetric monoidal category of vector spaces over a field $\bbK$. 
In this case $\Alg_\mathsf{uAs}(\TT)$ is the usual category of unital associative
$\bbK$-algebras.}
\begin{defi}\label{def:AQFT}
Let $\ovr{\CC}$ be an orthogonal category and $\TT$ a cocomplete
closed symmetric monoidal category. The category of \textit{$\TT$-valued AQFTs over $\ovr{\CC}$}
is defined as the full subcategory
\begin{flalign}
\AQFT(\ovr{\CC})\,\subseteq\,\Fun\big(\CC,\Alg_\mathsf{uAs}(\TT)\big)
\end{flalign}
consisting of all functors $\AAA : \CC\to \Alg_\mathsf{uAs}(\TT)$ which satisfy
the following $\perp$-commutativity axiom: For every orthogonal pair 
$(f_1:M_1\to N)\perp_\CC^{}(f_2:M_2\to N)$, the diagram
\begin{flalign}
\begin{gathered}
\xymatrix@C=5em{
\AAA(M_1)\otimes \AAA(M_2) \ar[r]^-{\AAA(f_1)\otimes \AAA(f_2)}\ar[d]_-{\AAA(f_1)\otimes \AAA(f_2)}~&~ \AAA(N)\otimes\AAA(N)\ar[d]^-{\mu_N} \\
\AAA(N)\otimes \AAA(N)\ar[r]_-{\mu_N^\op}~&~\AAA(N)
}
\end{gathered}
\end{flalign}
in $\TT$ commutes, where $\mu_N^{(\op)}$ denotes the (opposite)
multiplication in the algebra $\AAA(N)$.
\end{defi}

\begin{rem}
For the orthogonal categories from Example \ref{ex:OCat},
the $\perp$-commutativity axiom gives precisely the Einstein causality axiom
of locally covariant \cite{BFV,FewsterVerch} or Haag-Kastler-style \cite{HaagKastler} AQFTs.
The implementation of the time-slice axiom will be explained in Proposition \ref{prop:timeslice}
and Remark \ref{rem:timeslice} below.
\end{rem}

Given any orthogonal functor $F : \ovr{\CC}\to\ovr{\DD}$, there 
exists an associated pullback functor
\begin{flalign}\label{eqn:Fpullbackfunctor}
F^\ast\,:\, \AQFT(\ovr{\DD})~&\longrightarrow~\AQFT(\ovr{\CC})
\end{flalign}
between the corresponding AQFT categories. Explicitly, this pullback functor 
acts on objects $\AAA\in \AQFT(\ovr{\DD})$ by precomposition  $F^\ast(\AAA):= \AAA\,F$
of the underlying functor, which defines an object in 
$\AQFT(\ovr{\CC})$ since orthogonal functors preserve orthogonal pairs and hence the $\perp$-commutativity axiom.
On morphisms $\zeta : \AAA\Rightarrow \BBB$ in $\AQFT(\ovr{\DD})$, which are natural transformations
between the underlying functors, the pullback functor 
is given by whiskering $F^\ast (\zeta):= \zeta\,F : \AAA\,F\Rightarrow\BBB\,F$.
The following result is a non-trivial
consequence of the operadic description of AQFTs,
see e.g.\ \cite[Theorem 2.11]{BSWoperad} for a proof.
\begin{propo}\label{propo:operadicLKE}
For every orthogonal functor $F :\ovr{\CC}\to \ovr{\DD}$, the pullback 
functor in \eqref{eqn:Fpullbackfunctor} admits a left adjoint, i.e.\ one has an adjunction
\begin{flalign}\label{eqn:Fadjunction}
\xymatrix{
F_!\,:\, \AQFT(\ovr{\CC}) \ar@<0.75ex>[r]~&~\ar@<0.75ex>[l] \AQFT(\ovr{\DD})\,:\,F^\ast
}
\end{flalign}
between the corresponding AQFT categories. The left adjoint $F_!$ is 
called operadic left Kan extension along $F$, see also Appendix \ref{app:operadicLKE}
for an explicit model.
\end{propo}

This result can be strengthened in specific cases where the orthogonal functor 
has additional properties, see \cite[Sections 4.2--4.4]{BSWoperad}. For our present
work, the following strengthening will be crucial.
\begin{propo}\label{propo:Ffullorthogonal}
Suppose that the orthogonal functor $F : \ovr{\CC} \to \ovr{\DD}$
is fully faithful and reflects orthogonality, i.e.\
$F(f_1)\perp_{\DD}^{}F(f_2)$ if and only if $f_1\perp_\CC^{}f_2$.
Then the adjunction \eqref{eqn:Fadjunction} exhibits $\AQFT(\ovr{\CC})$
as a coreflective full subcategory of $\AQFT(\ovr{\DD})$, i.e.\ the left adjoint
$F_!$ is fully faithful or, equivalently, the adjunction unit
$\eta : \id_{\AQFT(\ovr{\CC})} \Rightarrow F^\ast\,F_!$ is a natural isomorphism.
\end{propo}

\begin{ex}\label{ex:RCinCOpen}
For every object $M\in\Loc$, the full orthogonal subcategory 
inclusion $i_M : \ovr{\RC(M)}\to \ovr{\COpen(M)}$
from item (2) of Example \ref{ex:OCat}
satisfies the hypotheses of Proposition \ref{propo:Ffullorthogonal}.
Hence, we obtain an adjunction
\begin{flalign}\label{eqn:iMadjunction}
\xymatrix{
i_{M\,!}\,:\, \AQFT(\ovr{\RC(M)}) \ar@<0.75ex>[r]~&~\ar@<0.75ex>[l] \AQFT(\ovr{\COpen(M)})\,:\,i_M^\ast
}
\end{flalign}
which exhibits $\AQFT(\ovr{\RC(M)})$ as a coreflective full subcategory
of $\AQFT(\ovr{\COpen(M)})$. This adjunction restricts to an adjoint equivalence
\begin{flalign}
\xymatrix{
i_{M\,!}\,:\, \AQFT(\ovr{\RC(M)}) \ar@<0.75ex>[r]_-{\sim}~&~\ar@<0.75ex>[l] \AQFT(\ovr{\COpen(M)})^{\epsilon\mathrm{-iso}}\,:\,i_M^\ast
}
\end{flalign}
between the category $\AQFT(\ovr{\RC(M)})$ and the full subcategory 
$\AQFT(\ovr{\COpen(M)})^{\epsilon\mathrm{-iso}}\subseteq\AQFT(\ovr{\COpen(M)})$ consisting of all objects
$\AAA\in \AQFT(\ovr{\COpen(M)})$ for which the 
counit $\epsilon_\AAA : i_{M\,!}i_M^\ast(\AAA)\stackrel{\cong}{\Longrightarrow}\AAA$
is an isomorphism. 
\sk

One can characterize the latter property very explicitly
by observing that $i_M$ satisfies the $j$-closedness
property from \cite[Definition 5.3]{BSWoperad}, which by
\cite[Corollary 5.5]{BSWoperad} implies that 
the left adjoint $i_{M\,!}$ can be modeled by 
a categorical (in contrast to operadic) left Kan extension.
Via the usual colimit formula for categorical left Kan extensions, see e.g.\ \cite[Chapter 6.2]{Riehl},
we have explicitly that, for every $\BBB\in  \AQFT(\ovr{\RC(M)})$, 
\begin{flalign}
i_{M\,!}(\BBB)(U)\,=\,\colim\Big(\RC(M)/U \longrightarrow \RC(M)\stackrel{\BBB}{\longrightarrow}
\Alg_{\mathsf{uAs}}(\TT)\Big)\quad,
\end{flalign}
for all $U\in \COpen(M)$, where the comma category $\RC(M)/U$ describes all relatively compact 
causally convex opens in $M$ which are also contained in $U$.
The component of the counit $\epsilon_\AAA$ at $U\in\COpen(M)$
is then given by the canonical map
\begin{flalign}\label{eqn:epsilonadd}
(\epsilon_\AAA)_U\,:\,\colim\Big(
\RC(M)/U \longrightarrow \COpen(M) \stackrel{\AAA}{\longrightarrow} \Alg_{\mathsf{uAs}}(\TT)
\Big) ~\longrightarrow~\AAA(U)\quad.
\end{flalign}
This implies that $\epsilon_\AAA$ is an isomorphism 
if and only if the value of $\AAA\in \AQFT(\ovr{\COpen(M)})$ 
on any causally convex open $U\in\COpen(M)$ can be 
recovered via the colimit \eqref{eqn:epsilonadd} from the 
restriction of $\AAA$ to the comma category $\RC(M)/U$,
i.e.\ to open subsets of $U$ that are relatively compact and causally convex with respect to $M$.
\sk

The property that $\epsilon_\AAA$ is an isomorphism is therefore a particular
kind of \textit{additivity} property on $\AAA \in \AQFT(\ovr{\COpen(M)})$.
In Subsection \ref{subsubsec:additivity}, we will compare this 
concept to an alternative additivity property used in locally covariant settings \cite[Definition 2.16]{BPS}.
\end{ex}

To conclude this subsection, we shall briefly explain how the time-slice
axiom of AQFT can be encoded in the framework presented above.
The key tool is given by the concept of localizations of orthogonal categories,
see \cite{BCStimeslice} for the technical details.
This is similar to localizations of ordinary (non-orthogonal) categories,
which allow one to universally ``add inverses'' to a chosen collection of morphisms in a category.
\begin{propo}\label{prop:timeslice}
Let $\ovr{\CC}$ be an orthogonal category and $W\subseteq \Mor\CC$
a subset of the set of morphisms in $\CC$. Then an orthogonal localization
functor $L : \ovr{\CC}\to \ovr{\CC}[W^{-1}]$ exists and it 
induces via pullback \eqref{eqn:Fpullbackfunctor} an equivalence
\begin{flalign}
L^\ast\,:\, \AQFT\big(\ovr{\CC}[W^{-1}]\big) ~\stackrel{\simeq}{\longrightarrow} ~\AQFT(\ovr{\CC})^W
\end{flalign}
between the full subcategory $\AQFT(\ovr{\CC})^W\subseteq \AQFT(\ovr{\CC})$ consisting 
of all AQFTs over $\ovr{\CC}$ which send all $W$-morphisms to isomorphisms and the category 
$\AQFT\big(\ovr{\CC}[W^{-1}]\big)$ of AQFTs over the localized orthogonal category $\ovr{\CC}[W^{-1}]$.
\end{propo}

\begin{rem}\label{rem:timeslice}
In the context of Example \ref{ex:OCat}, one chooses $W$ to be the set of all Cauchy morphisms
in, respectively, $\Loc$, $\COpen(M)$ or $\RC(M)$, 
i.e.\ morphisms with image containing a Cauchy surface of the codomain.
The property of an AQFT sending all $W$-morphisms to isomorphisms is then precisely
the time-slice axiom of locally covariant or, respectively, Haag-Kastler-style AQFTs.
Proposition \ref{prop:timeslice} implies that the time-slice axiom can be implemented
either as an additional property or, equivalently, as 
a structure by replacing the orthogonal categories $\ovr{\Loc}$, $\ovr{\COpen(M)}$ and $\ovr{\RC(M)}$
by their orthogonal localizations at all Cauchy morphisms.
\end{rem}

\begin{ex}\label{ex:localizations}
We present explicit models for the orthogonal localizations $\ovr{\COpen(M)}[W_M^{-1}]$ 
and $\ovr{\RC(M)}[W_{\mathrm{rc},M}^{-1}]$ of the orthogonal 
categories $\ovr{\COpen(M)}$ and $\ovr{\RC(M)}$ at all Cauchy morphisms. 
These models are obtained from a calculus of fractions and details
are explained in Appendix \ref{app:localization}.
\begin{itemize}
\item[(1)] The orthogonal category $\ovr{\COpen(M)}[W_M^{-1}]$ has as objects
all non-empty causally convex opens $U\subseteq M$ and there exists
at most one morphism $U \to U^\prime$ for any objects $U,U^\prime$.
The morphism $U\to U^\prime$ exists if and only if $U\subseteq D_M(U^\prime)$
is contained in the Cauchy development
of $U^\prime\subseteq M$ in $M$. Two morphisms are orthogonal
$(U_1\to U^\prime)\perp (U_2\to U^\prime)$
if and only if $(U_1\subseteq M)\perp (U_2\subseteq M)$ are causally disjoint in $M$.
The orthogonal localization functor $L_M : \ovr{\COpen(M)} \to \ovr{\COpen(M)}[W_M^{-1}]$
acts as the identity on objects and it sends a subset inclusion $U\subseteq U^\prime$
to the unique morphism $U\to U^\prime$ in the localized orthogonal category.

\item[(2)] The orthogonal category $\ovr{\RC(M)}[W_{\mathrm{rc},M}^{-1}]$ has as objects
all non-empty relatively compact causally convex opens $U\subseteq M$ and there exists
at most one morphism $U \to U^\prime$ for any objects $U,U^\prime$.
The morphism $U\to U^\prime$ exists if and only if $U\subseteq D_M(U^\prime)$
is contained in the Cauchy development
of $U^\prime\subseteq M$ in $M$. Two morphisms are orthogonal
$(U_1\to U^\prime)\perp (U_2\to U^\prime)$
if and only if $(U_1\subseteq M)\perp (U_2\subseteq M)$ are causally disjoint in $M$.
The orthogonal localization functor $L_{\mathrm{rc},M} : \ovr{\RC(M)} \to \ovr{\RC(M)}[W_{\mathrm{rc},M}^{-1}]$
acts as the identity on objects and it sends a subset inclusion $U\subseteq U^\prime$
to the unique morphism $U\to U^\prime$ in the localized orthogonal category.
\end{itemize}
We would like to note that an alternative but equivalent model 
for $\ovr{\COpen(M)}[W_M^{-1}]$ has been presented in \cite[Proposition 3.3]{BDSboundary}
in terms of a reflective orthogonal localization,
however this model does not generalize to the relatively compact case
$\ovr{\RC(M)}[W_{\mathrm{rc},M}^{-1}]$. In the present paper, we prefer
to work with our two models from above because they treat uniformly 
the non-relatively compact and the relatively compact case.
\end{ex}

\subsection{\label{subsec:stacks}Pseudo-functors and stacks}
We assume that the reader has some familiarity with
elementary $2$-categorical concepts, such as
(strict) $2$-categories, pseudo-functors, pseudo-natural 
transformations, modifications and bilimits. Complete definitions
and explanations of these concepts can be found for instance in the book \cite{Yau2Cats},
see also \cite{Lack} for a concise introduction. 
\sk

Since the particular case of pseudo-functors $X : \CC^\op \to \K$ from the opposite of a small $1$-category $\CC$ 
to a $2$-subcategory $\K\subseteq \CAT$ of the $2$-category $\CAT$ of (not necessarily small) categories, 
functors and natural transformations will appear very frequently in our work, we spell
this concept out in detail.
\begin{defi}\label{def:pseudofunctor}
Let $\CC$ be a small $1$-category and $\K\subseteq \CAT$ a $2$-subcategory. 
A \textit{pseudo-functor} $X : \CC^\op \to\K$ is given by the following data:
\begin{itemize}
\item[(1)] For each object $M\in\CC$, a category $X(M)$ in $\K$.

\item[(2)] For each morphism $f:M\to N$ in $\CC$, a functor
$X(f) : X(N)\to X(M)$ in $\K$.

\item[(3)] For each pair of composable morphisms $f : M\to N$
and $g: N\to O$ in $\CC$, a natural isomorphism
$X_{g,f} : X(f)\,X(g)\Rightarrow X(g\,f)$ in $\K$.

\item[(4)] For each object $M\in\CC$, a natural isomorphism
$X_M : \id_{X(M)}\Rightarrow X(\id_M)$ in $\K$.
\end{itemize}
These data have to satisfy the following axioms:
\begin{itemize}
\item[(i)] For all triples of composable morphisms 
$f : M\to N$, $g: N\to O$ and $h: O\to P$ in $\CC$, the diagram of natural transformations
\begin{flalign}
\begin{gathered}
\xymatrix@C=4em{
\ar@{=>}[d]_-{\Id \mathop{\ast} X_{h,g}}
X(f)\, X(g)\, X(h)
\ar@{=>}[r]^-{X_{g,f} \mathop{\ast} \Id}
~&~ X(g\,f)\,X(h)
\ar@{=>}[d]^-{X_{h,gf}}\\
X(f)\,X(h\,g)
\ar@{=>}[r]_-{X_{h g,f}}
~&~X(h\,g\,f)
}
\end{gathered}
\end{flalign}
commutes, where $\Id$ denotes the identity natural transformations and $\ast$ denotes
horizontal composition of natural transformations.

\item[(ii)] For all morphisms $f:M\to N$ in $\CC$, the two diagrams of natural transformations
\begin{flalign}
\begin{gathered}
\xymatrix@C=3em{
\ar@{=>}[d]_-{X_M\mathop{\ast} \Id}\id_{X(M)}\,X(f) \ar@{=}[rd]~&~ ~&~
\ar@{=>}[d]_-{\Id \mathop{\ast} X_{N}} X(f)\,\id_{X(N)} \ar@{=}[rd]~&~ \\
X(\id_M)\,X(f)\ar@{=>}[r]_-{X_{f,\id_M}}~&~ X(f\,\id_M)~&~
X(f)\,X(\id_{N})\ar@{=>}[r]_-{X_{\id_{N},f}}~&~ X(\id_{N}\,f)
}
\end{gathered}
\end{flalign}
commute.
\end{itemize}
To ease our notations, we will often suppress the coherence natural isomorphisms
$X_{g,f}$ and $X_M$ by simply writing $\cong$. 
\end{defi}

In the case where $\CC$ comes endowed with a Grothendieck topology, i.e.\ there exists a 
concept of coverings for objects $M\in\CC$, one can demand that the pseudo-functor $X :\CC^\op\to\K$ 
satisfies a suitable descent condition with respect to these coverings. This leads to the notion
of a \textit{stack}, see e.g.\ \cite{Vistoli} for the details. For the purpose of our work,
we do not have to introduce the concepts of Grothendieck topologies
and stacks in full generality. It suffices instead to discuss
the examples which will be relevant later. Let us recall from Example \ref{ex:OCat} 
the category $\Loc$ of oriented, time-oriented and globally hyperbolic Lorentzian manifolds 
of a fixed dimension $m\geq 1$. 
\begin{defi}\label{def:cover}
Given any object $M\in\Loc$, we say that a 
family of subsets $\U:=\{U_i\subseteq M\}$ is a \textit{causally convex open cover}  of $M$
if each $U_i\subseteq M$ is a non-empty causally 
convex open subset and $\bigcup_iU_i = M$. 
A causally convex open cover $\U=\{U_i\subseteq M\}$ is called \textit{$D$-stable}
if each $U_i$ coincides with its Cauchy development in $M$, i.e.\ $D_M(U_i)=U_i$ for all $i$.
\end{defi}

\begin{rem}\label{rem:inclusionLocmorphisms}
Note that each causally convex open cover $\U=\{U_i\subseteq M\}$ 
defines, through the canonical inclusion morphisms, a family of $\Loc$-morphisms $\{\iota_{U_i}^M : U_i\to M\}$.
Since intersections $U_{ij}:= U_i\cap U_j\subseteq M$ of causally convex open subsets 
are either causally convex open or empty, we further obtain canonical inclusion $\Loc$-morphisms
$\iota_{U_{ij}}^{U_i} : U_{ij}\to U_i$ and $\iota_{U_{ij}}^{U_j} : U_{ij}\to U_j$,
for all $i,j$ with $U_{ij}\neq \emptyset$.
A similar statement holds true for triple intersections $U_{ijk}:= U_i\cap U_j\cap U_k\subseteq M$,
which come with canonical inclusion $\Loc$-morphisms $\iota_{U_{ijk}}^{U_{ij}} : U_{ijk}\to U_{ij}$,
$\iota_{U_{ijk}}^{U_{ik}} : U_{ijk}\to U_{ik}$ and $\iota_{U_{ijk}}^{U_{jk}} : U_{ijk}\to U_{jk}$, 
for all $i,j,k$ with $U_{ijk}\neq \emptyset$.
Since $U_{ii} = U_i$, there are also canonical $\Loc$-morphisms $\id_{U_i} : U_i \to U_{ii}$.
\sk

Note that in the case where $\U = \{U_i \subseteq M\}$ is $D$-stable,
all intersections and triple intersections inherit $D$-stability.
Indeed, one observes that $U_{ij} \subseteq D_M(U_{ij}) \subseteq D_M(U_i) \cap D_M(U_j) = U_{ij}$
and similar for triple intersections.
\end{rem}

Let us now suppose that the $2$-subcategory $\K\subseteq \CAT$ admits all small bilimits.\footnote{In
the terminology of \cite{Yau2Cats}, our bilimits are pseudo bilimits.}
Given any pseudo-functor $X : \Loc^\op\to\K$, we can then define 
for each causally convex open cover $\U=\{U_i\subseteq M\}$ of an object $M\in\Loc$
the \textit{descent category}
\begin{flalign}\label{eqn:descentcat}
X(\U)\,:=\,\bilim\bigg(
\xymatrix{
\prod\limits_{i}X(U_i) \ar@<0.75ex>[r]\ar@<-0.75ex>[r]~&~\ar[l]\prod\limits_{ij}X(U_{ij})
\ar@<1.5ex>[r]\ar@<-1.5ex>[r] \ar[r]~&~\ar@<0.75ex>[l]\ar@<-0.75ex>[l]\prod\limits_{ijk}X(U_{ijk})
}\bigg)\,\in\,\K\quad,
\end{flalign}
where the second product runs over all non-empty intersections $U_{ij}\neq \emptyset$ and 
the third product runs over all non-empty triple intersections $U_{ijk}\neq \emptyset$.
The arrows in this diagram are given by applying the pseudo-functor
$X$ to the canonical $\Loc$-morphisms from Remark \ref{rem:inclusionLocmorphisms}.
Due to the universal property of bilimits, 
there exists a canonical functor
\begin{flalign}\label{eqn:functortodescentcat}
X(M) ~\longrightarrow~X(\U)
\end{flalign}
in $\K$ from $X(M)$ to the descent category, which is defined
by applying $X$ to the canonical inclusion morphisms $\iota_{U_i}^M : U_i\to M$.
\begin{defi}\label{def:stack}
Suppose that the $2$-subcategory $\K\subseteq \CAT$ admits all small bilimits.
A pseudo-functor $X : \Loc^\op\to \K$ is called a \textit{$\K$-valued stack} with respect
to the ($D$-stable) causally convex open Grothendieck topology on $\Loc$
if it satisfies the following descent conditions:
For every object $M\in\Loc$ and every ($D$-stable)
causally convex open cover $\U=\{U_i\subseteq M\}$,
the canonical functor $X(M) \to X(\U)$
of \eqref{eqn:functortodescentcat} is an equivalence in $\K$.
\end{defi}

\begin{rem}\label{rem:prestack}
Following the terminology of \cite[Definition 4.6]{Vistoli},
we call the pseudo-functor $X : \Loc^\op\to \K$ a \textit{prestack}
whenever the functor $X(M) \to X(\U)$
of \eqref{eqn:functortodescentcat} is fully faithful, for
every object $M\in\Loc$ and every ($D$-stable)
causally convex open cover $\U=\{U_i\subseteq M\}$.
Note that a prestack is a categorified variant of a 
separated presheaf, and not of an arbitrary presheaf, 
which makes this terminology slightly ambiguous.
\end{rem}

\begin{rem}
Note that a stack $X$ formalizes a 
local-to-global property since its `global'
value $X(M)$ on an object $M\in\Loc$ is determined (up to equivalence) 
via \eqref{eqn:functortodescentcat} from its `local' values \eqref{eqn:descentcat} 
on a cover $\U=\{U_i\subseteq M\}$. A prestack formalizes a weaker local-to-global
property for the morphisms in the categories assigned by $X$, but not necessarily for the objects.
\end{rem}

\begin{rem}\label{rem:descentexplicit}
In the case where $\K=\CAT$ is the $2$-category of all categories, functors and natural transformations,
the bilimit \eqref{eqn:descentcat} which defines the descent category $X(\U)$ 
for a cover $\U=\{U_i\subseteq M\}$ can be computed explicitly. This yields the following concrete description:
\begin{itemize}
\item An object in $X(\U)$ is a tuple $(\{x_i\},\{\varphi_{ij}\})$ 
consisting of a family of objects $x_i\in X(U_i)$, for all $i$, and a family of 
isomorphisms $\varphi_{ij} : x_j\vert_{U_{ij}}\to x_{i}\vert_{U_{ij}}$ in $X(U_{ij})$,
for all $i,j$ with $U_{ij}\neq \emptyset$,
where $x_j\vert_{U_{ij}}:=X\big(\iota_{U_{ij}}^{U_j}\big)(x_j) \in X(U_{ij})$ and 
$x_{i}\vert_{U_{ij}}:= X\big(\iota_{U_{ij}}^{U_i}\big)(x_i)\in X(U_{ij})$ are convenient short-hand
notations. These data have to satisfy the cocycle conditions
\begin{flalign}
\begin{gathered}
\xymatrix@C=1.5em@R=1.5em{
~&~ x_j\vert_{U_{jk}}\vert_{U_{ijk}} \ar[dr]^-{\cong}~&~\\
\ar[d]_-{\cong} x_k\vert_{U_{jk}}\vert_{U_{ijk}}\ar[ru]^-{\varphi_{jk}\vert_{U_{ijk}}}~&~~&~ x_j\vert_{U_{ij}}\vert_{U_{ijk}} \ar[d]^-{\varphi_{ij}\vert_{U_{ijk}}}\\
\ar[dr]_-{\varphi_{ik}\vert_{U_{ijk}}}x_k\vert_{U_{ik}}\vert_{U_{ijk}}~&~~&~ x_i\vert_{U_{ij}}\vert_{U_{ijk}}\\
~&~x_i\vert_{U_{ik}}\vert_{U_{ijk}}\ar[ru]_-{\cong}~&~
}
\end{gathered}
\qquad 
\begin{gathered}
\xymatrix@C=2em@R=2em{
x_i\vert_{U_{ii}} \ar[r]^-{\varphi_{ii}}~&~ x_i\vert_{U_{ii}}\\
\ar[u]^-{\cong} x_i \ar[r]_-{\id_{x_i}}~&~ x_i\ar[u]_-{\cong}
}
\end{gathered}\qquad,
\end{flalign}
for all $i,j,k$ with $U_{ijk}\neq \emptyset$. The arrows labeled by $\cong$ are given by the 
coherence isomorphisms of the pseudo-functor $X$.

\item A morphism $\{\psi_i\} : (\{x_i\},\{\varphi_{ij}\}) 
\to (\{x^\prime_i\},\{\varphi^\prime_{ij}\}) $ in $X(\U)$ is a family 
of morphisms $\psi_i : x_i\to x_i^\prime$ in $X(U_i)$, for all $i$,
which is compatible with the cocyles according to
\begin{flalign}
\begin{gathered}
\xymatrix@R=2em@C=3em{
\ar[d]_-{\varphi_{ij}}x_j\vert_{U_{ij}} \ar[r]^-{\psi_j\vert_{U_{ij}}}~&~ x_j^\prime\vert_{U_{ij}} \ar[d]^-{\varphi^\prime_{ij}}\\
x_i\vert_{U_{ij}}\ar[r]_-{\psi_i\vert_{U_{ij}}}~&~x^\prime_i\vert_{U_{ij}}
}
\end{gathered}\qquad,
\end{flalign}
for all $i,j$ with $U_{ij}\neq \emptyset$.\qedhere
\end{itemize}
\end{rem}

\subsection{\label{subsec:Pr}Locally presentable categories}
We will briefly recall some well-known aspects of the theory of locally presentable
categories which will become essential when we construct the
Haag-Kastler stacks in Section \ref{sec:HKstack}. Readers who are mainly 
interested in our discussion of the more elementary
Haag-Kastler $2$-functors in Section \ref{sec:HK2functor} can skip 
this technical subsection. 
\sk

A \textit{locally presentable category} $\EE$ is a special kind of 
category which satisfies the following technical
conditions: 1.)~It is cocomplete, i.e.\ all small colimits exist in $\EE$,
and 2.)~it is generated under $\lambda$-directed colimits
from a subset $\Gamma\subseteq \mathrm{Obj}(\EE)$ of $\lambda$-presentable objects,
for $\lambda$ some regular cardinal. We refer the reader to \cite{Adamek}
for an extensive description of the rich theory of locally presentable categories
and also to \cite[Chapter 5]{Borceux} for a more concise introduction.
\begin{ex}\label{ex:locprescats}
The following examples of locally presentable categories are relevant in the context of AQFT.
\begin{itemize}
\item[(1)] The category $\Vec_\bbK$ of vector spaces over a field $\bbK$ is locally presentable for 
$\lambda=\aleph_0$. Indeed, $\Vec_\bbK$ is clearly cocomplete and the $\aleph_0$-presentable objects are the 
finite-dimensional vector spaces. Each vector space $V\in\Vec_\bbK$ is an $\aleph_0$-directed colimit
over its finite-dimensional subspaces.

\item[(2)] For every small category $\CC$ and every locally presentable category $\EE$, the category 
$\Fun(\CC,\EE)$ of functors and natural transformations is locally presentable, see \cite[Corollary 1.54]{Adamek}.
As a special case, the product category $\EE^S := \prod_{s\in S} \EE 
\cong \Fun(S,\EE)$ corresponding to a set $S$, which we also
regard as a category with only identity morphisms, is locally presentable whenever $\EE$ is.

\item[(3)] Let $\EE$ be a locally presentable category which is endowed with a closed symmetric monoidal structure.
Given any small colored operad $\O$, i.e.\ its class of objects $\O_0$ is a set, the category
$\Alg_\O(\EE)$ of $\O$-algebras in $\EE$ is locally presentable. Indeed, the category of $\O$-algebras
is equivalent to the category of algebras over the monad $\O\circ (-) : \EE^{\O_0}\to\EE^{\O_0}$, see e.g.\
\cite[Chapter 4.5]{Yau} for details,  which is locally presentable as a consequence of item (2) and 
\cite[Theorem 5.5.9]{Borceux}.

\item[(4)] As a special case of item (3), we observe that, for each orthogonal category $\ovr{\CC}$,
the AQFT category $\AQFT(\ovr{\CC})\cong \Alg_{\O_{\ovr{\CC}}}(\TT)$ from Definition \ref{def:AQFT} 
is locally presentable whenever the target category $\TT$ is. 
By item (1), this is in particular the case for the standard choice $\TT=\Vec_\bbK$. \qedhere
\end{itemize}
\end{ex}

Locally presentable categories assemble into the following two 
interesting $2$-subcategories of $\CAT$, both of which will be important for our work.
\begin{defi}\label{def:PrL/R} 
\begin{itemize}
\item[(a)] We denote by $\Pr^L\subseteq \CAT$ the $2$-subcategory whose objects are
all locally presentable categories, $1$-morphisms are all \textit{left} adjoint functors 
between locally presentable categories and 
$2$-morphisms are all natural transformations between left adjoint functors. 

\item[(b)] We denote by $\Pr^R\subseteq\CAT$ the $2$-subcategory whose objects are
all locally presentable categories, $1$-morphisms are all \textit{right} adjoint functors 
between locally presentable categories and 
$2$-morphisms are all natural transformations between right adjoint functors.
\end{itemize}
\end{defi}

The two $2$-categories $\Pr^{L}$ and $\Pr^{R}$ from Definition \ref{def:PrL/R} can be 
related by the following construction. Let us denote by $(\Pr^{R})^{\mathrm{coop}}$ 
the $2$-category which is obtained by reversing the direction of all $1$-morphisms 
and of all $2$-morphisms in $\Pr^{R}$. Consider the pseudo-functor
\begin{subequations}\label{eqn:LtoR}
\begin{flalign}
(-)^\dagger\,:\, \Pr^L~\longrightarrow~(\Pr^R)^{\mathrm{coop}}
\end{flalign}
which acts on objects as the identity, on $1$-morphisms
$F : \EE\to\FF$ by taking right adjoints $F^\dagger : \FF\to \EE$,
and on $2$-morphisms $\zeta : F\Rightarrow G$ via
\begin{flalign}
\zeta^\dagger\,:\,\xymatrix@C=3.5em{
G^\dagger \ar@{=>}[r]^-{\eta_F\,G^\dagger} ~&~ 
F^\dagger\,F\,G^\dagger \ar@{=>}[r]^-{F^\dagger \,\zeta\,G^\dagger} ~&~ 
F^\dagger\,G\,G^\dagger\ar@{=>}[r]^-{F^\dagger\,\epsilon_G}~&~ 
F^\dagger
}\quad,
\end{flalign}
\end{subequations}
where $\eta_F$ denotes the unit of the adjunction $F\dashv F^\dagger$ and $\epsilon_G$ the counit
of the adjunction $G\dashv G^\dagger$.
\begin{lem}\label{lem:PrLvsPrR}
The pseudo-functor \eqref{eqn:LtoR} exhibits a biequivalence 
\begin{flalign}
\Pr^L \,\simeq \,(\Pr^R)^{\mathrm{coop}}\quad.
\end{flalign}
\end{lem}
\begin{rem}
An explicit quasi-inverse for the pseudo-functor \eqref{eqn:LtoR} is given by applying
${}^{\mathrm{coop}}$ to the pseudo-functor (denoted with abuse of notation by the same symbol 
as \eqref{eqn:LtoR})
\begin{subequations}\label{eqn:RtoL}
\begin{flalign}
(-)^\dagger\,:\,\Pr^R ~\longrightarrow~(\Pr^L)^{\mathrm{coop}}
\end{flalign}
which acts on objects as the identity, on $1$-morphisms
$F : \EE\to\FF$ by taking left adjoints $F^\dagger : \FF\to \EE$,
and on $2$-morphisms $\zeta : F\Rightarrow G$ via
\begin{flalign}
\zeta^\dagger \,:\,\xymatrix@C=3.5em{
G^\dagger \ar@{=>}[r]^-{G^\dagger\,\eta_F} ~&~ 
G^\dagger \, F\,F^\dagger \ar@{=>}[r]^-{G^\dagger \,\zeta\,F^\dagger} ~&~ 
G^\dagger\,G\, F^\dagger \ar@{=>}[r]^-{\epsilon_G\,F^\dagger}~&~ 
F^\dagger
}\quad,
\end{flalign}
\end{subequations}
where $\eta_F$ is the unit of the adjunction $F^\dagger \dashv F$ and 
$\epsilon_G$ the counit of the adjunction $G^\dagger \dashv G$.
\end{rem}

The following key result about bilimits and bicolimits in $\Pr^{L}$ and $\Pr^R$ is proven
in \cite{Bird}, see also \cite{BCJF} for a sketch.
\begin{theo}\label{theo:PrL/R limits}
The $2$-categories $\Pr^L$ and $\Pr^R$ from Definition \ref{def:PrL/R} 
admit all small bilimits and, as a consequence of Lemma \ref{lem:PrLvsPrR}, 
they also admit all small bicolimits. The forgetful $2$-functors
$\Pr^L\to\CAT$ and $\Pr^R\to \CAT$ preserve and reflect all bilimits.
\end{theo}

\begin{rem} \label{rem:PrL/R:limitreflection}
The statement that the forgetful $2$-functors $\Pr^{L/R} \to \CAT$
reflect all bilimits does not appear explicitly in \cite{Bird}, but it
is a simple consequence of the following argument.
Let us observe that the $2$-subcategories $\Pr^{L/R} \subseteq \CAT$ are 
closed under equivalences, and that the forgetful $2$-functors $\Pr^{L/R} \to \CAT$ reflect 
equivalences since any fully faithful and essentially surjective functor is both a left and right adjoint.
Since the forgetful $2$-functors preserve all bilimits by \cite{Bird},
it then follows that they also reflect all bilimits.
\end{rem}

\begin{constr}[Computing bilimits in $\Pr^{L/R}$]\label{constr:computingbilimits}
We would like to emphasize that Theorem \ref{theo:PrL/R limits}
is very useful to compute bilimits.
Given any diagram (i.e.\ pseudo-functor) $X : \DD\to \Pr^{L/R}$ from a small category $\DD$,
we can compute its bilimit in $\Pr^{L/R}$ by the following construction: 
\begin{enumerate}
\item Postcompose $X$ with the forgetful $2$-functor, which yields a
pseudo-functor $X : \DD\to\CAT$ to the $2$-category $\CAT$.

\item Compute the bilimit of $X : \DD\to\CAT$ in $\CAT$. For this one can use, for example,
the explicit model
\begin{flalign}\label{eqn:bilimXformula}
\bilim(X)\,=\,\Hom(\Delta\mathbf{1},X)\,\in\,\CAT
\end{flalign}
given by the category of pseudo-natural transformations from the constant
diagram $\Delta\mathbf{1} : \DD\to \CAT$ to $X : \DD\to\CAT$ and their modifications,
where $\mathbf{1}\in\CAT$ denotes the category consisting of a single object 
and its identity morphism.

\item Theorem \ref{theo:PrL/R limits} implies that \eqref{eqn:bilimXformula} is a locally presentable
category, i.e.\ $\bilim(X)\in\Pr^{L/R}$, and that the universal pseudo-cone $\Delta\bilim(X) \Rightarrow X$,
i.e.\ the projection maps from the bilimit to the diagram, consists of left/right adjoint functors.
This provides an explicit model for the bilimit of our original diagram
$X : \DD\to \Pr^{L/R}$ in $\Pr^{L/R}$. \qedhere
\end{enumerate}
\end{constr}

\begin{constr}[Computing bicolimits in $\Pr^{L/R}$]\label{constr:computingbicolimits}
Combining Theorem \ref{theo:PrL/R limits} and Lemma \ref{lem:PrLvsPrR}, one
also obtains an explicit approach to compute bicolimits.
To avoid notational clutter, we will spell out this construction only
for the case of the bicolimit of a diagram (i.e.\ pseudo-functor) $X : \DD\to \Pr^{L}$
from a small category $\DD$ to $\Pr^L$. The case of diagrams in $\Pr^R$ works analogously.
\begin{enumerate}
\item Postcompose $X$ with the pseudo-functor $(-)^{\dagger}$ from \eqref{eqn:LtoR},
which yields a pseudo-functor $X^\dagger : \DD\to (\Pr^{R})^{\mathrm{coop}}$. This is the same datum
as a pseudo-functor $X^\dagger : \DD^\op\to \Pr^R$ from the opposite category $\DD^\op$
to $\Pr^R$.

\item Compute the \textit{bilimit} of $X^\dagger : \DD^\op\to \Pr^R$ using,
for example, Construction \ref{constr:computingbilimits}. 
This defines an object $\bilim(X^\dagger)\in \Pr^R$
with a universal pseudo-cone $\Delta\bilim(X^\dagger) \Rightarrow X^\dagger$.

\item Apply the quasi-inverse pseudo-functor $(-)^\dagger$ from \eqref{eqn:RtoL} to the 
universal pseudo-cone, which defines a universal pseudo-cocone $X \simeq X^{\dagger\dagger} \Rightarrow 
\Delta\bilim(X^\dagger)^\dagger=\Delta\bilim(X^\dagger)$ for the original diagram $X : \DD\to \Pr^{L}$, 
where in the last step we used that $(-)^\dagger$ acts as the identity on objects, i.e.\ 
$\bilim(X^\dagger)^\dagger=\bilim(X^\dagger)$.
A model for the bicolimit of $X$ is then given by
\begin{flalign}
\bicolim\big(X : \DD \to \Pr^L\big)\,=\,
\bilim\big(X^\dagger: \DD^\op\to \Pr^R\big)
\end{flalign}
together with the given universal pseudo-cocone.
\end{enumerate}
In simpler words, this construction can be described as follows:
If one would like to compute the bicolimit of a diagram in $\Pr^{L/R}$,
one can equivalently compute the bilimit of the adjoint diagram in $\Pr^{R/L}$. The latter 
is relatively easy to do, e.g.\ by using the
explicit Construction \ref{constr:computingbilimits}.
\end{constr}

The above results imply that the concept of a $\Pr^{L/R}$-valued stack on 
some site is equivalent to that of a $\Pr^{R/L}$-valued costack on the same site.
Let us state this observation explicitly for the case of $\Pr^R$-valued stacks on $\Loc$ 
and $\Pr^L$-valued costacks on $\Loc$, which will be important 
for proving our main results in Section \ref{sec:HKstack}.
\begin{cor}\label{cor:adjointpseudofunctor}
Let $X : \Loc^\op\to \Pr^R$ be a pseudo-functor taking values in the $2$-category $\Pr^R$.
Then $X$ is a stack in the sense of Definition \ref{def:stack} if and only if
the adjoint pseudo-functor $X^\dagger : \Loc\to \Pr^L$ is a costack in the following sense: 
For every object $M\in\Loc$ and every ($D$-stable) 
causally convex open cover $\U=\{U_i\subseteq M\}$,
the canonical functor
\begin{subequations}\label{eqn:costack}
\begin{flalign}\label{eqn:costack1}
X^\dagger(\U)~\longrightarrow~X^\dagger(M)
\end{flalign}
in $\Pr^L$ from the codescent category 
\begin{flalign}
X^\dagger(\U)\,:=\,\bicolim\bigg(
\xymatrix{
\coprod\limits_{i}X^\dagger(U_i) \ar[r]  ~&~\ar@<0.75ex>[l]\ar@<-0.75ex>[l]
\coprod\limits_{ij}X^\dagger(U_{ij})
\ar@<0.75ex>[r]\ar@<-0.75ex>[r] ~&~ \ar@<1.5ex>[l]\ar@<-1.5ex>[l] \ar[l] 
\coprod\limits_{ijk}X^\dagger(U_{ijk})
}\bigg) 
\end{flalign}
\end{subequations}
is an equivalence in $\Pr^L$.
\end{cor}
\begin{proof}
This follows directly from the bilimit/bicolimit rewriting in 
Construction \ref{constr:computingbicolimits}. Indeed,
the codescent category 
\begin{flalign}
X^\dagger(\U)\,=\,X(\U)
\end{flalign}
of $X^\dagger$ coincides with the descent category of $X$, and the canonical functor
\eqref{eqn:costack1} is the left adjoint of the canonical functor
\begin{flalign}
X(M)~\longrightarrow~X(\U)
\end{flalign}
to the descent category. (Recall that $X^\dagger(M) = X(M)$ because
$(-)^\dagger$ acts as the identity on objects.)
\end{proof}

\begin{rem}
Using a similar terminology as in Remark \ref{rem:prestack}, 
we call the pseudo-functor $X^\dagger : \Loc^\op\to \Pr^L$ 
a \textit{precostack} whenever the functor
$X^\dagger(\U)\to X^\dagger(M)$ is fully faithful, 
for every object $M\in\Loc$ and every ($D$-stable) causally convex open cover 
$\U=\{U_i\subseteq M\}$.
\end{rem}

%%%%%%%%%%%%%%%%%%%%%%%%%%%%%%%%%%%%%%%%%%%%%%%%
%%%%%%%%%%%%%%%%%%%%%%%%%%%%%%%%%%%%%%%%%%%%%%%%

\section{\label{sec:HK2functor}Haag-Kastler \texorpdfstring{$2$-functors}{2-functors}}
Haag-Kastler-style AQFTs \cite{HaagKastler} are AQFTs which are defined on 
suitable causally convex opens in a fixed oriented, time-oriented
and globally hyperbolic Lorentzian manifold $M\in \Loc$. 
Depending on whether or not one wishes to demand a boundedness
condition for these opens, one can formalize such AQFTs by using 
either the orthogonal category $\ovr{\RC(M)}$ of relatively compact causally convex
opens in $M$ or the orthogonal category $\ovr{\COpen(M)}$ of all causally convex opens in $M$,
see Example \ref{ex:OCat}. Furthermore, if desired, the time-slice axiom can be implemented 
as in Proposition \ref{prop:timeslice} either as a property or, equivalently, 
through an orthogonal localization. 
The Haag-Kastler $2$-functors we define and study in this section describe the behavior 
of Haag-Kastler-style AQFTs under $\Loc$-morphisms $f : M\to N$. We consider
all of the above mentioned variations of Haag-Kastler-style AQFTs 
and establish relationships between different variations and also a comparison to locally covariant AQFT.

\subsection{\label{subsec:HK2functorCOpen}The case of causally convex opens}
In this subsection we describe the variants of the Haag-Kastler $2$-functor 
which are associated with Haag-Kastler-style AQFTs that are modeled on
the orthogonal categories $\ovr{\COpen(M)}$ of all causally convex opens
in $M\in\Loc$ from Example \ref{ex:OCat}.

\subsubsection{Definition and properties}
Let us start by observing that the assignment $M\mapsto \ovr{\COpen(M)}$ 
can be upgraded to a $2$-functor 
\begin{flalign}\label{eqn:COpenfunctor}
\ovr{\COpen(-)} \,:\, \Loc~\longrightarrow~\Cat^\perp
\end{flalign}
from $\Loc$ to the $2$-category $\Cat^\perp$ of orthogonal categories, orthogonal functors and natural transformations.
Indeed, given any $\Loc$-morphism $f:M\to N$, we can define an orthogonal functor (denoted
with abuse of notation by the same symbol $f$)
\begin{flalign}\label{eqn:fpullback}
f\,:=\,\ovr{\COpen(f)} \,:\, \ovr{\COpen(M)}~\longrightarrow~\ovr{\COpen(N)}~~,\quad 
U\subseteq M ~\longmapsto~f(U)\subseteq N
\end{flalign}
which sends causally convex opens in $M$ to their images under $f$ in $N$. 
Note that this assignment is strictly $2$-functorial.
\begin{defi}\label{def:HK2functor}
The \textit{Haag-Kastler $2$-functor}
\begin{subequations}
\begin{flalign}
\HK\,:\, \Loc^\op~\longrightarrow~\CAT
\end{flalign}
is defined by assigning to each object $M\in\Loc$ the category
\begin{flalign}
\HK(M)\,:=\, \AQFT(\ovr{\COpen(M)})\,\in\,\CAT
\end{flalign}
of Haag-Kastler-style AQFTs on $M$ and to each $\Loc$-morphism $f : M\to N$
the pullback functor
\begin{flalign}
\begin{gathered}
\xymatrix@R=1em{
\ar@{=}[d]\HK(N) \ar[r]^-{\HK(f)\,:=\, f^\ast} ~&~ \HK(M) \ar@{=}[d]\\
\AQFT(\ovr{\COpen(N)}) \ar[r]_-{f^\ast}~&~ \AQFT(\ovr{\COpen(M)})
}
\end{gathered}
\end{flalign}
\end{subequations}
from \eqref{eqn:Fpullbackfunctor} which is 
associated to the orthogonal functor $f : \ovr{\COpen(M)} \to \ovr{\COpen(N)}$ in \eqref{eqn:fpullback}.
\end{defi}

\begin{rem}
We would like to emphasize that the pullback functor $f^\ast$ describes 
the obvious and expected concept of pulling back Haag-Kastler-style AQFTs along $\Loc$-morphisms $f:M\to N$.
Indeed, given any $\AAA\in\HK(N)$ on $N$, the pullback $f^\ast(\AAA)\in\HK(M)$ on $M$ assigns
to a causally convex open subset $U\subseteq M$ the same algebra 
$f^\ast(\AAA)(U) = \AAA(f(U))$ as the
original theory assigns to the image $f(U)\subseteq N$. 
\end{rem}

It is natural to ask whether or not 
the Haag-Kastler $2$-functor is a stack on $\Loc$ in the sense
of Definition \ref{def:stack}, i.e.\ if Haag-Kastler-style AQFTs satisfy a 
local-to-global property with respect to either
causally convex open covers, or their Cauchy development stable counterparts.
This is in general not the case.
\begin{propo}\label{prop:HKnotastack}
Suppose that the category of algebras $\Alg_{\mathsf{uAs}}(\TT)$ has two objects
$A,B\in \Alg_{\mathsf{uAs}}(\TT)$ for which the $\Hom$-set $\Hom(A,B)$ is not a 
singleton.\footnote{This technical condition is very mild and it is only used to rule out 
pathological examples of $\TT$, such as the one-object category $\mathbf{1}$. 
In the standard case where $\TT=\Vec_\bbK$, one has that 
$\Hom(A,\bbK)=\emptyset$ is empty for every simple noncommutative $\bbK$-algebra $A$ 
and that $\Hom(A,A)$ has more than $1$ element for every $\bbK$-algebra $A$ with non-trivial automorphisms.}
Then the Haag-Kastler $2$-functor $\HK$ from Definition \ref{def:HK2functor} is not a stack with respect to 
either Grothendieck topology from Definition \ref{def:stack} on $\Loc$. It is not even a prestack.
\end{propo}
\begin{proof}
Let us choose any object $M\in \Loc$ with a ($D$-stable)
causally convex open cover $\U=\{U_i\subseteq M\}$
such that every $U_i\subset M$ is a proper subset of $M$.
(The existence of such a cover is guaranteed because
$M$ is globally hyperbolic, so also strongly causal.
For the $D$-stable case, see also
Proposition \ref{prop:Cauchy_development:small_D-stable_neighbourhoods}
and Remark \ref{rem:Cauchy_development:fine_D-stable_covers}.)
We will now show that the functor
\begin{flalign}\label{eqn:HKdescentmaptmp}
\HK(M)~\longrightarrow~\HK(\U)
\end{flalign}
to the descent category is not fully faithful, which in particular implies that it can not 
be an equivalence. For this we consider two specific objects $\AAA,\BBB\in \HK(M)$ 
which are defined by
\begin{flalign}\label{eqn:AAABBBexample}
\AAA(U)\,:=\, \begin{cases}
A &~,~~\text{if }U=M\\
I &~,~~\text{if }U\subset M
\end{cases}\quad,\qquad
\BBB(U)\,:=\, \begin{cases}
B &~,~~\text{if }U=M\\
I &~,~~\text{if }U\subset M
\end{cases}\quad,
\end{flalign}
for all causally convex opens $U\subseteq M$,
where $I\in \Alg_{\mathsf{uAs}}(\TT)$ denotes the initial algebra
and $A,B\in \Alg_{\mathsf{uAs}}(\TT)$ are the algebras from our hypotheses.
We endow $\AAA$ and $\BBB$ with the AQFT structures which are
uniquely determined by the universal property of the initial algebra.
\sk

Since every $U_i\subset M$ is a proper subset, the two AQFTs $\AAA,\BBB\in\HK(M)$ defined above
are mapped via the functor \eqref{eqn:HKdescentmaptmp} to the same object in the descent category.
In the model for the descent category from Remark \ref{rem:descentexplicit}, this object reads
as $\mathfrak{I}:=(\{\mathfrak{I}_{U_i}\},\{\id_{\mathfrak{I}_{U_{ij}}}\})\in \HK(\U)$, where
$\mathfrak{I}_{U_i}\in \HK(U_i)$ denotes the initial object, i.e.\ the constant 
AQFT sending all causally convex opens $U\subseteq U_i$ to the initial algebra. Acting with the functor
\eqref{eqn:HKdescentmaptmp} on $\Hom$-sets, we obtain
\begin{flalign}
\Hom_{\HK(M)}(\AAA,\BBB)\,\cong\,\Hom_{\Alg_{\mathsf{uAs}}(\TT)}(A,B)~
\longrightarrow~\Hom_{\HK(\U)}(\mathfrak{I},\mathfrak{I})\quad.
\end{flalign}
Since $\Hom_{\HK(\U)}(\mathfrak{I},\mathfrak{I})$ is a singleton
and $\Hom_{\Alg_{\mathsf{uAs}}(\TT)}(A,B)$ is by our hypotheses not a singleton,
this map can not be a bijection. In particular, the functor
\eqref{eqn:HKdescentmaptmp} fails to be full when 
$\Hom_{\HK(M)}(\AAA,\BBB) = \emptyset$ is empty 
and it fails to be faithful when $\Hom_{\HK(M)}(\AAA,\BBB)$ contains more than one element.
\end{proof}

\subsubsection{Comparison to locally covariant AQFT}
Our next aim is to explain how the Haag-Kastler $2$-functor from 
Definition \ref{def:HK2functor} is related to
locally covariant AQFT \cite{BFV,FewsterVerch}. For this we recall
the following standard concept.
\begin{defi}\label{def:HKpoints}
The \textit{category of points} of the Haag-Kastler $2$-functor
is defined as the category
\begin{flalign}
\HK(\mathrm{pt})\,:=\,\Hom(\Delta\mathbf{1},\HK)\,\in\,\CAT
\end{flalign}
of pseudo-natural transformations from the constant
$2$-functor $\Delta\mathbf{1} : \Loc^\op\to \CAT$ to $\HK:\Loc^\op\to\CAT$ 
and their modifications, where we recall that
$\mathbf{1}\in\CAT$ denotes the category consisting of a single object 
and its identity morphism. 
\end{defi}

\begin{rem}\label{rem:HKpoints}
Spelling out the definitions of pseudo-natural transformations
and modifications (see e.g.\ \cite{Yau2Cats}) and using that a functor $\mathbf{1}\to \DD$
from the one-object category to any category $\DD$ is the same datum as an object in $\DD$,
one obtains the following explicit description of the category of points $\HK(\mathrm{pt})$:
\begin{itemize}
\item An object in $\HK(\mathrm{pt})$ is a tuple $(\{\AAA_M\},\{\alpha_f\})$
consisting of a family of objects $\AAA_M\in\HK(M)$, for all $M\in\Loc$,
and a family of isomorphisms $\alpha_f : \AAA_M\Rightarrow f^\ast(\AAA_N)$ in $\HK(M)$,
for all $\Loc$-morphisms $f : M\to N$, which satisfies the following conditions:
\begin{subequations}\label{eqn:HKptconditions}
\begin{itemize}
\item[(i)]
For all composable $\Loc$-morphisms $f : M\to N$ and $g : N\to O$, the diagram
\begin{flalign}
\begin{gathered}
\xymatrix@C=3.5em{
\ar@{=>}[d]_-{\alpha_{gf}}\AAA_M \ar@{=>}[r]^-{\alpha_f}~&~ f^\ast(\AAA_N)\ar@{=>}[d]^-{f^\ast(\alpha_g)}\\
(g\,f)^\ast(\AAA_O)\ar@{=}[r]~&~f^\ast g^\ast(\AAA_O)
}
\end{gathered}
\end{flalign}
in $\HK(M)$ commutes.

\item[(ii)] For all objects $M\in\Loc$, the diagram
\begin{flalign}
\begin{gathered}
\xymatrix@C=3.5em{
\ar@{=>}[dr]_-{\id_{\AAA_M}}\AAA_M \ar@{=>}[r]^-{\alpha_{\id_M}}~&~\id_M^\ast(\AAA_M)\ar@{=}[d]\\
~&~\AAA_M
}
\end{gathered}
\end{flalign}
in $\HK(M)$ commutes.
\end{itemize}
\end{subequations}

\item A morphism $\{\zeta_M\} : (\{\AAA_M\},\{\alpha_f\})\Rightarrow (\{\BBB_M\},\{\beta_f\})$ in $\HK(\mathrm{pt})$
is a family of morphisms $\zeta_M :\AAA_M\Rightarrow \BBB_M$ in $\HK(M)$, for all $M\in\Loc$, 
which satisfies the following condition: For all $\Loc$-morphisms $f : M\to N$, the diagram
\begin{flalign}\label{eqn:HKmorphisms}
\begin{gathered}
\xymatrix@C=3.5em{
\ar@{=>}[d]_-{\alpha_f} \AAA_M \ar@{=>}[r]^-{\zeta_M}~&~ \BBB_M\ar@{=>}[d]^-{\beta_f}\\
 f^\ast(\AAA_N) \ar@{=>}[r]_-{f^\ast(\zeta_N)}~&~ f^\ast(\BBB_N)\\
}
\end{gathered}
\end{flalign}
in $\HK(M)$ commutes.
\end{itemize}
In simple words, this description shows that points of the Haag-Kastler $2$-functor
are natural families of Haag-Kastler-style AQFTs over all $M\in\Loc$.
\end{rem}

In order to set up a comparison between the category of points $\HK(\mathrm{pt})$
and the category of locally covariant AQFTs $\AQFT(\ovr{\Loc})$, we recall from 
Example \ref{ex:OCat} the orthogonal functors 
\begin{flalign}\label{eqn:jMfunctor}
k_M \,:\, \ovr{\COpen(M)}~\longrightarrow~ \ovr{\Loc}~~,\quad U\subseteq M~\longmapsto~U\quad,
\end{flalign} 
for all $M\in\Loc$. Given any $\Loc$-morphism $f:M\to N$, we can compose
the orthogonal functor $f : \ovr{\COpen(M)}\to \ovr{\COpen(N)}$ from \eqref{eqn:fpullback}
with the orthogonal functor $k_N : \ovr{\COpen(N)}\to \ovr{\Loc}$ and define a natural isomorphism
\begin{subequations}\label{eqn:jftransformation}
\begin{flalign}
k_f \,:\, k_M~\Longrightarrow~k_N\,f
\end{flalign}
of orthogonal functors from $\ovr{\COpen(M)}$ to $\ovr{\Loc}$ in terms of the component $\Loc$-isomorphisms
\begin{flalign}
(k_f)_U\,:=\, f\vert_U \,:\,U~\stackrel{\cong}{\longrightarrow} ~f(U)
\end{flalign}
\end{subequations}
which are obtained by restricting and corestricting the given $\Loc$-morphism $f:M\to N$,
for all causally convex opens $U\subseteq M$.
\begin{constr}\label{constr:LCQFTtoHK}
We define a functor
\begin{flalign}\label{eqn:LCQFTtoHK}
\AQFT(\ovr{\Loc})~\longrightarrow~\HK(\mathrm{pt})
\end{flalign}
from the category of locally covariant AQFTs to the category of points of the Haag-Kastler $2$-functor.
To an object $\AAA\in\AQFT(\ovr{\Loc})$, this functor assigns the tuple
\begin{flalign}
\Big(\big\{k_M^\ast(\AAA):= \AAA\,k_M\big\},
\big\{\AAA\,k_f \,:\, k_M^\ast(\AAA) = \AAA\,k_M\,\Rightarrow\,\AAA\,k_N\,f = f^\ast k_N^\ast (\AAA)\big\}\Big)
\,\in\,\HK(\mathrm{pt})
\end{flalign}
which is obtained by applying the pullback construction \eqref{eqn:Fpullbackfunctor} 
to \eqref{eqn:jMfunctor} and whiskering with \eqref{eqn:jftransformation}. One directly checks that
this tuple satisfies the compatibility conditions \eqref{eqn:HKptconditions} 
from Remark \ref{rem:HKpoints}.
To a morphism $\zeta : \AAA\Rightarrow  \BBB$ in $\AQFT(\ovr{\Loc})$, this functor assigns
the tuple
\begin{flalign}
\Big\{k_M^\ast(\zeta)\,:\, k_M^\ast(\AAA) \,\Rightarrow\, k_M^\ast(\BBB) \Big\}\,:\,
\big(\big\{k_M^\ast(\AAA)\big\} , \big\{\AAA\,k_f\big\}\big)~\Longrightarrow~
\big(\big\{k_M^\ast(\BBB)\big\} , \big\{\BBB\,k_f\big\}\big)
\end{flalign}
which is obtained by applying the pullback construction \eqref{eqn:Fpullbackfunctor} 
to \eqref{eqn:jMfunctor}.
One directly checks that this tuple satisfies the compatibility 
conditions \eqref{eqn:HKmorphisms} from Remark \ref{rem:HKpoints}.
\end{constr}

\begin{constr}\label{constr:HKtoLCQFT}
We now define a functor
\begin{flalign}\label{eqn:HKtoLCQFT}
\HK(\mathrm{pt})~\longrightarrow~\AQFT(\ovr{\Loc})
\end{flalign}
which goes in the reverse direction of Construction \ref{constr:LCQFTtoHK}.
To an object $(\{\AAA_M\},\{\alpha_f\})\in \HK(\mathrm{pt})$, this functor assigns the locally covariant 
AQFT $\AAA\in\AQFT(\ovr{\Loc})$ which is defined by the following functor $\AAA : \Loc\to\Alg_{\mathsf{uAs}}(\TT)$:
To an object $M\in\Loc$, we assign the evaluation
\begin{subequations}\label{eqn:AAAconstruction}
\begin{flalign}
\AAA(M)\,:=\, \AAA_M(M)\,\in\,\Alg_{\mathsf{uAs}}(\TT)
\end{flalign}
of the corresponding Haag-Kastler-style AQFT $\AAA_M\in\HK(M)$ on the terminal object $M\in\COpen(M)$. 
To a $\Loc$-morphism $f:M\to N$, 
we assign the $\Alg_{\mathsf{uAs}}(\TT)$-morphism defined by
\begin{flalign}
\xymatrix@C=2.5em{
\AAA(f) \,:\, \AAA(M) \,=\,\AAA_M(M) \ar[r]^-{(\alpha_f)_M^{}} ~&~
\AAA_N(f(M))\ar[r]~ &~\AAA_N(N)\,= \,\AAA(N)
}\quad,
\end{flalign}
\end{subequations}
where the second arrow uses the functorial structure of $\AAA_N\in \HK(N)$
for the inclusion $f(M)\subseteq N$. One can immediately verify that the functor 
$\AAA :\Loc\to \Alg_{\mathsf{uAs}}(\TT)$ 
satisfies the $\perp$-commutativity axiom 
from Definition \ref{def:AQFT} by observing that the diagram
\begin{equation}\label{eqn:HK2LCQFTcausality}
% https://q.uiver.app/#q=WzAsNSxbMSwxLCJcXG1hdGhmcmFre0F9X3tOfShmKE1fMSkpIFxcb3RpbWVzIFxcbWF0aGZyYWt7QX1fTihmXzIoTV8yKSkiXSxbMiwxLCJcXG1hdGhmcmFre0F9X3tOfShOKSBcXG90aW1lcyBcXG1hdGhmcmFre0F9X3tOfShOKSJdLFsxLDIsIlxcbWF0aGZyYWt7QX1fe059KE4pIFxcb3RpbWVzIFxcbWF0aGZyYWt7QX1fe059KE4pIl0sWzIsMiwiXFxtYXRoZnJha3tBfV97Tn0oTikiXSxbMCwwLCJcXG1hdGhmcmFre0F9X3tNXzF9KE1fMSkgXFxvdGltZXMgXFxtYXRoZnJha3tBfV97TV8yfShNXzIpIl0sWzAsMV0sWzAsMl0sWzIsMywiXFxtdV9OXntcXG1hdGhybXtvcH19IiwyXSxbMSwzLCJcXG11X04iXSxbNCwwLCIoXFxhbHBoYV97Zl8xfSlfe01fMX0gXFxvdGltZXMgKFxcYWxwaGFfe2ZfMn0pX3tNXzJ9IiwxXSxbNCwyLCJcXG1hdGhmcmFre0F9KGZfMSkgXFxvdGltZXMgXFxtYXRoZnJha3tBfShmXzIpIiwyLHsiY3VydmUiOjN9XSxbNCwxLCJcXG1hdGhmcmFre0F9KGZfMSkgXFxvdGltZXMgXFxtYXRoZnJha3tBfShmXzIpIiwwLHsiY3VydmUiOi0zfV1d
\begin{tikzcd}
	{\mathfrak{A}_{M_1}(M_1) \otimes \mathfrak{A}_{M_2}(M_2)} \\
	& {\mathfrak{A}_{N}(f_1(M_1)) \otimes \mathfrak{A}_N(f_2(M_2))} & {\mathfrak{A}_{N}(N) \otimes \mathfrak{A}_{N}(N)} \\
	& {\mathfrak{A}_{N}(N) \otimes \mathfrak{A}_{N}(N)} & {\mathfrak{A}_{N}(N)}
	\arrow[from=2-2, to=2-3]
	\arrow[from=2-2, to=3-2]
	\arrow["{\mu_N^{\mathrm{op}}}"', from=3-2, to=3-3]
	\arrow["{\mu_N^{}}", from=2-3, to=3-3]
	\arrow["{(\alpha_{f_1})_{M_1} \otimes (\alpha_{f_2})_{M_2}}"{description}, from=1-1, to=2-2]
	\arrow["{\mathfrak{A}(f_1) \otimes \mathfrak{A}(f_2)}"', curve={height=18pt}, from=1-1, to=3-2]
	\arrow["{\mathfrak{A}(f_1) \otimes \mathfrak{A}(f_2)}", curve={height=-18pt}, from=1-1, to=2-3]
\end{tikzcd}
\end{equation}
in $\TT$ commutes, for all orthogonal pairs $(f_1:M_1\to N)\perp (f_2:M_2\to N)$ 
in $\ovr{\Loc}$. This follows from the $\perp$-commutativity axiom of $\AAA_N\in\HK(N)$
and the orthogonal pair $(f_1(M_1)\subseteq N)\perp (f_2(M_2)\subseteq N)$ in $\ovr{\COpen(N)}$.
\sk

To a morphism $\{\zeta_M\} :(\{\AAA_M\},\{\alpha_f\})\Rightarrow (\{\BBB_M\},\{\beta_f\})$ in $\HK(\mathrm{pt})$, 
the functor \eqref{eqn:HKtoLCQFT} assigns the morphism $\zeta : \AAA\Rightarrow \BBB$ 
of locally covariant AQFTs which is defined by the following components
\begin{flalign}
\xymatrix@C=3em{
\AAA(M) \,=\,\AAA_M(M) \ar[r]^-{(\zeta_M)_M^{}}~&~\BBB_M(M)\,=\, \BBB(M)\quad,
}
\end{flalign}
for all $M\in\Loc$.
Recalling \eqref{eqn:HKmorphisms} and \eqref{eqn:AAAconstruction}, one obtains
the commutative diagrams
\begin{flalign}
\begin{gathered}
\xymatrix@C=3em{
\ar[d]_-{(\zeta_M)_M^{}}\AAA_M(M) \ar[r]^-{(\alpha_f)_M^{}}~&~ \ar[d]_-{(\zeta_N)_{f(M)}} \AAA_N(f(M)) 
\ar[r]~&~ \AAA_N(N)\ar[d]^-{(\zeta_{N})_N^{}}\\
\BBB_M(M) \ar[r]_-{(\beta_f)_M^{}}~&~\BBB_N(f(M)) \ar[r] ~&~ \BBB_N(N)
}
\end{gathered}
\end{flalign}
in $\Alg_{\mathsf{uAs}}(\TT)$, for all $\Loc$-morphisms $f:M\to N$,
which confirm that $\zeta$ as defined above is a natural transformation. 
\end{constr}

\begin{theo}\label{theo:LCQFTvsHK}
The two functors defined in Constructions \ref{constr:LCQFTtoHK} and \ref{constr:HKtoLCQFT}
are quasi-inverse to each other. Hence, they exhibit an equivalence
\begin{flalign}
\HK(\mathrm{pt})~\simeq~\AQFT(\ovr{\Loc})
\end{flalign}
between the category
of points $\HK(\mathrm{pt})$ of the Haag-Kastler $2$-functor
and the category  $\AQFT(\ovr{\Loc})$ of locally covariant AQFTs.
\end{theo}
\begin{proof}
One immediately checks that the composition 
$\AQFT(\ovr{\Loc})\to \HK(\mathrm{pt})\to \AQFT(\ovr{\Loc})$
of \eqref{eqn:LCQFTtoHK} followed by \eqref{eqn:HKtoLCQFT}
is the identity functor on $\AQFT(\ovr{\Loc})$.
\sk

Concerning the composition $\HK(\mathrm{pt})\to \AQFT(\ovr{\Loc})\to \HK(\mathrm{pt})$,
we apply first \eqref{eqn:HKtoLCQFT} and then \eqref{eqn:LCQFTtoHK}
to any object $(\{\AAA_M\},\{\alpha_f\})\in \HK(\mathrm{pt})$, which results 
in the object
\begin{subequations}
\begin{flalign}
\big(\{k_M^\ast(\AAA)\},\{k_M^\ast(\AAA)\Rightarrow f^\ast k_N^\ast(\AAA)\} \big)\,\in\,\HK(\mathrm{pt})\quad.
\end{flalign}
More explicitly, the family of objects is specified by
\begin{flalign}
k_M^\ast(\AAA)(U)\,=\, \AAA_U(U)\quad,
\end{flalign}
for all $M\in\Loc$ and all causally convex opens $U\subseteq M$, and the family 
of isomorphisms by
\begin{flalign}
\begin{gathered}
\xymatrix@C=3.5em{
\ar@{=}[d] k_M^\ast(\AAA)(U) \ar[r]~&~f^\ast k_N^\ast(\AAA)(U) \ar@{=}[d]\\
\AAA_U(U)\ar[r]_-{(\alpha_{f\vert_U})_U^{}}~&~ \AAA_{f(U)}(f(U))
}
\end{gathered}\quad,
\end{flalign}
\end{subequations}
for all $\Loc$-morphisms $f:M\to N$ and all causally convex opens $U\subseteq M$,
where $f\vert_U : U\to f(U)$ denotes the $\Loc$-isomorphism obtained by restricting and corestricting $f$.
There exists an isomorphism in $\HK(\mathrm{pt})$ from this object to the original object,
which is given by the components
\begin{flalign}\label{tmp:HK2LCQFT2HKiso}
\xymatrix@C=3em{
\AAA_U(U) \ar[r]^-{(\alpha_{\iota_U^M})_U}~&~\AAA_M(U)
}\quad,
\end{flalign}
for all $M\in\Loc$ and all causally convex opens $U\subseteq M$, where 
$\iota_U^M : U\to M$ denotes the canonical inclusion $\Loc$-morphism.
From this one shows that the composition of functors $\HK(\mathrm{pt})\to\AQFT(\ovr{\Loc})\to \HK(\mathrm{pt})$
is naturally isomorphic to the identity functor on  $\HK(\mathrm{pt})$.
\end{proof}

\begin{rem}\label{rem:CastHKvLCQFT}
We would like to highlight that the result of 
Theorem \ref{theo:LCQFTvsHK} adapts to the case
where one considers AQFTs taking values in the category 
of $C^\ast$-algebras. This is a direct consequence of
the fact that the Constructions \ref{constr:LCQFTtoHK} 
and \ref{constr:HKtoLCQFT} do not make use of any operadic technology
and hence they carry over ad verbum to $C^\ast$-algebras.
\end{rem}

\begin{rem}\label{rem:(op)laxpoints}
The Equivalence Theorem \ref{theo:LCQFTvsHK} 
does not only show that locally covariant AQFTs are subsumed
in our Haag-Kastler $2$-functor framework, but it also
provides some natural directions for generalizations. 
For instance, one may consider generalizations of the concept
of points from Definition \ref{def:HKpoints} which are given
by oplax or lax transformations $\AAA : \Delta \mathbf{1}\Rightarrow\HK$.
Concretely, this amounts to families $\{\AAA_M\}$ of Haag-Kastler theories
related by \textit{not necessarily invertible} morphisms $\alpha_f : \AAA_M\Rightarrow f^\ast(\AAA_N)$
in the oplax case or morphisms $\tilde{\alpha}_f :f^\ast(\AAA_N)\Rightarrow \AAA_M$
in the lax case, which satisfy similar cocycle conditions as the ones in Remark \ref{rem:HKpoints}.
It is interesting to observe that oplax points already appeared implicitly in \cite[Section 6]{BDHS}
where a certain quotient construction for AQFTs was proposed with the aim to restore isotony 
for gauge invariant observables in a quantum gauge theory.
\sk

For completeness, let us provide an abstract version of this construction.
Let $\AAA\in\AQFT(\ovr{\Loc})$ be a locally covariant AQFT and consider
its associated point $(\{\AAA_M\},\{\alpha_f: \AAA_M \Rightarrow f^\ast(\AAA_N)\})$ 
of the Haag-Kastler $2$-functor which is given by Theorem \ref{theo:LCQFTvsHK}. 
We define a new family $\{\BBB_M\}$ of Haag-Kastler theories by taking the quotient algebras
\begin{flalign}
\BBB_M(U)\,:=\,\frac{\AAA_M(U)}{\ker\big(\AAA_M(U)\to \AAA_M(M)\big)}\quad,
\end{flalign}
for all $U\in \COpen(M)$. The isomorphisms $\alpha_f: \AAA_M \Rightarrow f^\ast(\AAA_N)$
descend to the quotients and thereby define morphisms $\beta_f: \BBB_M \Rightarrow f^\ast(\BBB_N)$
which are however in general \textit{not invertible}, see e.g.\ the example in \cite[Section 6]{BDHS}.
This defines an oplax point $(\{\BBB_M\},\{\beta_f: \BBB_M \Rightarrow f^\ast(\BBB_N)\})$ 
of the Haag-Kastler $2$-functor $\HK$.
\end{rem}

\subsubsection{Time-slice axiom}
We conclude this subsection with a brief study of the time-slice axiom,
which by Proposition \ref{prop:timeslice} can be implemented
either as an additional property or, equivalently, as a structure
through an orthogonal localization. In the present context, it
will be more convenient to regard the time-slice axiom as an additional property.
\begin{defi}\label{def:HK2functortimeslice}
The \textit{time-sliced Haag-Kastler $2$-functor}
is defined as the $2$-subfunctor $\HK^W\subseteq \HK$
of the Haag-Kastler $2$-functor from Definition \ref{def:HK2functor}
which assigns to every $M\in\Loc$ the full subcategory 
$\HK^W(M)\subseteq \HK(M)$ consisting of all Haag-Kastler-style AQFTs
on $M$ which satisfy the time-slice axiom.
\end{defi}

\begin{rem}\label{rem:HK2functortimeslice}
We note that Proposition \ref{prop:timeslice} and Example \ref{ex:localizations}
provide us with the following equivalent model for the time-sliced Haag-Kastler $2$-functor.
The assignment of the localized orthogonal categories 
$M\mapsto \ovr{\COpen(M)}[W_M^{-1}]$ from Example \ref{ex:localizations} is $2$-functorial
\begin{subequations}\label{eqn:COpenfunctorW}
\begin{flalign}
\ovr{\COpen(-)}[W_{(-)}^{-1}]\,:\,\Loc~\longrightarrow~\Cat^\perp
\end{flalign}
with action on $\Loc$-morphisms $f:M\to N$ given by
\begin{flalign}\label{eqn:COpenfunctorWf}
f_W\,:=\,\ovr{\COpen(f)}[W_{f}^{-1}] \,:\, \ovr{\COpen(M)}[W_{M}^{-1}]~&\longrightarrow~\ovr{\COpen(N)}[W_N^{-1}]
\quad, \\
\nn U\subseteq M ~&\longmapsto~f(U)\subseteq N
\quad,  \\
\nn (U\to V) ~&\longmapsto~\big(f(U)\to f(V)\big)
\quad,
\end{flalign}
\end{subequations}
see also Appendix \ref{app:localization}.
Replacing in Definition \ref{def:HK2functor} the $2$-functor $\ovr{\COpen(-)}$ by
$\ovr{\COpen(-)}[W_{(-)}^{-1}]$, one obtains an equivalent model for the time-sliced Haag-Kastler $2$-functor
which, with abuse of notation, we denote by the same symbol $\HK^W : \Loc^\op\to \CAT$. Explicitly, this 
$2$-functor assigns to an object $M\in\Loc$ the category 
\begin{subequations}
\begin{flalign}
\HK^W(M)\,=\,\AQFT\big(\ovr{\COpen(M)}[W_M^{-1}]\big)\,\in\,\CAT
\end{flalign}
of AQFTs on the orthogonal localization $\ovr{\COpen(M)}[W_M^{-1}]$ 
and to a $\Loc$-morphism $f:M\to N$ the pullback functor
\begin{flalign}
\begin{gathered}
\xymatrix@R=1em{
\ar@{=}[d]\HK^W(N) \ar[r]^-{\HK^W(f)\,:=\, f_W^\ast} ~&~ \HK^W(M) \ar@{=}[d]\\
\AQFT\big(\ovr{\COpen(N)}[W_N^{-1}]\big) \ar[r]_-{f_W^\ast}~&~ \AQFT\big(\ovr{\COpen(M)}[W_M^{-1}]\big)
}
\end{gathered}
\end{flalign}
\end{subequations}
associated to the orthogonal functor \eqref{eqn:COpenfunctorWf}.
The equivalence between this model and the one in Definition \ref{def:HK2functortimeslice}
is implemented as in Proposition \ref{prop:timeslice} by pullbacks along the orthogonal
localization functors $L_M : \ovr{\COpen(M)}\to \ovr{\COpen(M)}[W_M^{-1}]$, for all $M\in\Loc$.
\end{rem}

The result of Proposition \ref{prop:HKnotastack} remains valid in the present case.
\begin{propo}\label{prop:HKtimeslicenotastack}
Suppose that the category of algebras $\Alg_{\mathsf{uAs}}(\TT)$ has two objects
$A,B\in \Alg_{\mathsf{uAs}}(\TT)$ for which the $\Hom$-set $\Hom(A,B)$ is not a 
singleton. Then the time-sliced Haag-Kastler $2$-functor $\HK^W$ from Definition 
\ref{def:HK2functortimeslice} is not a stack with respect to 
either Grothendieck topology from Definition \ref{def:stack} on $\Loc$.
It is not even a prestack.
\end{propo}
\begin{proof}
The proof is very similar to one of Proposition \ref{prop:HKnotastack}.
The only differences are: 1.)~We take a ($D$-stable) causally convex open cover $\U=\{U_i\subseteq M\}$
such that every $U_i\subseteq M$ does not contain a Cauchy surface of $M$.
2.)~Instead of the objects $\AAA$ and $\BBB$ from \eqref{eqn:AAABBBexample},
we consider the objects $\AAA,\BBB\in \HK^W(M)$ 
which are defined by
\begin{subequations}
\begin{flalign}
\AAA(U)\,&:=\, \begin{cases}
A &~,~~\text{if $U$ contains a Cauchy surface of $M$}\\
I &~,~~\text{otherwise}
\end{cases}\quad,\\
\BBB(U)\,&:=\, \begin{cases}
B &~,~~\text{if $U$ contains a Cauchy surface of $M$}\\
I &~,~~\text{otherwise}
\end{cases}\quad,
\end{flalign}
\end{subequations}
for all causally convex opens $U\subseteq M$. 
We endow $\AAA$ and $\BBB$ with the AQFT structures which are
defined by the universal property of the initial algebra
and the identity morphisms $\id_A : A\to A$ and $\id_B : B\to B$.
As in the proof of Proposition \ref{prop:HKnotastack},
one then shows that the canonical functor $\HK^W(M) \to \HK^W(\mathcal{U})$ 
to the descent category is not fully faithful
by using these $\AAA$ and $\BBB$.
\end{proof}

Our Comparison Theorem \ref{theo:LCQFTvsHK} between locally covariant AQFTs and
points of the Haag-Kastler $2$-functor adapts to the case where
all AQFTs satisfy their relevant time-slice axiom. Indeed, by direct inspection,
one verifies that the functors from Constructions \ref{constr:LCQFTtoHK} and \ref{constr:HKtoLCQFT} 
preserve the time-slice axioms, hence they induce functors 
\begin{flalign}\label{eqn:localizedHKvsLCQFT}
\AQFT(\ovr{\Loc})^W~\longrightarrow~\HK^W(\mathrm{pt})\quad,\qquad
\HK^W(\mathrm{pt})~\longrightarrow~\AQFT(\ovr{\Loc})^W
\end{flalign}
between the category $\AQFT(\ovr{\Loc})^W$ of locally covariant AQFTs satisfying the time-slice axiom
and the category of points $\HK^W(\mathrm{pt})$ of the time-sliced Haag-Kastler $2$-functor. 
As a consequence of Theorem \ref{theo:LCQFTvsHK}, we then obtain the following result.
\begin{cor}\label{cor:LCQFTvsHKW}
The two functors in \eqref{eqn:localizedHKvsLCQFT} exhibit an equivalence
\begin{flalign}
\HK^W(\mathrm{pt})~\simeq~\AQFT(\ovr{\Loc})^W
\end{flalign}
between the category of points $\HK^W(\mathrm{pt})$ of the time-sliced Haag-Kastler $2$-functor
and the category $\AQFT(\ovr{\Loc})^W$ of locally covariant AQFTs satisfying the time-slice axiom.
\end{cor}

\subsection{\label{subsec:HK2functorRC}The case of relatively compact causally convex opens}
In this subsection we describe the variants of the Haag-Kastler $2$-functor 
which are associated with Haag-Kastler-style AQFTs that are modeled on
the orthogonal categories $\ovr{\RC(M)}$ of relatively compact 
causally convex opens in $M\in\Loc$ from Example \ref{ex:OCat}.

\subsubsection{Definition and properties}
Analogously to the $2$-functor $\ovr{\COpen(-)} : \Loc\to\Cat^\perp$ 
from \eqref{eqn:COpenfunctor}, there exists a $2$-functor 
\begin{flalign}\label{eqn:RCfunctor}
\ovr{\RC(-)} \,:\, \Loc~\longrightarrow~\Cat^\perp\quad.
\end{flalign}
This $2$-functor assigns to each object $M\in\Loc$ the orthogonal category
$\ovr{\RC(M)}$ of relatively compact causally convex opens in $M$ 
and to each $\Loc$-morphism $f:M\to N$ the orthogonal
functor (denoted with abuse of notation by the same symbol $f$)
\begin{flalign}\label{eqn:RCfpullback}
f\,:=\,\ovr{\RC(f)} \,:\, \ovr{\RC(M)}~\longrightarrow~\ovr{\RC(N)}~~,\quad 
U\subseteq M ~\longmapsto~f(U)\subseteq N
\end{flalign}
which sends relatively compact causally convex opens in $M$ to their images under $f$ in $N$. 
\begin{defi}\label{def:RCHK2functor}
The \textit{relatively compact Haag-Kastler $2$-functor}
\begin{subequations}
\begin{flalign}
\HK^{\mathrm{rc}}\,:\, \Loc^\op~\longrightarrow~\CAT
\end{flalign}
is defined by assigning to each object $M\in\Loc$ the category
\begin{flalign}
\HK^{\mathrm{rc}}(M)\,:=\, \AQFT(\ovr{\RC(M)})\,\in\,\CAT
\end{flalign}
of relatively compact Haag-Kastler-style AQFTs on $M$ and 
to each $\Loc$-morphism $f : M\to N$ the pullback functor
\begin{flalign}
\begin{gathered}
\xymatrix@R=1em{
\ar@{=}[d]\HK^{\mathrm{rc}}(N) \ar[r]^-{\HK^{\mathrm{rc}}(f)\,:=\, f^\ast} ~&~ \HK^{\mathrm{rc}}(M) \ar@{=}[d]\\
\AQFT(\ovr{\RC(N)}) \ar[r]_-{f^\ast}~&~ \AQFT(\ovr{\RC(M)})
}
\end{gathered}
\end{flalign}
\end{subequations}
from \eqref{eqn:Fpullbackfunctor} which is 
associated to the orthogonal functor $f : \ovr{\RC(M)} \to \ovr{\RC(N)}$ in \eqref{eqn:RCfpullback}.
\end{defi}

The result of Proposition \ref{prop:HKnotastack} remains valid in the present case.
\begin{propo}\label{prop:RCHKnotastack}
Suppose that the category of algebras $\Alg_{\mathsf{uAs}}(\TT)$ has two objects
$A,B\in \Alg_{\mathsf{uAs}}(\TT)$ for which the $\Hom$-set $\Hom(A,B)$ is not a 
singleton. Then the relatively compact Haag-Kastler $2$-functor $\HK^{\mathrm{rc}}$ from Definition 
\ref{def:RCHK2functor} is not a stack with respect to 
either Grothendieck topology from Definition \ref{def:stack} on $\Loc$.
It is not even a prestack.
\end{propo}
\begin{proof}
The proof is very similar to the one of Proposition \ref{prop:HKtimeslicenotastack}.
The only difference is that we assume that $M\in \Loc$ admits compact Cauchy surfaces.
Under this hypothesis, the two objects $\AAA,\BBB\in\HK^{\mathrm{rc}}(M)$ defined by
\begin{subequations}
\begin{flalign}
\AAA(U)\,&:=\, \begin{cases}
A &~,~~\text{if $U$ contains a Cauchy surface of $M$}\\
I &~,~~\text{else}
\end{cases}\quad,\\
\BBB(U)\,&:=\, \begin{cases}
B &~,~~\text{if $U$ contains a Cauchy surface of $M$}\\
I &~,~~\text{else}
\end{cases}\quad,
\end{flalign}
\end{subequations}
for all relatively compact causally convex opens $U\subseteq M$,
differ from the initial object in $\HK^{\mathrm{rc}}(M)$ 
since there exist relatively compact causally convex opens $U\subseteq M$
which contain a Cauchy surface of $M$. (Consider for example time-slabs with a bounded time interval.)
As in the proof of Proposition \ref{prop:HKnotastack},
one then shows that the canonical functor $\HK^{\mathrm{rc}}(M) \to \HK^{\mathrm{rc}}(\mathcal{U})$ 
to the descent category is not fully faithful by using these $\AAA$ and $\BBB$.
\end{proof}

\subsubsection{\label{subsubsec:additivity}Comparison to additivity properties}
The canonical full orthogonal subcategory inclusions $i_M :\ovr{\RC(M)}\to \ovr{\COpen(M)}$
from Example \ref{ex:RCinCOpen},
for all $M\in\Loc$, assemble into a (strict) $2$-natural transformation
\begin{flalign}
i\,:\,\ovr{\RC(-)}~\Longrightarrow~\ovr{\COpen(-)}
\end{flalign}
between the $2$-functors defined in \eqref{eqn:RCfunctor} and \eqref{eqn:COpenfunctor}. This induces
via object-wise pullback a $2$-natural transformation
\begin{flalign}\label{eqn:ipullback}
i^\ast\,:\, \HK~\Longrightarrow~\HK^{\mathrm{rc}}
\end{flalign}
which allows us to compare the Haag-Kastler $2$-functor from Definition \ref{def:HK2functor}
with the relatively compact Haag-Kastler $2$-functor from Definition \ref{def:RCHK2functor}.

\paragraph{Additivity in Haag-Kastler-style AQFTs:}
It was shown in Example \ref{ex:RCinCOpen} that, restricting to any fixed object $M\in \Loc$,
the component
\begin{flalign}\label{eqn:HKrcepsilon}
i_M^\ast \,:\, \HK^{\epsilon\mathrm{-iso}}(M)~\stackrel{\simeq}{\longrightarrow}~\HK^{\mathrm{rc}}(M)
\end{flalign}
of the $2$-natural transformation \eqref{eqn:ipullback}
defines an equivalence between the category of relatively compact Haag-Kastler-style AQFTs 
on $M$ and the full subcategory $ \HK^{\epsilon\mathrm{-iso}}(M)\subseteq \HK(M)$
consisting of all Haag-Kastler-style AQFTs $\AAA\in \HK(M)$
which satisfy the property that the counit $\epsilon_{\AAA}$ of 
the adjunction $i_{M\,!}\dashv i_M^\ast$ in \eqref{eqn:iMadjunction} is an isomorphism.
Moreover, \eqref{eqn:epsilonadd} identifies this property
as an additivity property of the Haag-Kastler-style AQFT $\AAA\in \HK(M)$ on $M$.
\sk

This object-wise identification between $\HK^{\mathrm{rc}}(M)$ and
the full subcategory $ \HK^{\epsilon\mathrm{-iso}}(M)\subseteq \HK(M)$
however fails to extend to the level of $2$-functors because
the family of full subcategories 
$\HK^{\epsilon\mathrm{-iso}}(M)\subseteq \HK(M)$, for all $M\in\Loc$,
does \textit{not} form a $2$-subfunctor of the Haag-Kastler $2$-functor $\HK$.
\begin{propo}\label{prop:HKepsilonNOTsubfunctor}
Suppose that there exists an algebra $A\in \Alg_{\mathsf{uAs}}(\TT)$
which is not isomorphic to the initial algebra.\footnote{This technical
condition is very mild and it is only used to rule out pathological examples
of $\TT$. In particular, it holds true for $\TT=\Vec_\bbK$.} 
Then, for every $\Loc$-morphism $f:M\to N$ whose image $f(M)\subseteq N$
is relatively compact, the pullback functor
$f^\ast : \HK(N)\to\HK(M)$ of the Haag-Kastler $2$-functor 
does not restrict to a functor
between $\HK^{\epsilon\mathrm{-iso}}(N)$ and $\HK^{\epsilon\mathrm{-iso}}(M)$.
\end{propo}
\begin{proof}
Our proof strategy is to present an explicit example
for an object $\AAA\in \HK^{\epsilon\mathrm{-iso}}(N)$ whose pullback
$f^\ast(\AAA)\in \HK(M)$ does not lie in the full subcategory $\HK^{\epsilon\mathrm{-iso}}(M)\subseteq \HK(M)$.
Let us define
\begin{flalign}
\AAA(V)\,=\,\begin{cases}
A &~,~~\text{if } f(M)\subseteq V\quad,\\
I &~,~~\text{if }f(M)\not\subseteq V\quad,
\end{cases}
\end{flalign}
for all $V\in\COpen(N)$, where $A\in \Alg_{\mathsf{uAs}}(\TT)$ denotes the
algebra from our hypotheses and $I\in \Alg_{\mathsf{uAs}}(\TT)$ denotes the initial algebra.
We endow $\AAA\in\HK(N)$ with the AQFT structure which is defined
by the universal property of the initial algebra
and the identity morphism $\id_A : A\to A$. (Note that the $\perp$-commutativity axiom 
follows from the fact that $(V_1\subseteq V)\perp (V_2\subseteq V)$ implies
that $V_1\cap V_2 = \varnothing$, hence $f(M)$ can not be contained 
in both $V_1$ and $V_2$.) To show that 
$\AAA\in\HK^{\epsilon\mathrm{-iso}}(N)\subseteq \HK(N)$, recall \eqref{eqn:epsilonadd}
and consider the components of the counit 
\begin{flalign}
(\epsilon_{\AAA})_V\,:\, \colim\Big(\RC(N)/V\longrightarrow\COpen(N)\stackrel{\AAA}{\longrightarrow}\Alg_{\mathsf{uAs}}(\TT)\Big)~\longrightarrow~\AAA(V)\quad,
\end{flalign}
for all $V\in\COpen(N)$. In the case where $f(M)\not\subseteq V$,
the restriction of $\AAA$ to the comma category yields the constant functor assigning $I$,
hence we obtain an isomorphism. In the case where $f(M)\subseteq V$, we 
use that $f(M)\subseteq V$ defines an object in the comma category $\RC(N)/V$ 
since the image of $f$ is by our hypotheses relatively compact in $N$.
Using further that $\AAA$ is constantly
assigning $A$ to all $\tilde{V}\supseteq f(M)$, we obtain again an isomorphism.
\sk

It remains to show that $f^\ast(\AAA)\in\HK(M)$ does not define an object in 
$\HK^{\epsilon\mathrm{-iso}}(M)\subseteq \HK(M)$. For this we consider
the component of the counit 
\begin{flalign}
(\epsilon_{f^\ast(\AAA)})_M\,:\, \colim\Big(\RC(M)/M \longrightarrow\COpen(M)\stackrel{f^\ast(\AAA)}{\longrightarrow}\Alg_{\mathsf{uAs}}(\TT)\Big)~\longrightarrow~\AAA(f(M))\,=\,A
\end{flalign}
on the terminal object $M\in\COpen(M)$. Note that the restriction of the functor $f^\ast(\AAA)$ 
to the comma category 
$\RC(M)/M $ yields the constant functor assigning $I$
because $M\subseteq M$ is not relatively compact for a (necessarily non-compact)
globally hyperbolic Lorentzian manifold. Hence, $\epsilon_{f^\ast(\AAA)}$ is not an isomorphism. 
\end{proof}

The implication of this result is that the \textit{structure}
of the relatively compact Haag-Kastler $2$-functor $\HK^{\mathrm{rc}}$
can not be encoded in terms of a \textit{property} of the Haag-Kastler $2$-functor
$\HK$, despite the object-wise equivalence $\HK^{\mathrm{rc}}(M) \simeq \HK^{\epsilon\mathrm{-iso}}(M)$
of \eqref{eqn:HKrcepsilon} with Haag-Kastler-style AQFTs satisfying a particular additivity property.
This makes the relatively compact Haag-Kastler $2$-functor
$\HK^{\mathrm{rc}}$ a genuinely new concept.

\paragraph{Additivity in locally covariant AQFTs:}
We next compare the relatively
compact Haag-Kastler $2$-functor $\HK^{\mathrm{rc}}$ with
an additivity property used in the context of locally covariant AQFTs, see e.g.\ \cite[Definition 2.16]{BPS}.
\begin{defi}\label{def:additivity}
For every object $M\in\Loc$, we denote by $\HK^{\mathrm{add}}(M)\subseteq \HK(M)$
the full subcategory of the category of Haag-Kastler-style AQFTs on
$M$ consisting of all objects $\AAA\in \HK(M)$ which satisfy the following
\textit{locally covariant additivity property}: For every $U\in \COpen(M)$, the canonical map
\begin{flalign}\label{eqn:additivity}
\colim\Big(\RC(U)\stackrel{\subseteq}{\longrightarrow}\COpen(M)\stackrel{\AAA}{\longrightarrow}\Alg_{\mathsf{uAs}}(\TT)\Big)~\stackrel{\cong}{\longrightarrow}~\AAA(U)
\end{flalign}
is an isomorphism in $\Alg_{\mathsf{uAs}}(\TT)$.
\end{defi}

This will be related to an additivity property on $\AQFT(\ovr{\Loc})$ below, justifying the name.
\begin{rem}
We would like to highlight a subtle but important
difference between the locally covariant additivity property \eqref{eqn:additivity}
and the $\epsilon$-iso property \eqref{eqn:epsilonadd} from Example \ref{ex:RCinCOpen},
which we rewrite here for comparison
\begin{flalign}\label{eqn:epsilonaddx}
(\epsilon_\AAA)_U\,:\,\colim\Big(
\RC(M)/U \longrightarrow \COpen(M) \stackrel{\AAA}{\longrightarrow} \Alg_{\mathsf{uAs}}(\TT)
\Big) ~\stackrel{\cong}{\longrightarrow}~\AAA(U)\quad.
\end{flalign}
The colimit in the locally covariant additivity property is indexed over
the category $\RC(U)$ of all relatively compact causally convex
opens in $U$, while the colimit in the $\epsilon$-iso property is indexed over the comma category
$\RC(M)/U$ of all relatively compact causally convex opens in $M$ which are also contained in $U\subseteq M$.
Note that these two categories are in general very different. For example,
if $U\subseteq M$ is a relatively compact causally convex open subset, then 
the comma category $\RC(M)/U$ has a terminal object $U\subseteq U$, while the category
$\RC(U)$ never has a terminal object because $U$ is a (necessarily non-compact) globally hyperbolic
Lorentzian manifold. In simpler words, this means that relative compactness
is a relative condition which is sensitive to the ambient manifold
in which one considers open subsets. For the $\epsilon$-iso property \eqref{eqn:epsilonaddx},
the relevant relative compactness condition $\RC(M)/U$ is formulated relative to the 
ambient manifold $M$ itself, while for the locally covariant additivity property \eqref{eqn:additivity} the
relative compactness condition $\RC(U)$ is formulated intrinsically relative to
the submanifold $U\subseteq M$.
\end{rem}

In stark contrast to Proposition \ref{prop:HKepsilonNOTsubfunctor},
the family of full subcategories $\HK^{\mathrm{add}}(M)\subseteq \HK(M)$, for all $M\in\Loc$,
forms a $2$-subfunctor of the Haag-Kastler $2$-functor $\HK$.
\begin{propo}
For every $\Loc$-morphism $f:M\to N$, the pullback functor
$f^\ast : \HK(N)\to\HK(M)$ of the Haag-Kastler $2$-functor 
restricts to a functor $f^\ast : \HK^{\mathrm{add}}(N)\to\HK^{\mathrm{add}}(M)$
between the full subcategories of locally covariantly additive objects.
This defines a $2$-subfunctor $\HK^{\mathrm{add}}\subseteq \HK$.
\end{propo}
\begin{proof}
We have to show that, given any locally covariantly 
additive object $\AAA\in\HK^{\mathrm{add}}(N)\subseteq \HK(N)$ on $N$,
the pullback $f^\ast(\AAA)\in \HK(M)$ satisfies the
locally covariant additivity property from Definition \ref{def:additivity}.
This follows from the direct calculation
\begin{flalign}
\nn &\colim\Big(\RC(U) \stackrel{\subseteq}{\longrightarrow}\COpen(M)\stackrel{f^\ast(\AAA)}{\longrightarrow}\Alg_{\mathsf{uAs}}(\TT)\Big) \\
\nn &\qquad= \colim\Big(\RC(U) \stackrel{\subseteq}{\longrightarrow}\COpen(M)\stackrel{f}{\longrightarrow}
\COpen(N)\stackrel{\AAA}{\longrightarrow}\Alg_{\mathsf{uAs}}(\TT)\Big) \\
&\qquad \cong \colim\Big(\RC(f(U)) \stackrel{\subseteq}{\longrightarrow}
\COpen(N)\stackrel{\AAA}{\longrightarrow}\Alg_{\mathsf{uAs}}(\TT)\Big)\,\cong\, \AAA(f(U))\,=\,f^\ast(\AAA)(U)
\quad,
\end{flalign}
for all $U\in\COpen(M)$. In the first step we used the definition of the pullback functor $f^\ast(\AAA) = \AAA\,f$.
In the second step we used the commutative diagram
\begin{flalign}
\begin{gathered}
\xymatrix{
\ar[d]_-{f\vert_U}^-{\cong}\RC(U) \ar[r]^-{\subseteq}~&~ \COpen(M)\ar[d]^-{f}\\
\RC(f(U))\ar[r]_-{\subseteq}~&~\COpen(N)
}
\end{gathered}\quad,
\end{flalign}
where $f\vert_U : U\to f(U)$ denotes the $\Loc$-isomorphism obtained by restricting and corestricting $f:M\to N$.
The last two steps follow from the locally covariant additivity property
of $\AAA\in\HK^{\mathrm{add}}(N)\subseteq \HK(N)$
and using the definition of the pullback functor $f^\ast(\AAA) = \AAA\,f$ once more.
\end{proof}

We can now compare the relatively compact Haag-Kastler $2$-functor
$\HK^{\mathrm{rc}}$ with the
locally covariantly additive $2$-subfunctor $\HK^{\mathrm{add}}\subseteq \HK$
by restricting the $2$-natural transformation \eqref{eqn:ipullback} to
\begin{flalign}\label{eqn:ipullbackadd}
i^\ast\,:\,\HK^{\mathrm{add}}~\Longrightarrow~\HK^{\mathrm{rc}}\quad.
\end{flalign}
As a consequence of Proposition \ref{prop:HKepsilonNOTsubfunctor},
we already know that this $2$-natural transformation can \textit{not} be an equivalence
of $2$-functors, so the relative compactness structure of $\HK^{\mathrm{rc}}$
and the locally covariant additivity property of $\HK^{\mathrm{add}}\subseteq \HK$
are \textit{not} equivalent.
However, we have the following result which implies that the
locally covariant additivity property is stronger than the relative compactness structure.
\begin{theo}\label{theo:rcvsadd}
For every object $M\in\Loc$, the component $i^\ast_M : \HK^{\mathrm{add}}(M)\to \HK^{\mathrm{rc}}(M)$
of the $2$-natural transformation \eqref{eqn:ipullbackadd} is a fully faithful functor. 
Hence, by taking essential images, one can present the
locally covariantly additive Haag-Kastler $2$-functor
$\HK^{\mathrm{add}}\subseteq\HK$ as a $2$-subfunctor of the relatively compact 
Haag-Kastler $2$-functor $\HK^{\mathrm{rc}}$.
\end{theo}
\begin{proof}
Using the equivalence $i_M^\ast : \HK^{\epsilon\mathrm{-iso}}(M)\stackrel{\simeq}{\longrightarrow}
\HK^{\mathrm{rc}}(M)$ from \eqref{eqn:HKrcepsilon}, see also Example \ref{ex:RCinCOpen},
it suffices to show that $\HK^{\mathrm{add}}(M)\subseteq \HK^{\epsilon\mathrm{-iso}}(M)$
is a full subcategory. In other words, we have to show that every
locally covariantly additive
object $\AAA\in\HK^{\mathrm{add}}(M)\subseteq \HK(M)$ satisfies the $\epsilon$-iso
property \eqref{eqn:epsilonadd} from Example \ref{ex:RCinCOpen}, i.e.\ the canonical map
\begin{flalign}
\colim\Big(
\RC(M)/U \longrightarrow \COpen(M) \stackrel{\AAA}{\longrightarrow} \Alg_{\mathsf{uAs}}(\TT)
\Big) ~\longrightarrow~\AAA(U)
\end{flalign}
is an isomorphism in $\Alg_{\mathsf{uAs}}(\TT)$, for all $U\in\COpen (M)$.
Using that $\AAA$ is additive in the sense of Definition \ref{def:additivity},
we can rewrite the source of the canonical map as an iterated colimit
\begin{flalign}
\nn &\colim\Big(
\RC(M)/U \longrightarrow \COpen(M) \stackrel{\AAA}{\longrightarrow} \Alg_{\mathsf{uAs}}(\TT)
\Big)
\, =\, \colim_{(U^\prime\subseteq U)\in \RC(M)/U}^{~}\big(\AAA(U^\prime)\big)\\[4pt]
\nn &\qquad\qquad \,\cong\, \colim_{(U^\prime\subseteq U)\in \RC(M)/U}^{}~
\colim_{U^{\prime\prime}\in \RC(U^\prime)}^{}\big(\AAA(U^{\prime\prime})\big) \\[4pt]
&\qquad\qquad \,\cong\, \colim\Big(\textstyle{\int_U\RC}
\stackrel{\pi}{\longrightarrow} \RC(U)
\stackrel{\subseteq}{\longrightarrow}\COpen(M)
\stackrel{\AAA}{\longrightarrow} \Alg_{\mathsf{uAs}}(\TT)\Big)
\quad,\label{eqn:GrothendieckconstructionTMP}
\end{flalign}
hence as a colimit over the Grothendieck construction $\int_U\RC$ of the $2$-functor
$\RC : \RC(M)/U\to \Cat\,,~(U^\prime\subseteq U)\mapsto \RC(U^\prime)$.
Similarly, the target of the canonical map can be rewritten by using
locally covariant additivity as a colimit
\begin{flalign}
\AAA(U)\,\cong\,\colim\Big(\RC(U)\stackrel{\subseteq}{\longrightarrow} \COpen(M)\stackrel{\AAA}{\longrightarrow}
\Alg_{\mathsf{uAs}}(\TT)\Big)\quad.
\end{flalign}
The problem then reduces to proving that the functor 
$\pi : \int_U\RC\to \RC(U)$ in \eqref{eqn:GrothendieckconstructionTMP} is final.
\sk

An explicit model for the relevant Grothendieck construction $\int_U\RC$
is given by the category whose objects are pairs $(V,W)$
with $V\in\RC(M)/U$ a relatively compact causally convex open in $M$
such that $V\subseteq U$ and $W\in\RC(V)$
a relatively compact causally convex open in $V$. There exists a unique
morphism $(V,W)\to (\widetilde{V},\widetilde{W})$
if and only if $V\subseteq\widetilde{V} $ and $W\subseteq \widetilde{W}$.
The functor $\pi : \int_U\RC\to \RC(U)$ then sends an object $(V,W)$ to $W$
and the unique morphism $(V,W)\to (\widetilde{V},\widetilde{W})$
in $\int_U\RC$ to the subset inclusion $W\subseteq \widetilde{W}$ in $\RC(U)$.
\sk

Recall that the functor $\pi : \int_U\RC\to \RC(U)$ is final if, for each $U^\prime \in \RC(U)$,
the comma category $U^\prime / \pi$ is non-empty and connected.
An object of $U^\prime / \pi$ is an object $(V,W) \in \int_U \RC$ such that $U^\prime \subseteq W$ in $\RC(U)$,
while a morphism $(V,W) \to (\widetilde{V}, \widetilde{W})$ in $U^\prime / \pi$ exists
if and only if the underlying $\int_U \RC$-morphism exists,
i.e.\ if and only if $V \subseteq \widetilde{V}$ and $W \subseteq \widetilde{W}$.
\sk

Let us show that $U^\prime / \pi$ is non-empty.
Since $U^\prime \in \RC(U)$, its closure $\mathrm{cl}(U^\prime)\subseteq U$ 
with respect to $U$ is compact.
For each point $p \in \mathrm{cl}(U^\prime)$,
global hyperbolicity and hence strong causality of $M$
entails the existence
of a relatively compact causally convex open neighborhood $V_p \subseteq U$ of $p$.
Since $\left\{V_p\right\}_{p \in \mathrm{cl}(U^\prime)}$ 
is an open cover of the compact subset $\mathrm{cl}(U^\prime) \subseteq U$, 
it admits a finite subcover $\left\{V_{p_1}, \ldots, V_{p_n}\right\}$.
Define $V := J^{+}_U \left(\cup_{i=1}^n V_{p_i}\right) \cap J^{-}_U \left(\cup_{i=1}^n V_{p_i}\right)\subseteq U$ 
as the causally convex hull of the finite subcover. 
Using Lemma \ref{lem:hullofrelativelycompact}, we obtain that
$V \subseteq U \subseteq M$ is a relatively compact 
causally convex open subset, hence $V$ is an object of $\RC(M)/U$. 
Furthermore, $\mathrm{cl}(U^\prime)\subseteq V$ by construction, 
hence $(V, U^\prime) \in U^\prime / \pi$.
\sk

Let us now show that $U^\prime / \pi$ is connected.
Take any two objects $(V_1,W_1)$ and $(V_2,W_2)$ of $U^\prime / \pi$.
We exhibit an object $(V,W)$ and morphisms
$(V_1,W_1) \leftarrow (V,W) \rightarrow (V_2,W_2)$ in $U^\prime / \pi$.
Define the relatively compact causally convex opens 
$V := V_1 \cap V_2 \subseteq M$ and $W := W_1 \cap W_2 \subseteq V$. 
(To confirm that $W$ is relatively compact in $V$, recall that the 
closures of $W_i$ with respect to $V_i$ are compact, 
hence the closure of $W$ with respect to $V$ is compact too.) 
Then $U^\prime \subseteq W_i$ entails that $U^\prime \subseteq W$,
hence $(V,W) \in U^\prime / \pi$. The morphisms 
$(V,W) \to (V_i, W_i)$ in $U^\prime / \pi$ exist because 
$V \subseteq V_i$ and $W \subseteq W_i$. 
\end{proof}

We conclude this subsection with a comment on the relationship between additive 
locally covariant AQFTs and the relatively compact Haag-Kastler $2$-functor. 
The usual definition of the full subcategory
$\AQFT(\ovr{\Loc})^{\mathrm{add}}\subseteq \AQFT(\ovr{\Loc})$
of additive locally covariant AQFTs, see e.g.\ \cite[Definition 2.16]{BPS},
can be equivalently rephrased through the equivalence from Theorem \ref{theo:LCQFTvsHK}
as follows: An object $\AAA\in \AQFT(\ovr{\Loc})$ is additive if and only if 
its underlying Haag-Kastler-style AQFTs $k_M^\ast(\AAA)\in\HK(M)$ are locally covariantly additive
in the sense of Definition \ref{def:additivity}, for all $M\in\Loc$.
This implies that Theorem \ref{theo:LCQFTvsHK} restricts to an equivalence
\begin{flalign}\label{eqn:HKaddvsLCQFTadd}
\HK^{\mathrm{add}}(\mathrm{pt})\,\simeq\,\AQFT(\ovr{\Loc})^{\mathrm{add}}
\end{flalign}
between the category of points of
the locally covariantly additive Haag-Kastler $2$-subfunctor $\HK^{\mathrm{add}}\subseteq\HK$
and the category $\AQFT(\ovr{\Loc})^{\mathrm{add}}$ of additive locally covariant AQFTs.
Together with Theorem \ref{theo:rcvsadd}, this yields the following result.
\begin{cor}\label{cor:additiveLCQFT}
The category $\AQFT(\ovr{\Loc})^{\mathrm{add}}$ of additive locally covariant AQFTs
is presented via the fully faithful functor
\begin{flalign}
\AQFT(\ovr{\Loc})^{\mathrm{add}} ~\simeq~\HK^{\mathrm{add}}(\mathrm{pt})~\longrightarrow~\HK^{\mathrm{rc}}(\mathrm{pt})
\end{flalign}
as a full subcategory of the category of points $\HK^{\mathrm{rc}}(\mathrm{pt})$ 
of the relatively compact Haag-Kastler $2$-functor,
where the last functor is obtained by inducing \eqref{eqn:ipullbackadd} to the categories of points.
\end{cor}

\subsubsection{Time-slice axiom} 
This subsection contains a brief study of the time-slice axiom in the context of the 
relatively compact Haag-Kastler $2$-functor. 
The following definition is analogous to Definition \ref{def:HK2functortimeslice}.
\begin{defi}\label{def:RCHK2functortimeslice}
The \textit{time-sliced relatively compact Haag-Kastler $2$-functor}
is defined as the $2$-subfunctor $\HK^{\mathrm{rc},W}\subseteq \HK^{\mathrm{rc}}$
of the relatively compact Haag-Kastler $2$-functor from Definition \ref{def:RCHK2functor}
which assigns to every $M\in\Loc$ the full subcategory 
$\HK^{\mathrm{rc},W}(M)\subseteq \HK^{\mathrm{rc}}(M)$ 
consisting of all relatively compact Haag-Kastler-style AQFTs on $M$ which satisfy the time-slice axiom.
\end{defi}

\begin{rem}\label{rem:RCHK2functortimeslice}
Similarly to Remark \ref{rem:HK2functortimeslice},
we have the following equivalent model for the time-sliced relatively compact Haag-Kastler $2$-functor.
The assignment of the localized orthogonal categories 
$M\mapsto \ovr{\RC(M)}[W_{\mathrm{rc},M}^{-1}]$ from Example \ref{ex:localizations} 
is $2$-functorial
\begin{subequations}\label{eqn:RCfunctorW}
\begin{flalign}
\ovr{\RC(-)}[W_{\mathrm{rc},(-)}^{-1}]\,:\,\Loc~\longrightarrow~\Cat^\perp
\end{flalign}
with action on $\Loc$-morphisms $f:M\to N$ given by
\begin{flalign}\label{eqn:RCfunctorWf}
f_{W}\,:=\,\ovr{\RC(f)}[W_{\mathrm{rc},f}^{-1}] \,:\, \ovr{\RC(M)}[W_{\mathrm{rc},M}^{-1}] ~&\longrightarrow~\ovr{\RC(N)}[W_{\mathrm{rc},N}^{-1}]
\quad,\\
\nn U\subseteq M ~&\longmapsto~f(U)\subseteq N
\quad, \\
\nn (U\to V) ~&\longmapsto~ \big(f(U)\to f(V)\big)
\quad,
\end{flalign}
\end{subequations}
see also Appendix \ref{app:localization}.
Replacing in Definition \ref{def:RCHK2functor} the $2$-functor $\ovr{\RC(-)}$ by
$\ovr{\RC(-)}[W_{\mathrm{rc},(-)}^{-1}]$, one obtains an equivalent model for the time-sliced
relatively compact Haag-Kastler $2$-functor which, with abuse of notation, 
we denote by the same symbol $\HK^{\mathrm{rc},W} : \Loc^\op\to \CAT$. Explicitly, 
this $2$-functor assigns to an object $M\in\Loc$ the category 
\begin{subequations}
\begin{flalign}
\HK^{\mathrm{rc},W}(M)\,=\,\AQFT\big(\ovr{\RC(M)}[W_{\mathrm{rc},M}^{-1}]\big)\,\in\,\CAT
\end{flalign}
of AQFTs on the orthogonal localization $\ovr{\RC(M)}[W_{\mathrm{rc},M}^{-1}]$ 
and to a $\Loc$-morphism $f:M\to N$ the pullback functor
\begin{flalign}
\begin{gathered}
\xymatrix@R=1em{
\ar@{=}[d]\HK^{\mathrm{rc},W}(N) \ar[r]^-{\HK^{\mathrm{rc},W}(f)\,:=\, f_{W}^\ast} ~&~ \HK^{\mathrm{rc},W}(M) \ar@{=}[d]\\
\AQFT\big(\ovr{\RC(N)}[W_{\mathrm{rc},N}^{-1}]\big) \ar[r]_-{f_{W}^\ast}~&~ 
\AQFT\big(\ovr{\RC(M)}[W_{\mathrm{rc},M}^{-1}]\big)
}
\end{gathered}
\end{flalign}
\end{subequations}
associated to the orthogonal functor \eqref{eqn:RCfunctorWf}.
The equivalence between this model and the one in Definition \ref{def:RCHK2functortimeslice}
is implemented as in Proposition \ref{prop:timeslice} by pullbacks along the orthogonal
localization functors $L_{\mathrm{rc},M} : \ovr{\RC(M)}\to \ovr{\RC(M)}[W_{\mathrm{rc},M}^{-1}]$, for all $M\in\Loc$.
\end{rem}

The result of Proposition \ref{prop:RCHKnotastack} remains valid in the present case.
\begin{propo}\label{prop:RCHKtimeslicenotastack}
Suppose that the category of algebras $\Alg_{\mathsf{uAs}}(\TT)$ has two objects
$A,B\in \Alg_{\mathsf{uAs}}(\TT)$ for which the $\Hom$-set $\Hom(A,B)$ is not a 
singleton. Then the time-sliced relatively compact Haag-Kastler $2$-functor 
$\HK^{\mathrm{rc},W}$ from Definition 
\ref{def:RCHK2functortimeslice} is not a stack with respect to 
either Grothendieck topology from Definition \ref{def:stack} on $\Loc$.
It is not even a prestack.
\end{propo}
\begin{proof}
Note that the two objects $\AAA$ and $\BBB$ constructed in the proof of Proposition 
\ref{prop:RCHKnotastack} satisfy the time-slice axiom, hence that proof applies to the
present case without any alterations.
\end{proof}

We conclude by adapting the results of Theorem \ref{theo:rcvsadd}
and Corollary \ref{cor:additiveLCQFT} to the case where the time-slice axiom
is implemented. For this we observe that the $2$-natural transformation \eqref{eqn:ipullbackadd}
restricts to a $2$-natural transformation
\begin{flalign}\label{eqn:Wipullbackadd}
i^\ast\,:\, \HK^{\mathrm{add},W}~\Longrightarrow~\HK^{\mathrm{rc},W}
\end{flalign}
between the $2$-subfunctor $\HK^{\mathrm{add},W}\subseteq \HK^{\mathrm{add}}\subseteq \HK$,
which implements both the locally covariant additivity property and the time-slice axiom,
and the time-sliced relatively compact Haag-Kastler $2$-functor $\HK^{\mathrm{rc},W}$
from Definition \ref{def:RCHK2functortimeslice}. This is a consequence of the fact
that the orthogonal functor $i_M : \ovr{\RC(M)}\to\ovr{\COpen(M)}$ sends
Cauchy morphisms in $\ovr{\RC(M)}$ to Cauchy morphisms in $\ovr{\COpen(M)}$, for all $M\in\Loc$.
\begin{theo}\label{theo:Wrcvsadd}
For every object $M\in\Loc$, the component $i^\ast_M : \HK^{\mathrm{add},W}(M)\to \HK^{\mathrm{rc},W}(M)$
of the $2$-natural transformation \eqref{eqn:Wipullbackadd} is a fully faithful functor. 
Hence, by taking essential images,
one can present the time-sliced and locally covariantly additive Haag-Kastler $2$-functor
$\HK^{\mathrm{add},W}\subseteq \HK$ as a $2$-subfunctor of the time-sliced relatively 
compact Haag-Kastler $2$-functor $\HK^{\mathrm{rc},W}$.
\end{theo}
\begin{proof}
Recall from Theorem \ref{theo:rcvsadd} that 
$i^\ast_M : \HK^{\mathrm{add}}(M)\to \HK^{\mathrm{rc}}(M)$ 
is a fully faithful functor, hence its restriction to the full subcategories 
$\HK^{\mathrm{add},W}(M)\subseteq \HK^{\mathrm{add}}(M)$ 
and $\HK^{\mathrm{rc},W}(M)\subseteq \HK^{\mathrm{rc}}(M)$
is fully faithful too.
\end{proof}

Recalling also Corollary \ref{cor:LCQFTvsHKW}, we then obtain the following result.
\begin{cor}\label{cor:WadditiveLCQFT}
The category $\AQFT(\ovr{\Loc})^{\mathrm{add},W}$ 
of additive locally covariant AQFTs satisfying the time-slice axiom 
is presented via the fully faithful functor
\begin{flalign}
\AQFT(\ovr{\Loc})^{\mathrm{add},W} ~\simeq~\HK^{\mathrm{add},W}(\mathrm{pt})~\longrightarrow~\HK^{\mathrm{rc},W}(\mathrm{pt})
\end{flalign}
as a full subcategory of the category of points $\HK^{\mathrm{rc},W}(\mathrm{pt})$
of the time-sliced relatively compact Haag-Kastler $2$-functor,
where the last functor is obtained by inducing \eqref{eqn:Wipullbackadd} to the categories of points.
\end{cor}

%%%%%%%%%%%%%%%%%%%%%%%%%%%%%%%%%%%%%%%%%%%%%%%%
%%%%%%%%%%%%%%%%%%%%%%%%%%%%%%%%%%%%%%%%%%%%%%%%

\section{\label{sec:HKstack}Haag-Kastler stacks}
All variants $\HK$, $\HK^W$, $\HK^{\mathrm{rc}}$ and
$\HK^{\mathrm{rc},W}$ of the Haag-Kastler $2$-functor we have encountered
in the previous section (see Definitions \ref{def:HK2functor}, \ref{def:HK2functortimeslice}, 
\ref{def:RCHK2functor} and \ref{def:RCHK2functortimeslice})
have failed to satisfy the local-to-global (descent) properties which are described by
the concept of a stack, see Definition \ref{def:stack}. Even worse, we have shown in
Propositions \ref{prop:HKnotastack}, \ref{prop:HKtimeslicenotastack}, \ref{prop:RCHKnotastack}
and \ref{prop:RCHKtimeslicenotastack} that each of these $2$-functors is not even a prestack.
In this section we shall address and partially solve this issue 
by constructing from our original $2$-functors new ones
which, under certain hypotheses that hold true for the 
relatively compact examples $\HK^{\mathrm{rc}}$ and $\HK^{\mathrm{rc},W}$, 
enjoy better descent properties. Our construction is 
guided by leveraging very specific properties of AQFTs and the Haag-Kastler $2$-functors 
which arise by combining the techniques from Subsection \ref{subsec:AQFT}
with the theory of locally presentable categories from Subsection \ref{subsec:Pr}.

\subsection{\label{subsec:HKnotations}Preparations}
In order to streamline our arguments and to avoid unnecessary repetitions,
let us introduce the following abstract
notion of a Haag-Kastler $2$-functor.
\begin{defi}\label{def:HKabstract}
Let $\ovr{\CC(-)} : \Loc\to \Cat^\perp$ be a strict $2$-functor
to the $2$-category of orthogonal categories. The associated 
\textit{Haag-Kastler-style $2$-functor}
\begin{subequations}
\begin{flalign}
\HK_{\ovr{\CC}}\,:\,\Loc^\op~\longrightarrow~\CAT
\end{flalign}
is defined by assigning to an object $M\in \Loc$ the category
\begin{flalign}
\HK_{\ovr{\CC}}(M)\, :=\,\AQFT(\ovr{\CC(M)})\,\in\,\CAT
\end{flalign}
of AQFTs over $\ovr{\CC(M)}$,
and to a $\Loc$-morphism $f:M\to N$ the pullback functor
\begin{flalign}
\begin{gathered}
\xymatrix@R=1em{
\ar@{=}[d]\HK_{\ovr{\CC}}(N) \ar[r]^-{\HK_{\ovr{\CC}}(f)\,:=\, f^\ast} ~&~ \HK_{\ovr{\CC}}(M) \ar@{=}[d]\\
\AQFT(\ovr{\CC(N)}) \ar[r]_-{f^\ast}~&~ \AQFT(\ovr{\CC(M)})
}
\end{gathered}
\end{flalign}
\end{subequations}
from \eqref{eqn:Fpullbackfunctor} along the orthogonal functor $f:=\ovr{\CC(f)} : \ovr{\CC(M)}\to \ovr{\CC(N)}$.
\end{defi}

\begin{ex}\label{ex:HK2functors}
We observe that Definition \ref{def:HKabstract} covers
all the different variants of the Haag-Kastler $2$-functor from Section \ref{sec:HK2functor}.
\begin{itemize}
\item[(1)] The Haag-Kastler $2$-functor $\HK$ 
from Definition \ref{def:HK2functor} is associated to the $2$-functor
$\ovr{\COpen(-)}$ which assigns the orthogonal categories of causally convex opens.

\item[(2)] The time-sliced Haag-Kastler $2$-functor $\HK^W$
from Definition \ref{def:HK2functortimeslice} (see also Remark \ref{rem:HK2functortimeslice})
is associated to the $2$-functor $\ovr{\COpen(-)}[W_{(-)}^{-1}]$ which assigns
the orthogonal categories of causally convex opens localized at all Cauchy morphisms.

\item[(3)] The relatively compact Haag-Kastler $2$-functor $\HK^{\mathrm{rc}}$
from Definition \ref{def:RCHK2functor} is associated to the $2$-functor
$\ovr{\RC(-)}$ which assigns the orthogonal categories of relatively compact causally convex
opens.

\item[(4)] The time-sliced relatively compact Haag-Kastler $2$-functor $\HK^{\mathrm{rc},W}$
from Definition \ref{def:RCHK2functortimeslice} (see also Remark \ref{rem:RCHK2functortimeslice})
is associated to the $2$-functor
$\ovr{\RC(-)}[W_{\mathrm{rc},(-)}^{-1}]$ which assigns the orthogonal categories of
relatively compact causally convex opens localized at all Cauchy morphisms. \qedhere
\end{itemize}
\end{ex}

We will now provide a useful description of the descent category $\HK_{\ovr{\CC}}(\U)$
of the general Haag-Kastler-style $2$-functor from Definition \ref{def:HKabstract}
for a causally convex open cover $\U=\{U_i\subseteq M\}$ of an object $M\in\Loc$.
The key idea is to present this descent category in terms of an AQFT category over
the following orthogonal category associated with the cover.
\begin{defi}\label{def:CcoverOCat}
Let $\ovr{\CC(-)} : \Loc\to \Cat^\perp$ be a strict $2$-functor.
For each causally convex open cover $\U=\{U_i\subseteq M\}$ of an object $M\in\Loc$,
we define the category $\CC(\U)$ by the following generators and relations description:
\begin{itemize}
\item An object in $\CC(\U)$ is a pair $(i,U)$ consisting of an index $i$ of the cover
and an object $U\in \CC(U_i)$.

\item The morphisms in $\CC(\U)$ are generated by the following two types of generators:
\begin{subequations}
\begin{itemize}
\item[(i)] For every $i$ and every morphism $g : U\to U^\prime$ in $\CC(U_i)$, there exists a morphism
\begin{flalign}
(i,g)\,:\, (i,U)~\longrightarrow~(i,U^\prime)\quad.
\end{flalign}

\item[(ii)] For every $i,j$ such that $U_{ij}\neq \emptyset$ and every $V\in \CC(U_{ij})$, 
there exists a morphism
\begin{flalign}
\varphi_{ij,V}\,:\,\Big(j,\iota_{U_{ij}}^{U_j}(V)\Big)~\longrightarrow~\Big(i,\iota_{U_{ij}}^{U_i}(V)\Big)\quad,
\end{flalign}
where here and below we use the short-hand notations $\iota_{U_{ij}}^{U_j}(V)=\ovr{\CC\big(\iota_{U_{ij}}^{U_j}\big)}(V)$
and $\iota_{U_{ij}}^{U_i}(V)=\ovr{\CC\big(\iota_{U_{ij}}^{U_i}\big)}(V)$.
\end{itemize}
\end{subequations}
\noindent These generators are required to satisfy the following relations:
\begin{itemize}
\item[(r1)] For all $i$ and all composable morphisms $g:U\to U^\prime$ and $g^\prime: U^\prime\to U^{\prime\prime}$
in $\CC(U_i)$,
\begin{subequations}
\begin{flalign}
(i,g^\prime)\circ (i,g) \,=\, (i,g^\prime\,g)\quad.
\end{flalign}
Furthermore, for all $i$ and all $U\in \CC(U_i)$,
\begin{flalign}
(i,\id_U)\,=\, \id_{(i,U)}\quad.
\end{flalign}
\end{subequations}

\item[(r2)] For all $i,j$ with $U_{ij}\neq \emptyset$ and all morphisms
$h : V\to V^\prime$ in $\CC(U_{ij})$, the diagram
\begin{flalign}
\begin{gathered}
\xymatrix@C=3em@R=3em{
\ar[d]_-{\big(j,\, \iota_{U_{ij}}^{U_j}(h)\big)} 
\Big(j, \iota_{U_{ij}}^{U_j}(V)\Big)\ar[r]^-{\varphi_{ij,V}}~&~
\Big(i, \iota_{U_{ij}}^{U_i}(V)\Big) 
\ar[d]^-{\big(i,\, \iota_{U_{ij}}^{U_i}(h)\big)} \\
\Big(j, \iota_{U_{ij}}^{U_j}(V^\prime)\Big)\ar[r]_-{\varphi_{ij,V^\prime}}~&~
\Big(i, \iota_{U_{ij}}^{U_i}(V^\prime)\Big)
}
\end{gathered} 
\end{flalign}
commutes.

\item[(r3)] For all $i$ and all $U\in \CC(U_{ii}) = \CC(U_{i})$,
\begin{subequations}
\begin{flalign}
\varphi_{ii,U}\,=\,\id_{(i,U)}\quad.
\end{flalign}
Furthermore, for all $i,j,k$ with $U_{ijk}\neq \emptyset$ and all $W\in\CC(U_{ijk})$, the diagram
\begin{flalign}
\begin{gathered}
\xymatrix@C=6em@R=3em{
\ar[dr]_-{\varphi_{ik,\, \iota_{U_{ijk}}^{U_{ik}}(W)}}
\Big(k, \iota_{U_{ijk}}^{U_k}(W)\Big)
\ar[r]^-{\varphi_{jk,\, \iota_{U_{ijk}}^{U_{jk}}(W)}}~&~
\Big(j,\iota_{U_{ijk}}^{U_j}(W)\Big)
\ar[d]^-{\varphi_{ij,\,\iota_{U_{ijk}}^{U_{ij}}(W)}}\\
~&~\Big(i, \iota_{U_{ijk}}^{U_i}(W)\Big)
}
\end{gathered}
\end{flalign}
\end{subequations}
commutes.
\end{itemize}
\end{itemize}
We endow this category with the structure of an orthogonal category
$\ovr{\CC(\U)}$ by taking the smallest orthogonality relation
such that
\begin{flalign}
\big((i,g_1)\,:\,(i,U_1)\to (i,U)\big)\perp \big((i,g_2)\,:\,(i,U_2)\to (i,U)\big)
\end{flalign}
is orthogonal, for all $i$ and all orthogonal pairs $(g_1:U_1\to U)\perp (g_2:U_2\to U)$ in $\ovr{\CC(U_i)}$.
\end{defi}

For every $M\in\Loc$ and every causally convex open cover $\U=\{U_i\subseteq M\}$, 
we define an orthogonal functor
\begin{subequations}\label{eqn:jcoverfunctor}
\begin{flalign}
j_{\U}\,:\,\ovr{\CC(\U)}~\longrightarrow~\ovr{\CC(M)}
\end{flalign}
from the orthogonal category in Definition \ref{def:CcoverOCat} 
to the orthogonal category $\ovr{\CC(M)}$. This orthogonal functor maps an object $(i,U)$ to
\begin{flalign}
j_{\U}(i,U)\,:=\,\iota^{M}_{U_i}(U)\,:=\, \ovr{\CC\big(\iota^{M}_{U_i}\big)}(U)\quad,
\end{flalign}
a type (i) generating morphism $(i,g) : (i,U)\to (i,U^\prime)$ to
\begin{flalign}
j_{\U}(i,g)\,:=\, \iota^{M}_{U_i}(g)\,:=\,\ovr{\CC\big(\iota^{M}_{U_i}\big)}(g)\,:\, 
\iota^{M}_{U_i}(U)~\longrightarrow~\iota^{M}_{U_i}(U^\prime)\quad,
\end{flalign}
and a type (ii) generating morphism $\varphi_{ij,V} : \big(j,\iota_{U_{ij}}^{U_j}(V)\big)
\to \big(i,\iota_{U_{ij}}^{U_i}(V)\big)$ to the identity morphism
\begin{flalign}
j_{\U}(\varphi_{ij,V}) \,:=\,\id_{\iota_{U_{ij}}^{M}(V)}^{}\,:\, 
\iota_{U_{ij}}^{M}(V)~\longrightarrow~\iota_{U_{ij}}^{M}(V)\quad.
\end{flalign}
\end{subequations}
One directly checks that this assignment is compatible with the relations from Definition \ref{def:CcoverOCat}
and also that it preserves the orthogonality relations.
\begin{propo}\label{prop:descentvsAQFT}
For every $M\in\Loc$ and every causally convex open cover $\U=\{U_i\subseteq M\}$, 
there exists a canonical identification
\begin{flalign}
\HK_{\ovr{\CC}}(\U)\,\cong\, \AQFT(\ovr{\CC(\U)})
\end{flalign}
between the descent category of the Haag-Kastler-style $2$-functor $\HK_{\ovr{\CC}}$
and the category of AQFTs over the orthogonal category from 
Definition \ref{def:CcoverOCat}. Upon this identification, the canonical functor to
the descent category coincides with the pullback functor
\begin{flalign}
\begin{gathered}
\xymatrix@C=3em{
\ar@{=}[d]\HK_{\ovr{\CC}}(M) \ar[r]~&~\HK_{\ovr{\CC}}(\U)\ar[d]^-{\cong}\\
 \AQFT(\ovr{\CC(M)})\ar[r]_-{j_{\U}^\ast}~&~ \AQFT(\ovr{\CC(\U)})
}
\end{gathered}
\end{flalign}
along the orthogonal functor \eqref{eqn:jcoverfunctor}. 
Operadic left Kan extension as in Proposition \ref{propo:operadicLKE} then defines an adjunction
\begin{flalign}\label{eqn:descentadjunction}
\xymatrix@C=3em{
j_{\U\,!}\,:\, \HK_{\ovr{\CC}}(\U) \ar@<0.75ex>[r]~&~\ar@<0.75ex>[l] \HK_{\ovr{\CC}}(M)\,:\, j_{\U}^\ast
}
\end{flalign}
between the descent category $\HK_{\ovr{\CC}}(\U)$ and $\HK_{\ovr{\CC}}(M)$.
\end{propo}
\begin{proof}
The first statement follows directly by spelling out and comparing
the data and properties of a $\perp$-commutative
functor $\AAA : \CC(\U)\to \Alg_{\mathsf{uAs}}(\TT)$ from the orthogonal
category in Definition \ref{def:CcoverOCat} with the model for the descent category in
Remark \ref{rem:descentexplicit}. The canonical identification is then given by
$\AAA(i,U) =\AAA_i(U)$, for all $i$ and all objects $U\in\CC(U_i)$. The type (i) morphisms 
$(i,g)$ describe the structure of the local AQFTs $\AAA_i$ on $U_i$ in the descent category,
and the type (ii) morphisms $\varphi_{ij,V}$ describe the cocycle data in the descent category.
Note that, choosing triple overlaps of the form $U_{iji}$ and $U_{jij}$, 
the relations (r3) in Definition \ref{def:CcoverOCat} imply
that every type (ii) morphism is invertible, as required for the cocyles.
Through this identification, one directly verifies that the canonical functor to the descent
category coincides with $j_{\U}^\ast$.
\end{proof}

\begin{assu}\label{assu:LPCat}
Throughout the remainder of this section, we shall assume that 
the target symmetric monoidal category $\TT$ in which the AQFTs take values 
is \textit{locally presentable},
see also Section \ref{subsec:Pr}.
Let us recall from Example \ref{ex:locprescats} that the standard choice, given by the
category of vector spaces $\Vec_\bbK$ over a field $\bbK$, is of this kind.
\end{assu}

As a consequence of Example \ref{ex:locprescats} and Proposition \ref{propo:operadicLKE},
it follows from Assumption \ref{assu:LPCat} 
that the general Haag-Kastler-style $2$-functor from Definition \ref{def:HKabstract}
takes values in the $2$-subcategory $\Pr^{R}\subseteq \CAT$ from Definition \ref{def:PrL/R}, i.e.\
\begin{flalign}
\HK_{\ovr{\CC}}\,:\,\Loc^\op~\longrightarrow~\Pr^R\quad.
\end{flalign}
This provides by Corollary \ref{cor:adjointpseudofunctor} a new 
angle through which one can study the descent properties of $\HK_{\ovr{\CC}}$,
which is given by studying the codescent properties of the adjoint pseudo-functor
\begin{flalign}\label{eqn:HKdagger}
\HK_{\ovr{\CC}}^\dagger \,:\,\Loc~\longrightarrow~\Pr^L\quad.
\end{flalign}
Let us recall that the adjoint pseudo-functor assigns to each object
$M \in \Loc$ the same category $\HK_{\ovr{\CC}}^\dagger(M) := \HK_{\ovr{\CC}}(M)$
as assigned by $\HK_{\ovr{\CC}}$, and to each $\Loc$-morphism $f$ the left adjoint
$\HK_{\ovr{\CC}}^\dagger(f) := f_! \dashv f^\ast = \HK_{\ovr{\CC}}(f)$ of the pullback 
functor assigned by $\HK_{\ovr{\CC}}$.
The adjunction \eqref{eqn:descentadjunction} established in Proposition \ref{prop:descentvsAQFT}
gives an explicit and powerful model for both the canonical functor $j_{\U}^\ast : 
\HK_{\ovr{\CC}}(M)\to \HK_{\ovr{\CC}}(\U)$ to the descent category of 
$\HK_{\ovr{\CC}}$ and the canonical functor $j_{\U\,!} : 
\HK_{\ovr{\CC}}(\U)\to \HK_{\ovr{\CC}}(M)$ from the codescent category of
$\HK_{\ovr{\CC}}^\dagger$ in terms of pullback and operadic left Kan extension
along the orthogonal functor $j_{\U}$ from \eqref{eqn:jcoverfunctor}.
This explicit description will be a very useful ingredient for proving 
our results in the subsections below.
\begin{rem}\label{rem:Castpresentability}
For readers interested in $C^\ast$-algebraic AQFTs, we would like
to note that the corresponding Haag-Kastler-style $2$-functor
$C^\ast\HK_{\ovr{\CC}}:\Loc^\op \to \Pr^R$ in this context takes 
values in locally presentable categories and right adjoint functors too. 
Let us present the relevant arguments: Recall from
e.g.\ \cite{Castpresentable} that the category $C^\ast\Alg$ 
of $C^\ast$-algebras is locally ($\aleph_1$-)presentable,
hence the category of functors $\Fun(\CC,C^\ast\Alg)$
out of any small category $\CC$ is locally presentable by \cite[Corollary 1.54]{Adamek}.
The category $C^\ast\AQFT(\ovr{\CC})\subseteq \Fun(\CC,C^\ast\Alg)$
of $C^\ast$-algebraic AQFTs over an orthogonal category $\ovr{\CC}$
is defined as in Definition \ref{def:AQFT} in terms of 
the full subcategory consisting of all functors $\AAA : \CC\to C^\ast\Alg$ 
satisfying the $\perp$-commutativity axiom. From the usual
point-wise construction of limits and colimits in $\Fun(\CC,C^\ast\Alg)$, 
one directly checks that the full subcategory of $C^\ast$-algebraic AQFTs
is closed under limits and filtered colimits, as these preserve the
$\perp$-commutativity axiom, hence $C^\ast\AQFT(\ovr{\CC})$ 
is locally presentable by \cite{Adamek2}. It remains to show that, given any
orthogonal functor $F:\ovr{\CC}\to\ovr{\DD}$, the pullback functor
$F^\ast: C^\ast\AQFT(\ovr{\DD})\to C^\ast\AQFT(\ovr{\CC})$ is right adjoint.
This functor clearly preserves limits and filtered colimits as 
these are computed point-wise, hence it is right adjoint as a consequence
of the special adjoint functor theorem \cite[Theorem 1.66]{Adamek}.
This implies in particular that $C^\ast\HK_{\ovr{\CC}}:\Loc^\op \to \Pr^R$ takes values in 
$2$-subcategory $\Pr^{R}\subseteq \CAT$ from Definition \ref{def:PrL/R}.
\sk

The key difference between our algebraic context and the $C^\ast$-algebraic 
one is that we have available concrete models for the associated left adjoint functors 
in terms of operadic left Kan extensions, which we will use 
explicitly in our proofs in this section. In the $C^\ast$-algebraic case,
one knows a priori only about the existence of such left adjoints, which
introduces the additional challenge of finding useful models to carry out proofs. 
This means that further work is required to find out 
whether or not the results of this section generalize to $C^\ast$-algebraic AQFTs.
\end{rem}

\subsection{\label{subsec:HKcostack}Adjoint precostacks}
In this subsection we will prove that, for all our variants of the Haag-Kastler 
$2$-functor from Example \ref{ex:HK2functors}, the adjoint pseudo-functor \eqref{eqn:HKdagger}
is a precostack with respect to a suitable Grothendieck topology on $\Loc$.
In the case where no time-slice axiom is implemented, the topology is
given by all causally convex open covers $\U=\{U_i\subseteq M\}$.
In the case where a time-slice axiom is implemented, the topology must be chosen coarser and
it is given by all $D$-stable causally convex open covers $\U=\{U_i\subseteq M\}$. 
See also Definition \ref{def:cover}.

\subsubsection{The case of no time-slice}
Of the Haag-Kastler $2$-functors in Example \ref{ex:HK2functors},
the non-time-sliced examples (1) and (3) are associated with $2$-functors
$\ovr{\CC(-)}: \Loc\to\Cat^\perp$ that have additional properties which considerably simplify
the study of codescent of the adjoint pseudo-functor $\HK_{\ovr{\CC}}^\dagger$.
The following definition captures these properties abstractly.
\begin{defi}\label{def:nice2functor}
A $2$-functor $\ovr{\CC(-)} : \Loc\to \Cat^\perp$ is 
called a \textit{net domain} if it satisfies the following properties:
\begin{itemize}
\item[(1)] For all $M\in\Loc$, the underlying category of $\ovr{\CC(M)}$ is thin, 
i.e.\ there exists at most one morphism between every two objects.

\item[(2)] For all $M\in\Loc$ and all causally convex open subsets $V\subseteq M$, the orthogonal functor
$\iota_V^M = \ovr{\CC\big(\iota_V^M\big)} : \ovr{\CC(V)}\to \ovr{\CC(M)}$ is a full orthogonal subcategory inclusion 
$\ovr{\CC(V)}\subseteq \ovr{\CC(M)}$.

\item[(3)] For all $M\in\Loc$ and all causally convex open subsets $V_1,V_2\subseteq M$, 
if there exists a morphism $U_1\to U_2$ in $\ovr{\CC(M)}$ from 
$U_1\in \ovr{\CC(V_1)} \subseteq \ovr{\CC(M)}$ to $U_2 \in \ovr{\CC(V_2)} \subseteq  \ovr{\CC(M)}$,
then $U_1\in \ovr{\CC(V_1\cap V_2)}\subseteq \ovr{\CC(M)}$.
\end{itemize}
\end{defi}

\begin{ex}\label{ex:nice2functor}
The $2$-functors $\ovr{\COpen(-)} : \Loc\to \Cat^\perp$ from \eqref{eqn:COpenfunctor}
and $\ovr{\RC(-)} : \Loc\to \Cat^\perp$ from \eqref{eqn:RCfunctor} clearly satisfy
the properties of Definition \ref{def:nice2functor}. Hence, the Haag-Kastler $2$-functor
$\HK$ and the relatively compact Haag-Kastler $2$-functor $\HK^{\mathrm{rc}}$ are defined over net domains.
The orthogonal localizations $\ovr{\COpen(-)}[W^{-1}_{(-)}] : \Loc\to \Cat^\perp$ from \eqref{eqn:COpenfunctorW}
and $\ovr{\RC(-)}[W^{-1}_{\mathrm{rc},(-)}] : \Loc\to \Cat^\perp$ from \eqref{eqn:RCfunctorW}
satisfy only a weaker variant of these properties, see Definition \ref{def:Wnice2functor}
and Example \ref{ex:Wnice2functor} below for more details.
\end{ex}

In the case of net domains, one can drastically simplify 
the description of the orthogonal category $\ovr{\CC(\U)}$ 
from Definition \ref{def:CcoverOCat}.
\begin{lem}\label{lem:niceCcoverOCat}
Suppose that the $2$-functor $\ovr{\CC(-)} : \Loc\to \Cat^\perp$ is a net domain.
Let $\U=\{U_i\subseteq M\}$ be any causally convex open cover of any object $M\in\Loc$.
Then the orthogonal category $\ovr{\CC(\U)}$ from 
Definition \ref{def:CcoverOCat} admits the following explicit description:
\begin{itemize}
\item There exists a unique morphism $(i,U)\to (j,V)$ if and only if there exists
a morphism $U\to V$ in $\ovr{\CC(M)}$ from $U\in \ovr{\CC(U_i)}\subseteq \ovr{\CC(M)}$
to $V\in \ovr{\CC(U_j)}\subseteq \ovr{\CC(M)}$. 

\item A pair of morphisms $(i_1,U_1) \to (j,V) \leftarrow (i_2,U_2)$ in $\ovr{\CC(\U)}$
is orthogonal if and only if the pair of morphisms $U_1 \to V \leftarrow U_2$ in $\ovr{\CC(M)}$ is orthogonal.
\end{itemize}
\end{lem}
\begin{proof}
We will use the properties from Definition \ref{def:nice2functor} in order to simplify the
generators and relations presentation of the category $\ovr{\CC(\U)}$ 
from Definition \ref{def:CcoverOCat}. As a consequence of thinness (item (1)), 
we can drop the labels for the type (i) generators because there exists at most one such generator
for fixed source and target. Using also the full subcategory assumption (item (2)), we can regard all
objects and morphisms associated with the cover as living in $\ovr{\CC(M)}$. Given any type (ii) generator
$\psi_{ji,V} : (i,V)\to (j,V)$ with $V\in\ovr{\CC(U_{ij})}\subseteq \ovr{\CC(M)}$
and any type (i) generator $(i,U)\to (i,V)$ with $U\in\ovr{\CC(U_i)} \subseteq \ovr{\CC(M)}$,
then item (3) implies that $U\in \ovr{\CC(U_{ij})}\subseteq \ovr{\CC(M)}$. Hence, we get from the relations 
(r2) a commutative diagram
\begin{flalign}
\begin{gathered}
\xymatrix@C=3em{
\ar[d](i,U)\ar[r]^-{\psi_{ji,U}}~&~(j,U)\ar[d]\\
(i,V)\ar[r]_-{\psi_{ji,V}}~&~(j,V)
}
\end{gathered}
\end{flalign}
exchanging the order of composition of type (i) and type (ii) generators. This implies that
any morphism in  $\ovr{\CC(\U)}$  can be written in the form
$(\mathrm{type ~(i)})\circ (\mathrm{type ~(ii)})$ where all type (ii) generators are to the right
of the type (i) generators. Using thinness (item (1)) and the relations (r1) and (r3), it then follows that
there exists at most one morphism $(i,U) \to (j,V)$ in $\ovr{\CC(\U)}$,
with $U\in \ovr{\CC(U_i)}\subseteq \ovr{\CC(M)}$ and $V \in \ovr{\CC(U_j)}\subseteq \ovr{\CC(M)}$,
which must be of the form
\begin{flalign}
(i,U) ~\stackrel{\psi_{ji,U}}{\longrightarrow}~(j,U)~\longrightarrow~(j,V)
\end{flalign}
given by a type (ii) generator composed with a type (i) generator.
The type (i) morphism in this composition exists 
if and only if there exists a morphism $U\to V$ in $\ovr{\CC(M)}$,
which implies using item (3) that  $U\in \ovr{\CC(U_{ij})}\subseteq \ovr{\CC(M)}$
and hence also the type (ii) morphism exists. 
\sk

Let us now focus on the characterization of the orthogonality relation. 
Given two morphisms $(i_1,U_1) \to (j,V)$ and $(i_2,U_2) \to (j,V)$ in 
$\ovr{\CC(\U)}$, the above description of morphisms yields factorizations 
$(i_1,U_1) \cong (j,U_1) \to (j,V)$
and $(i_2,U_2) \cong (j,U_2) \to (j,V)$ in $\ovr{\CC(\U)}$.
(The first steps are isomorphisms by the relations (r3) from Definition \ref{def:CcoverOCat}.)
Therefore, by composition stability $(i_1,U_1) \to (j,V) \leftarrow (i_2,U_2)$ in $\ovr{\CC(\U)}$
is orthogonal if and only if $(j,U_1) \to (j,V) \leftarrow (j,U_2)$ in $\ovr{\CC(\U)}$
is orthogonal. According to Definition \ref{def:CcoverOCat} and item (2) from Definition \ref{def:nice2functor}, 
the latter is orthogonal if and only if $U_1 \to V \leftarrow U_2$ in $\ovr{\CC(M)}$ is orthogonal. 
\end{proof}

\begin{theo}\label{theo:precostack}
Suppose that the $2$-functor $\ovr{\CC(-)} : \Loc\to \Cat^\perp$ is a net domain 
in the sense of Definition \ref{def:nice2functor}. 
Let $\U=\{U_i\subseteq M\}$ be any causally convex open cover of any object $M\in\Loc$.
Then the canonical functor
\begin{flalign}\label{eqn:precostackcodescent}
j_{\U\,!}\,:\, \HK_{\ovr{\CC}}(\U)~\longrightarrow~\HK_{\ovr{\CC}}(M)
\end{flalign}
from the codescent category of $\HK_{\ovr{\CC}}^\dagger$ in this cover to its value on $M$ is fully faithful.
Hence, the adjoint pseudo-functor $\HK_{\ovr{\CC}}^\dagger : \Loc\to \Pr^L$ is a precostack
with respect to the Grothendieck topology given by all causally convex open covers.
\end{theo}
\begin{proof}
Using Lemma \ref{lem:niceCcoverOCat}, 
one verifies that the orthogonal
functor $j_{\U} : \ovr{\CC(\U)} \to\ovr{\CC(M)}$
from \eqref{eqn:jcoverfunctor} is fully faithful and reflects orthogonality.
Proposition \ref{propo:Ffullorthogonal} then implies 
that the associated operadic left Kan extension $j_{\U\,!}$ is a fully faithful functor.
\end{proof}

\begin{rem}\label{rem:dualprecostack}
The result of Theorem \ref{theo:precostack}
implies that, in the case where $\ovr{\CC(-)}$ is a net domain,
the canonical functor 
\begin{flalign}\label{eqn:precostackcodescentadjoint}
j_{\U}^\ast\,:\, \HK_{\ovr{\CC}}(M)~\longrightarrow~\HK_{\ovr{\CC}}(\U)
\end{flalign}
to the descent category of the original $2$-functor $\HK_{\ovr{\CC}} : \Loc^\op\to\Pr^R$
is essentially surjective, for every causally convex open cover $\U=\{U_i\subseteq M\}$. 
Indeed, since $j_{\U\,!}$ is fully faithful,
the unit of the adjunction $j_{\U\,!}\dashv j_{\U}^\ast$ is 
a natural isomorphism, so we obtain an isomorphism
\begin{flalign}
(\eta_{\U})_{\BBB}\,:\,\BBB ~\stackrel{\cong}{\Longrightarrow}~j_{\U}^\ast \,j_{\U\,!}(\BBB)\quad,
\end{flalign}
for every object $\BBB\in \HK_{\ovr{\CC}}(\U)$. Hence, every object 
of $\HK_{\ovr{\CC}}(\U)$ lies in the essential 
image of the functor \eqref{eqn:precostackcodescentadjoint}.
\end{rem}

As a direct consequence of Examples \ref{ex:HK2functors} and \ref{ex:nice2functor},
we obtain the following result.
\begin{cor}
The adjoint pseudo-functors $\HK^\dagger$ and $\HK^{\mathrm{rc}\,\dagger}$ 
of the Haag-Kastler $2$-functor and the relatively compact Haag-Kastler $2$-functor are
both precostacks with respect to the Grothendieck topology given by all 
causally convex open covers.
\end{cor}

\subsubsection{\label{subsubsec:precostack:timeslice}The case of time-slice}
For our examples (2) and (4) of time-sliced Haag-Kastler $2$-functors from Example \ref{ex:HK2functors},
the $2$-functor $\ovr{\CC(-)}: \Loc\to\Cat^\perp$ satisfies only a weaker variant
of the properties in Definition \ref{def:nice2functor}.
\begin{defi}\label{def:Wnice2functor}
A $2$-functor $\ovr{\CC(-)} : \Loc\to \Cat^\perp$ is 
called a \textit{localized net domain} if it satisfies the following properties:
\begin{itemize}
\item[(1)] For all $M\in\Loc$, the underlying category of $\ovr{\CC(M)}$ is thin, 
i.e.\ there exists at most one morphism between every two objects.

\item[(2)] For all $M\in\Loc$ and all $D$-stable causally convex open subsets $V\subseteq M$, i.e.\ 
$D_M(V)=V$ is stable under Cauchy development in $M$, the orthogonal functor
$\iota_V^M = \ovr{\CC\big(\iota_V^M\big)} : \ovr{\CC(V)}\to \ovr{\CC(M)}$ is a full orthogonal subcategory inclusion 
$\ovr{\CC(V)}\subseteq \ovr{\CC(M)}$.

\item[(3)] For all $M\in\Loc$ and all $D$-stable causally convex open subsets $V_1,V_2\subseteq M$,
if there exists a morphism $U_1\to U_2$ in $\ovr{\CC(M)}$ from 
$U_1\in \ovr{\CC(V_1)} \subseteq \ovr{\CC(M)}$ to $U_2 \in \ovr{\CC(V_2)} \subseteq  \ovr{\CC(M)}$,
then $U_1\in \ovr{\CC(V_1\cap V_2)}\subseteq \ovr{\CC(M)}$.
\end{itemize}
\end{defi}

\begin{ex}\label{ex:Wnice2functor}
Using the explicit localization models from Example \ref{ex:localizations}, we will now verify that the 
$2$-functors $\ovr{\COpen(-)}[W_{(-)}^{-1}] : \Loc\to \Cat^\perp$ from \eqref{eqn:COpenfunctorW}
and $\ovr{\RC(-)}[W_{\mathrm{rc},(-)}^{-1}] : \Loc\to \Cat^\perp$ from \eqref{eqn:RCfunctorW}
satisfy the properties of Definition \ref{def:Wnice2functor}. Hence, the time-sliced Haag-Kastler
$2$-functor $\HK^W$ and the time-sliced relatively compact Haag-Kastler $2$-functor 
$\HK^{\mathrm{rc},W}$ are defined over localized net domains.
\sk

Let us start with the case of $\ovr{\COpen(-)}[W_{(-)}^{-1}]$. 
We observe that the category underlying $\ovr{\COpen(M)}[W_M^{-1}]$ 
is manifestly thin, for all $M\in\Loc$.
Furthermore, given any $D$-stable causally convex open subset 
$V\subseteq M$, one has that the orthogonal functor
$\iota_V^M : \ovr{\COpen(V)}[W_{V}^{-1}] \to \ovr{\COpen(M)}[W_{M}^{-1}]$ 
is faithful and reflects orthogonality. 
To prove fullness, we observe that 
given any causally convex open subsets 
$U, U^\prime \subseteq V$ such that $U \subseteq D_M(U^\prime)$, 
the $D$-stability $D_M(V) = V$ of $V$ 
entails that $D_V(U^\prime) = D_M(U^\prime)$, 
hence $U \subseteq D_V(U^\prime)$. 
Finally, given any $D$-stable causally convex open subsets $V_1, V_2\subseteq M$ 
and any causally convex open subsets $U_1 \subseteq V_1$ and $U_2 \subseteq V_2$ 
such that $U_1 \subseteq D_M(U_2)$, it follows that 
$U_1 \subseteq V_1 \cap D_M(U_2) = V_1 \cap D_{V_2}(U_2) \subseteq V_1 \cap V_2$.
\sk

The case of $\ovr{\RC(-)}[W_{\mathrm{rc},(-)}^{-1}]$ follows by specializing
the above arguments to causally convex open subsets $U,U^\prime\subseteq V$, 
$U_1\subseteq V_2$ and $U_2\subseteq V_2$ which are also relatively compact.
\end{ex}

The following statement is the analogue of Lemma \ref{lem:niceCcoverOCat} in the present case.
\begin{lem}\label{lem:WniceCcoverOCat}
Suppose that the $2$-functor $\ovr{\CC(-)} : \Loc\to \Cat^\perp$ is a localized net domain.
Let $\U=\{U_i\subseteq M\}$ be any $D$-stable causally convex open cover of any object $M\in\Loc$.
Then the orthogonal category $\ovr{\CC(\U)}$ from 
Definition \ref{def:CcoverOCat} admits the following explicit description:
\begin{itemize}
\item There exists a unique morphism $(i,U)\to (j,V)$ if and only if there exists
a morphism $U\to V$ in $\ovr{\CC(M)}$ from $U\in \ovr{\CC(U_i)}\subseteq \ovr{\CC(M)}$
to $V\in \ovr{\CC(U_j)}\subseteq \ovr{\CC(M)}$. 

\item A pair of morphisms $(i_1,U_1) \to (j,V) \leftarrow (i_2,U_2)$ in $\ovr{\CC(\U)}$
is orthogonal if and only if the pair of morphisms $U_1 \to V \leftarrow U_2$ in $\ovr{\CC(M)}$ is orthogonal.
\end{itemize}
\end{lem}
\begin{proof}
Since the cover $\U=\{U_i\subseteq M\}$ is by hypothesis 
$D$-stable, it follows that all intersections $U_{ij} = U_i\cap U_j\subseteq M$ are $D$-stable too,
i.e.\ $D_M(U_{ij}) = U_{ij}$. The proof is then identical to the one of Lemma \ref{lem:niceCcoverOCat}.
\end{proof}

\begin{theo}\label{theo:Wprecostack}
Suppose that the $2$-functor $\ovr{\CC(-)} : \Loc\to \Cat^\perp$ is a localized
net domain in the sense of Definition \ref{def:Wnice2functor}. 
Let $\U=\{U_i\subseteq M\}$ be any $D$-stable causally 
convex open cover of any object $M\in\Loc$.
Then the canonical functor
\begin{flalign}\label{eqn:Wprecostackcodescent}
j_{\U\,!}\,:\, \HK_{\ovr{\CC}}(\U)~\longrightarrow~\HK_{\ovr{\CC}}(M)
\end{flalign}
from the codescent category of $\HK_{\ovr{\CC}}^\dagger$ in this cover to its value on $M$ is fully faithful.
Hence, the adjoint pseudo-functor $\HK_{\ovr{\CC}}^\dagger : \Loc\to \Pr^L$ is a precostack
with respect to the Grothendieck topology given by all 
$D$-stable causally convex open covers.
\end{theo}
\begin{proof}
Using that the cover $\U$ is $D$-stable 
and Lemma \ref{lem:WniceCcoverOCat}, one verifies that the orthogonal
functor $j_{\U} : \ovr{\CC(\U)} \to\ovr{\CC(M)}$
from \eqref{eqn:jcoverfunctor} is fully faithful and reflects orthogonality. 
Proposition \ref{propo:Ffullorthogonal} then implies 
that the associated operadic left Kan extension $j_{\U\,!}$ is a fully faithful functor.
\end{proof}

As a direct consequence of Examples \ref{ex:HK2functors} and \ref{ex:Wnice2functor},
we obtain the following result.
\begin{cor}
The adjoint pseudo-functors $\HK^{W\,\dagger}$ and $\HK^{\mathrm{rc},W\,\dagger}$ 
of the time-sliced Haag-Kastler $2$-functor and the time-sliced relatively compact Haag-Kastler $2$-functor are
both precostacks with respect to the Grothendieck topology given by all 
$D$-stable causally convex open covers.
\end{cor}

\subsection{\label{subsec:HKcostackification}Improving descent}
Our results in Subsection \ref{subsec:HKcostack} 
provide insights about why the Haag-Kastler-style $2$-functors $\HK_{\ovr{\CC}}$ 
fail to satisfy the descent conditions of a stack. From the dual perspective
of the adjoint pseudo-functors $\HK_{\ovr{\CC}}^\dagger$,
we see that the canonical functors \eqref{eqn:precostackcodescent} and \eqref{eqn:Wprecostackcodescent}
from the (co)descent category $\HK_{\ovr{\CC}}(\U)$ to $\HK_{\ovr{\CC}}(M)$
are fully faithful but fail to be essentially surjective. (This failure
follows by an argument as in Remark \ref{rem:dualprecostack}
from our results in Propositions \ref{prop:HKnotastack}, \ref{prop:HKtimeslicenotastack},
\ref{prop:RCHKnotastack} and \ref{prop:RCHKtimeslicenotastack}.) 
Loosely speaking, this means that the category $\HK_{\ovr{\CC}}(M)$ contains also `bad objects'
which do not interplay well with descent 
and it suggests that one should select a suitable class of `good objects' 
in $\HK_{\ovr{\CC}}(M)$. Our proposal for a selection criterion
can be stated in simple terms as follows: We would like to intersect
the essential images of the (co)descent categories
$\HK_{\ovr{\CC}}(\U)\to \HK_{\ovr{\CC}}(M)$ over all covers $\U$
and thereby define a full subcategory of `good objects' in $\HK_{\ovr{\CC}}(M)$.
In this subsection we shall formalize and discuss this idea for both
the case where $\ovr{\CC(-)}$ is a net domain and the case where $\ovr{\CC(-)}$ 
is a localized net domain. These two cases are very similar but, as already indicated
in Theorems \ref{theo:precostack} and \ref{theo:Wprecostack}, they require slightly different
choices of Grothendieck topologies on $\Loc$. To avoid messy notations and case distinctions, 
we shall present each case in an individual subsection.

\subsubsection{\label{subsubsec:notimeslicedescent}The case of no time-slice}
Throughout this subsection, let us assume that  $\ovr{\CC(-)} : \Loc\to\Cat^\perp$ is a net domain
as in Definition \ref{def:nice2functor}. We formalize our selection criterion for `good objects'
as follows.
\begin{defi}\label{def:improvedHK}
Let $\ovr{\CC(-)} : \Loc\to\Cat^\perp$ be a net domain.
For every object $M\in\Loc$, we denote by
\begin{subequations}\label{eqn:improvedHKdescentcondition}
\begin{flalign}
\calHK_{\ovr{\CC}}(M)\,\subseteq\,\HK_{\ovr{\CC}}(M)
\end{flalign}
the full subcategory consisting of all objects $\AAA\in\HK_{\ovr{\CC}}(M)$ 
which satisfy the following descent conditions: For every
causally convex open cover $\U=\{U_i\subseteq M\}$, the $\AAA$-component 
of the counit
\begin{flalign}
(\epsilon_{\U})_\AAA \,:\, j_{\U\,!}\,j_{\U}^\ast(\AAA)~
\stackrel{\cong}{\Longrightarrow}~\AAA
\end{flalign}
\end{subequations}
of the adjunction
$j_{\U\,!} : \HK_{\ovr{\CC}}(\U) \rightleftarrows \HK_{\ovr{\CC}}(M) : j_\U^\ast$
of \eqref{eqn:descentadjunction} is an isomorphism in $\HK_{\ovr{\CC}}(M)$.
\end{defi}

From this definition it is not immediately clear 
that the categories $\calHK_{\ovr{\CC}}(M)$ are locally presentable
and that the inclusion functors $\calHK_{\ovr{\CC}}(M) \subseteq \HK_{\ovr{\CC}}(M)$
are left adjoints. We will now show that $\calHK_{\ovr{\CC}}(M)$ arises
as the bilimit of a suitable diagram in $\Pr^L$ which will manifestly imply 
these properties. To set up this construction, we briefly recall the concept of
refinements of covers.
\begin{defi}\label{def:covcats}
Given any object $M\in\Loc$, we denote by $\mathbf{cov}(M)$ the category
whose objects are all causally convex open covers $\U=\{U_i\subseteq M\}$ of $M$
and whose morphisms $\alpha : \U=\{U_i\subseteq M\}\to \U^\prime=\{U^\prime_{i^\prime}\subseteq M\}$ are refinements,
i.e.\ $\alpha :\mathcal{I}\to\mathcal{I}^\prime\,,~i\mapsto \alpha(i)$ is a map
between the indexing sets such that $U_i\subseteq U^\prime_{\alpha(i)}$, for all $i$.
Note that the category $\mathbf{cov}(M)$ has a terminal object, which is given
by the coarsest cover $\{M\subseteq M\}$.
\end{defi}

We observe that the assignment of the orthogonal categories
from Definition \ref{def:CcoverOCat} is $2$-functorial
\begin{flalign}
\ovr{\CC(-)}\,:\,\mathbf{cov}(M)~\longrightarrow~\Cat^\perp~~,\quad
\U~\longmapsto~\ovr{\CC(\U)}\quad.
\end{flalign}
Using the simplification for a net domain from Lemma \ref{lem:niceCcoverOCat},
the action of this $2$-functor on refinements 
$\alpha : \U\to \U^\prime$
is given by the orthogonal functors (denoted with abuse of notation by the same symbol $\alpha$)
\begin{flalign}\label{eqn:alphaOFun}
\alpha\, :=\,\ovr{\CC(\alpha)} \,:\,
\ovr{\CC(\U)}~&\longrightarrow~\ovr{\CC(\U^\prime)}\quad, \\
\nn (i,U)~&\longmapsto~(\alpha(i),U)\quad,\\
\nn  \big((i,U)\to (j,V)\big)~&\longmapsto~\big((\alpha(i),U)\to(\alpha(j),V)\big)\quad.
\end{flalign}
For each refinement $\alpha : \U\to \U^\prime$,
the orthogonal functor \eqref{eqn:alphaOFun} is fully faithful and reflects orthogonality.
Applying Proposition \ref{propo:Ffullorthogonal}, we then obtain an adjunction
\begin{flalign}\label{eqn:alphaadjunction}
\xymatrix{
\alpha_! \,:\, \HK_{\ovr{\CC}}(\U) \ar@<0.75ex>[r]~&~\ar@<0.75ex>[l] 
\HK_{\ovr{\CC}}(\U^\prime)\,:\,\alpha^\ast
}
\end{flalign}
whose left adjoint $\alpha_!$ is fully faithful. 
\begin{rem}\label{rem:descentfinerimpliescoarser}
As a side remark which will become relevant in some of our proofs below, 
we observe that the descent conditions in Definition \ref{def:improvedHK} are not independent
because descent on finer covers implies descent on coarser ones.
Let us provide the relevant argument.
Given any refinement of covers $\alpha : \U\to\U^\prime$,
we obtain a commutative diagram of orthogonal functors
\begin{flalign}\label{eqn:refinementdiagram}
\begin{gathered}
\xymatrix@C=1em@R=1.5em{
\ar[dr]_-{\alpha}\ovr{\CC(\U)} \ar[rr]^-{j_{\U}}~&~ ~&~\ovr{\CC(M)}\\
~&~ \ovr{\CC(\U^\prime)} \ar[ur]_-{j_{\U^\prime}}~&~
}
\end{gathered}\quad.
\end{flalign}
Suppose that an object $\AAA\in\HK_{\ovr{\CC}}(M)$ satisfies descent on the finer cover $\U$, 
i.e.\ the counit $(\epsilon_{\U})_\AAA : j_{\U\,!}\,j_{\U}^\ast(\AAA) \Rightarrow\AAA$ 
is an isomorphism $\HK_{\ovr{\CC}}(M)$. Applying $(-)_!$ to the diagram \eqref{eqn:refinementdiagram}
gives a diagram which commutes up to a natural isomorphism, so we find that our object
\begin{flalign}
\AAA \,\cong\,   j_{\U\,!}\,j_{\U}^\ast(\AAA) \,\cong\, j_{\U^\prime\,!}\,\alpha_!\,j_{\U}^\ast(\AAA)\,=:\,
 j_{\U^\prime\,!}(\BBB)
\end{flalign}
lies in the essential image of $j_{\U^\prime\,!}:\HK_{\ovr{\CC}}(\U^\prime)\to \HK_{\ovr{\CC}}(M)$ 
for the coarser cover. From this isomorphism we obtain a commutative diagram
\begin{flalign}
\begin{gathered}
\xymatrix@C=5em{
j_{\U^\prime\,!}\,j_{\U^\prime}^\ast(\AAA) \ar@{=>}[r]^-{(\epsilon_{\U^\prime})_\AAA}~&~ \AAA\\
\ar@{=>}[u]^-{\cong}j_{\U^\prime\,!}\,j_{\U^\prime}^\ast\,j_{\U^\prime\,!}(\BBB) \ar@{=>}[r]^-{(\epsilon_{\U^\prime})_{j_{\U^\prime\,!}(\BBB)}}~&~j_{\U^\prime\,!}(\BBB)\ar@{=>}[u]_-{\cong}\\
\ar@{=>}[u]^-{j_{\U^\prime\,!}(\eta_{\U^\prime})_\BBB}_-{\cong} j_{\U^\prime\,!}(\BBB)\ar@{=}[ru]~&~
}
\end{gathered}
\end{flalign}
with bottom part the triangle identity for the unit and counit of the adjunction
$j_{\mathcal{U}^\prime\,!}\dashv j_{\mathcal{U}^\prime}^\ast$.
This implies that $\AAA$ satisfies descent on the coarser cover $\U^\prime$
because the unit of the adjunction $j_{\U^\prime\,!} \dashv j_{\U^\prime}^\ast$
is a natural isomorphism by Theorem \ref{theo:precostack}.
\end{rem}

Let us define the pseudo-functor
\begin{flalign}\label{eqn:covfunctor}
\HK_{\ovr{\CC}}^{\dagger}\,:\, \mathbf{cov}(M)~\longrightarrow~\Pr^L
\end{flalign}
which assigns to each cover $\U$ the locally presentable
category $\HK_{\ovr{\CC}}(\U)$ and to each refinement
$\alpha$ the corresponding left adjoint $\alpha_!$ from \eqref{eqn:alphaadjunction}.
\begin{propo}\label{propo:bilimovercover}
Let $\ovr{\CC(-)} : \Loc\to\Cat^\perp$ be a net domain.
For every object $M\in\Loc$, the category 
\begin{flalign}
\calHK_{\ovr{\CC}}(M)~\simeq~\bilim\big( \HK_{\ovr{\CC}}^{\dagger} : \mathbf{cov}(M) \to \Pr^L \big)
\end{flalign}
from Definition \ref{def:improvedHK} 
is a bilimit of the pseudo-functor \eqref{eqn:covfunctor},
hence it is locally presentable $\calHK_{\ovr{\CC}}(M)\in \Pr^L$. 
Furthermore, the full subcategory inclusion $\iota_M : \calHK_{\ovr{\CC}}(M)\subseteq \HK_{\ovr{\CC}}(M)$
is coreflective, i.e.\ there exists an adjunction
\begin{flalign}\label{eqn:coreflector}
\xymatrix{
\iota_M\,:\, \calHK_{\ovr{\CC}}(M) \ar@<0.75ex>[r]^-{\subseteq} ~&~  \ar@<0.75ex>[l]\HK_{\ovr{\CC}}(M)\,:\,\pi_M
}
\end{flalign}
with coreflector $\pi_M$.
\end{propo}
\begin{proof}
The key properties which enable this result are 1.)~all functors $\alpha_!$ are fully faithful,
and 2.)~the category $\mathbf{cov}(M)$ has a terminal object, given by the coarsest
cover $\{M\subseteq M\}$. The proof then follows from our general 
bilimit computation in Appendix \ref{app:bilimit}.
\end{proof}

The categories of covers from Definition \ref{def:covcats} are $2$-functorial
$\mathbf{cov}(-) : \Loc^\op \to \Cat$ with respect to pullbacks of covers along $\Loc$-morphisms.
Explicitly, given any $\Loc$-morphism $f:M\to N$, we obtain a functor
\begin{flalign}
f^{-1}\,:=\,\mathbf{cov}(f)\,:\, \mathbf{cov}(N)~&\longrightarrow~\mathbf{cov}(M)\quad,\\
\nn \V= \{V_j\subseteq N\}~&\longmapsto~ f^{-1}\V = \{f^{-1}(V_j)\subseteq M\}\quad,\\
\nn \big(\alpha : \V\to \V^\prime\big)~&\longmapsto~
\big(\alpha : f^{-1}\V\to f^{-1}\V^\prime\big)
\end{flalign}
by taking preimages under $f$. (Since we focus on non-empty causally convex opens,
we always discard all empty preimages $f^{-1}(V_j)=\emptyset$ from the covers.)
This $2$-functorial structure endows the bilimits from Proposition \ref{propo:bilimovercover}
with a pseudo-functorial structure. Transferring this structure to our explicit
models from Definition \ref{def:improvedHK} yields the pseudo-functor
\begin{subequations}\label{eqn:adjointimproved}
\begin{flalign}
\calHK_{\ovr{\CC}}^\dagger\,:\,\Loc~\longrightarrow~\Pr^L
\end{flalign} 
which assigns to each object $M\in\Loc$ the locally presentable category
$\calHK_{\ovr{\CC}}^\dagger(M):=\calHK_{\ovr{\CC}}(M)\in\Pr^L$ from Definition \ref{def:improvedHK}
and to each $\Loc$-morphism $f:M\to N$ the restriction
\begin{flalign}
\calHK_{\ovr{\CC}}^\dagger(f)\,:=\,f_! \,:\, \calHK_{\ovr{\CC}}(M)~\longrightarrow~\calHK_{\ovr{\CC}}(N)
\end{flalign}
\end{subequations}
of the left adjoint $f_! : \HK_{\ovr{\CC}}(M)\to \HK_{\ovr{\CC}}(N)$ to the full subcategories
$\calHK_{\ovr{\CC}}(M)\subseteq \HK_{\ovr{\CC}}(M)$ and $\calHK_{\ovr{\CC}}(N)\subseteq \HK_{\ovr{\CC}}(N)$.
\begin{rem}
For readers who prefer a more direct and computational argument, let us also verify 
explicitly that the functors $f_! : \HK_{\ovr{\CC}}(M)\to \HK_{\ovr{\CC}}(N)$ restrict to the full subcategories
$f_! : \calHK_{\ovr{\CC}}(M)\to \calHK_{\ovr{\CC}}(N)$.
Given any object $\AAA\in \calHK_{\ovr{\CC}}(M)$, we have to show that 
$f_! (\AAA)\in \HK_{\ovr{\CC}}(N) $ satisfies the descent condition from Definition \ref{def:improvedHK}
for every causally convex open cover $\mathcal{V}$ of $N\in\Loc$. 
Taking the pullback $f^{-1}\mathcal{V}$ of this cover, we obtain a commutative square
\begin{flalign}\label{eqn:comdiag}
\begin{gathered}
\xymatrix@C=4em{
\ar[d]_-{\tilde{f}}\ovr{\CC(f^{-1}\mathcal{V})} \ar[r]^-{j_{f^{-1}\mathcal{V}}}~&~\ovr{\CC(M)}\ar[d]^-{f}\\
\ovr{\CC(\mathcal{V})}\ar[r]_-{j_{\mathcal{V}}}~&~\ovr{\CC(N)}
}
\end{gathered}
\end{flalign}
of orthogonal functors. The orthogonal functor $\tilde{f}$ is defined on objects by 
$(j,U)\mapsto (j, f(U))$
and on morphisms by $\big((j,U)\to (j^\prime,U^\prime)\big)\mapsto 
\big((j,f(U))\to (j^\prime,f(U^\prime)\big)$.
We then obtain the commutative diagram
\begin{flalign}
\begin{gathered}
\xymatrix@C=8em{
j_{\mathcal{V}\,!}\,j^\ast_{\mathcal{V}} \, f_!(\AAA) 
\ar@{=>}[r]^-{(\epsilon_{\mathcal{V}})_{f_!(\AAA)}}
~&~f_!(\AAA)\\
\ar@{=>}[u]_-{\cong}^-{j_{\mathcal{V}\,!}j^\ast_{\mathcal{V}}f_! (\epsilon_{f^{-1}\mathcal{V}})_\AAA} 
j_{\mathcal{V}\,!}\,j^\ast_{\mathcal{V}}\,f_! \,j_{f^{-1}\mathcal{V}\,!}\,j_{f^{-1}\mathcal{V}}^\ast(\AAA) 
\ar@{=>}[r]^-{(\epsilon_{\mathcal{V}})_{f_!j_{f^{-1}\mathcal{V}\,!}j_{f^{-1}\mathcal{V}}^\ast(\AAA)}}
~&~f_!\,j_{f^{-1}\mathcal{V}\,!}\,j_{f^{-1}\mathcal{V}}^\ast(\AAA)
\ar@{=>}[u]^-{\cong}_-{f_! (\epsilon_{f^{-1}\mathcal{V}})_\AAA}\\
\ar@{=>}[u]^-{\cong} 
j_{\mathcal{V}\,!}\,j^\ast_{\mathcal{V}}\,j_{\mathcal{V}\,!} \,\tilde{f}_!\,  j_{f^{-1}\mathcal{V}}^\ast(\AAA) 
\ar@{=>}[r]^-{(\epsilon_{\mathcal{V}})_{j_{\mathcal{V}\,!} \tilde{f}_!j_{f^{-1}\mathcal{V}}^\ast(\AAA)}}
~&~j_{\mathcal{V}\,!}\, \tilde{f}_! \,j_{f^{-1}\mathcal{V}}^\ast(\AAA)
\ar@{=>}[u]_-{\cong}\\
\ar@{=>}[u]_-{\cong}^-{j_{\mathcal{V}\,!}(\eta_{\mathcal{V}})_{\tilde{f}_!  j_{f^{-1}\mathcal{V}}^\ast(\AAA)}}j_{\mathcal{V}\,!}\, \tilde{f}_! \, j_{f^{-1}\mathcal{V}}^\ast(\AAA) \ar@{=}[ru]~&~
}
\end{gathered}\qquad.
\end{flalign}
In the top square we use that $\AAA\in \calHK_{\ovr{\CC}}(M)$
satisfies the descent condition from Definition \ref{def:improvedHK} for the cover $f^{-1}\mathcal{V}$
of $M$ and in the middle square we use that applying $(-)_!$ to the commutative 
diagram \eqref{eqn:comdiag} yields a diagram which commutes up to a natural isomorphism.
In the bottom triangle we use the triangle identity for the unit and counit of the adjunction
$j_{\mathcal{V}\,!}\dashv j_{\mathcal{V}}^\ast$,
as well as the fact that the unit is a natural isomorphism since
$j_{\mathcal{V}}$ is fully faithful and reflects orthogonality.
From this diagram it follows that $(\epsilon_{\mathcal{V}})_{f_!(\AAA)}$ is an isomorphism,
hence $f_!(\AAA)$ satisfies the descent conditions from Definition \ref{def:improvedHK}.
\end{rem}

\begin{defi}\label{def:improvedHKpseudofunctor}
Let $\ovr{\CC(-)} : \Loc\to\Cat^\perp$ be a net domain.
The \textit{improved Haag-Kastler-style pseudo-functor}
\begin{flalign}
\calHK_{\ovr{\CC}}\,:=\, \calHK_{\ovr{\CC}}^{\dagger\dagger}\,:\, \Loc^\op~\longrightarrow~\Pr^R
\end{flalign}
is defined as the adjoint via \eqref{eqn:LtoR} 
of the pseudo-functor $\calHK_{\ovr{\CC}}^\dagger$ in \eqref{eqn:adjointimproved}.
\end{defi}

The improved Haag-Kastler-style pseudo-functor $\calHK_{\ovr{\CC}}$ is in general 
hard to work with because the right adjoints to the functors
$\calHK_{\ovr{\CC}}^\dagger(f)=f_! : \calHK_{\ovr{\CC}}(M)\to \calHK_{\ovr{\CC}}(N)$
in \eqref{eqn:adjointimproved} are difficult to construct explicitly.
Indeed, the pullback functors $f^\ast : \HK_{\ovr{\CC}}(N)\to \HK_{\ovr{\CC}}(M)$
assigned by the (non-improved) Haag-Kastler-style $2$-functor $\HK_{\ovr{\CC}}:\Loc^\op\to\Pr^R$
for a generic net domain $\ovr{\CC}$
do \textit{not} necessarily restrict to the full subcategories $\calHK_{\ovr{\CC}}(N)\subseteq \HK_{\ovr{\CC}}(N)$
and $\calHK_{\ovr{\CC}}(M)\subseteq \HK_{\ovr{\CC}}(M)$ from Definition \ref{def:improvedHK}.
To describe the pseudo-functorial structure of $\calHK_{\ovr{\CC}}$,
one can then use the coreflectors $\pi_M$ from \eqref{eqn:coreflector},
with a possible model for the right adjoint functor 
$\calHK_{\ovr{\CC}}(f) \vdash f_! = \calHK_{\ovr{\CC}}^\dagger(f)$ given by
\begin{flalign}\label{eqn:fpullcoreflected}
\calHK_{\ovr{\CC}}(N)~\stackrel{\iota_N}{\longrightarrow}~
\HK_{\ovr{\CC}}(N)  ~\stackrel{f^\ast}{\longrightarrow}~\HK_{\ovr{\CC}}(M)
~\stackrel{\pi_M}{\longrightarrow}~\calHK_{\ovr{\CC}}(M)\quad.
\end{flalign}
The existence of these coreflectors was argued in Proposition
\ref{propo:bilimovercover} by abstract reasoning, but we are
currently not aware of any explicit models for $\pi_M$ which are useful for computations.
In order to simplify the remaining part of this subsection,
we will now include a very useful, but possibly quite strong, assumption on the behavior of 
the Haag-Kastler-style $2$-functor $\HK_{\ovr{\CC}}$. We will verify
in Theorem \ref{theo:RCHKstack} below that this assumption holds true
for the relatively compact Haag-Kastler $2$-functor
$\HK^{\mathrm{rc}}$, but it is currently not clear to us if 
it also holds true for the Haag-Kastler $2$-functor $\HK$ 
which is modeled on all causally convex opens.
\begin{assu}\label{assu:fastrestricts}
We assume that, for every $\Loc$-morphism $f:M\to N$, the pullback
functor $f^\ast :\HK_{\ovr{\CC}}(N)\to \HK_{\ovr{\CC}}(M)$
restricts to a functor $f^\ast :\calHK_{\ovr{\CC}}(N)\to \calHK_{\ovr{\CC}}(M)$
between the full subcategories $\calHK_{\ovr{\CC}}(N)\subseteq \HK_{\ovr{\CC}}(N)$
and $\calHK_{\ovr{\CC}}(M)\subseteq \HK_{\ovr{\CC}}(M)$ from Definition \ref{def:improvedHK}.
\end{assu}

Provided that Assumption \ref{assu:fastrestricts} holds true, 
one obtains a particularly simple model for the improved Haag-Kastler-style pseudo-functor
from Definition \ref{def:improvedHKpseudofunctor} in terms of a $2$-subfunctor 
$\calHK_{\ovr{\CC}}\subseteq \HK_{\ovr{\CC}}$. This $2$-functor
assigns to each object $M\in\Loc$ the locally presentable
category $\calHK_{\ovr{\CC}}(M)\in\Pr^R$ from Definition \ref{def:improvedHK}
and to each $\Loc$-morphism the restricted pullback functor 
$f^\ast :\calHK_{\ovr{\CC}}(N)\to \calHK_{\ovr{\CC}}(M)$,
which in this case is right adjoint to the functor
$f_! :  \calHK_{\ovr{\CC}}(M)\to  \calHK_{\ovr{\CC}}(N)$ from \eqref{eqn:adjointimproved}.
For every object $M\in\Loc$ and every causally convex open cover $\U=\{U_i\subseteq M\}$, 
the descent category of $\calHK_{\ovr{\CC}}$ is then given by the full subcategory
\begin{flalign}\label{eqn:caldescentcat}
\calHK_{\ovr{\CC}}(\U)\,\subseteq\,\HK_{\ovr{\CC}}(\U)
\end{flalign}
consisting of all objects $(\{\AAA_i\},\{\varphi_{ij}\})\in \HK_{\ovr{\CC}}(\U)$
such that $\AAA_i\in \calHK_{\ovr{\CC}}(U_i)\subseteq \HK_{\ovr{\CC}}(U_i)$
lies in the full subcategory of objects satisfying the descent conditions 
from Definition \ref{def:improvedHK}, for all $i$.
The canonical functor 
\begin{flalign}\label{eqn:caldescentmap}
j_{\U}^\ast\,:\,\calHK_{\ovr{\CC}}(M)~\longrightarrow~ \calHK_{\ovr{\CC}}(\U)
\end{flalign}
to the descent category is given by restricting the right adjoint 
$j_{\U}^\ast:\HK_{\ovr{\CC}}(M) \to \HK_{\ovr{\CC}}(\U)$ from \eqref{eqn:descentadjunction}
to the full subcategories $\calHK_{\ovr{\CC}}(M)\subseteq \HK_{\ovr{\CC}}(M)$
and $\calHK_{\ovr{\CC}}(\U)\subseteq \HK_{\ovr{\CC}}(\U)$. The following
result provides an explicit model for the left adjoint of \eqref{eqn:caldescentmap}.
\begin{propo}\label{prop:caldescentmapadjoint}
Suppose that the $2$-functor $\ovr{\CC(-)} : \Loc\to \Cat^\perp$ is a
net domain and that Assumption \ref{assu:fastrestricts} holds true.
Then, for every object $M\in\Loc$ and every causally convex open cover $\U=\{U_i\subseteq M\}$, 
the left adjoint $j_{\U\,!} :\HK_{\ovr{\CC}}(\U) \to \HK_{\ovr{\CC}}(M)$ from \eqref{eqn:descentadjunction}
restricts to the full subcategories $\calHK_{\ovr{\CC}}(\U)\subseteq \HK_{\ovr{\CC}}(\U)$
and $\calHK_{\ovr{\CC}}(M)\subseteq \HK_{\ovr{\CC}}(M)$, and thereby defines
a left adjoint $j_{\U\,!} :\calHK_{\ovr{\CC}}(\U) \to \calHK_{\ovr{\CC}}(M)$
for the functor \eqref{eqn:caldescentmap}.
\end{propo}
\begin{proof}
We have to prove that, given any object $\AAA:=(\{\AAA_i\},\{\varphi_{ij}\})\in \calHK_{\ovr{\CC}}(\U)$
in the improved descent category \eqref{eqn:caldescentcat}, i.e.\ every $\AAA_i\in \calHK_{\ovr{\CC}}(U_i)$ 
satisfies the descent conditions from Definition \ref{def:improvedHK} for all 
causally convex open covers of $U_i$, the resulting object 
$j_{\U\,!}(\AAA)\in \HK_{\ovr{\CC}}(M)$ 
satisfies the descent conditions from Definition \ref{def:improvedHK} 
for all causally convex open covers of $M$. We shall denote
these arbitrary covers by $\U^\prime$ since the symbol $\U$
is already reserved by the choice of cover in the statement of this proposition.
\sk

To prove that $j_{\U\,!}(\AAA)\in \HK_{\ovr{\CC}}(M)$
satisfies descent on any causally convex open cover $\U^\prime$ of $M$, we can use 
Remark \ref{rem:descentfinerimpliescoarser} and study instead descent on any 
cover which is finer than $\U^\prime$. A suitable choice is given by 
intersecting the covers $\U^\prime$ and $\U$, yielding the causally convex open cover
$\U\cap \U^\prime  := \{U_i\cap U_{i^\prime}^\prime\subseteq M\}$ of $M$ which is labeled by pairs of indices
$(i,i^\prime)\in\mathcal{I}\times\mathcal{I}^\prime$. The projection maps
$\pr_1 : \mathcal{I}\times\mathcal{I}^\prime\to \mathcal{I}$ and 
$\pr_2:\mathcal{I}\times\mathcal{I}^\prime\to \mathcal{I}^\prime$ then 
define refinements which fit into the following commutative diagram of orthogonal functors
\begin{flalign}
\begin{gathered}
\xymatrix@C=1em@R=1em{
~&~ \ovr{\CC(\U)}\ar[dr]^-{j_{\U}} ~&~\\
\ovr{\CC(\U\cap \U^\prime)} \ar[ur]^-{\pr_1}\ar[dr]_-{\pr_2}\ar[rr]^-{j_{\U\cap \U^\prime}}~&~ ~&~ \ovr{\CC(M)}\\
~&~ \ovr{\CC(\U^\prime)} \ar[ru]_-{j_{\U^\prime}}~&~\\
}
\end{gathered}\qquad.
\end{flalign}
Taking the upper path of this diagram
and using the same argument as the one at the end 
of Remark \ref{rem:descentfinerimpliescoarser}, we find that
the descent condition $(\epsilon_{\U\cap \U^\prime})_{j_{\U\,!}(\AAA)} : 
j_{\U\cap \U^\prime\,!}\,j_{\U\cap \U^\prime}^\ast\, j_{\U\,!}(\AAA) \stackrel{\cong}{\Longrightarrow}j_{\U\,!}(\AAA)$
on the intersection cover holds true provided that 
$\AAA$ lies in the essential image of $\pr_{1\,!}$. 
To this end, we will argue that the counit component $\epsilon_\AAA : \pr_{1\,!}\,\pr_1^\ast(\AAA)\Rightarrow \AAA$
is an isomorphism, which follows from the fact that, for every $i$, 
the component $\AAA_i\in\calHK_{\ovr{\CC}}(U_i)$ of the tuple 
$\AAA=(\{\AAA_i\},\{\varphi_{ij}\})\in \calHK_{\ovr{\CC}}(\U)$ 
satisfies by definition the descent conditions on $U_i$ and the fiber
$\U_i := \pr_1^{-1}(i) := \{U_i\cap U_{i^\prime}^\prime : i^\prime\in\mathcal{I}^\prime\}$ 
of $\pr_1 : \U\cap \U^\prime\to \U$ defines a causally convex 
open cover of $U_i$.
\sk

Using Appendix \ref{app:operadicLKE},
one obtains an explicit model for the operadic left Kan extension $\pr_{1\,!}$ 
for which the counit component $\epsilon_\AAA$ is given by the canonical maps
\begin{flalign}\label{eqn:epsilonpr}
(\epsilon_\AAA)_{(i,U)}\,:\, \colim\Big(\O_{\pr_1}^\otimes\big/(i,U)\longrightarrow
\O_{\ovr{\CC(\U\cap\U^\prime)}}^\otimes\stackrel{\O_{\pr_1}^\otimes}{\longrightarrow}
\O_{\ovr{\CC(\U)}}^\otimes  \stackrel{\AAA^\otimes}{\longrightarrow} \TT\Big)~\longrightarrow~\AAA(i,U)\quad,
\end{flalign}
for all $(i,U)\in \ovr{\CC(\U)}$. Using also the explicit description
of the orthogonal categories $\ovr{\CC(\U)}$ and $\ovr{\CC(\U\cap \U^\prime)}$
from Lemma \ref{lem:niceCcoverOCat}, one finds that an
object in the comma category $\O_{\pr_1}^\otimes\big/(i,U)$
is given by a tuple $\big(((i_1,i^\prime_1),V_1),\dots,((i_n,i^\prime_n),V_n)\big)$,
with $V_j\in\ovr{\CC(U_{i_j}\cap U^\prime_{i_j^\prime})}\subseteq \ovr{\CC(M)}$ for all $j\in\{1,\dots,n\}$,
together with an operation $ \big((i_1,V_1),\dots,(i_n,,V_n)\big)\to (i,U)$
in the operad $\O_{\ovr{\CC(\U)}}$. Using the factorizations
$(i_{j},V_j) \cong (i,V_j) \to (i,U)$ in $\ovr{\CC(\U)}$, 
we observe that every object in $\O_{\pr_1}^\otimes\big/(i,U)$ is isomorphic to one 
whose underlying tuple is of the form $\big(((i,i^\prime_1),V_1),\dots,((i,i^\prime_n),V_n)\big)$.
This provides an equivalence $\O_{\pr_1}^\otimes\big/(i,U)\simeq \O^\otimes_{j_{\U_i}}/U$
with the comma category of the functor 
$\O^\otimes_{j_{\U_i}} : \O_{\ovr{\CC(\U_i)}}^\otimes \to \O_{\ovr{\CC(U_i)}}^\otimes$
which is associated with the cover $\U_i = \pr_1^{-1}(i)$ of $U_i$.
Under this equivalence, the family of maps in \eqref{eqn:epsilonpr} gets identified
with the counit components $(\epsilon_{\U_i})_{\AAA_i} : j_{\U_i\,!}\,j_{\U_i}^\ast(\AAA_i)\Rightarrow\AAA_i$,
which are isomorphisms because $\AAA_i\in\calHK_{\ovr{\CC}}(U_i)$  satisfies descent, for all $i$.
\end{proof}

We can now prove the main result of this subsection.
\begin{theo}\label{theo:HKstacks}
Suppose that the $2$-functor $\ovr{\CC(-)} : \Loc\to \Cat^\perp$ is a
net domain in the sense of Definition \ref{def:nice2functor} and that Assumption \ref{assu:fastrestricts} holds true.
Then the improved Haag-Kastler-style pseudo-functor $\calHK_{\ovr{\CC}} : \Loc^\op\to \Pr^R$
from Definition \ref{def:improvedHKpseudofunctor} is a stack with respect to the Grothendieck
topology given by all causally convex open covers. 
\end{theo}
\begin{proof}
By Assumption \ref{assu:fastrestricts}, we can present the improved 
Haag-Kastler-style pseudo-functor as a $2$-subfunctor $\calHK_{\ovr{\CC}} \subseteq \HK_{\ovr{\CC}}$.
Using Proposition \ref{prop:caldescentmapadjoint}, we obtain for every object 
$M\in\Loc$ and every causally convex open cover $\U=\{U_i\subseteq M\}$ the adjunction
\begin{flalign}\label{eqn:caldescentadjunction}
\xymatrix@C=3em{
j_{\U\,!}\,:\, \calHK_{\ovr{\CC}}(\U) \ar@<0.75ex>[r]~&~\ar@<0.75ex>[l] \calHK_{\ovr{\CC}}(M)\,:\, j_{\U}^\ast
}
\end{flalign}
whose right adjoint is the canonical functor \eqref{eqn:caldescentmap} to the descent category.
The unit $\eta_{\mathcal{U}}$ of this adjunction is a natural isomorphism by 
Theorem \ref{theo:precostack} and the counit $\epsilon_{\mathcal{U}}$
is a natural isomorphism by Definition \ref{def:improvedHK}
of the full subcategories $\calHK_{\ovr{\CC}}(M)\subseteq \HK_{\ovr{\CC}}(M)$. This implies that \eqref{eqn:caldescentadjunction} 
is an (adjoint) equivalence, for every object $M\in\Loc$ and every causally convex open cover $\U=\{U_i\subseteq M\}$,
hence $\calHK_{\ovr{\CC}}$ is a stack.
\end{proof}

The category of points of the Haag-Kastler-style stack $\calHK_{\ovr{\CC}}$ from Theorem \ref{theo:HKstacks}
is defined similarly to Definition \ref{def:HKpoints} in terms of the category
\begin{flalign}\label{eqn:calHKpoints}
\calHK_{\ovr{\CC}}(\mathrm{pt})\,:=\,\Hom(\Delta\mathbf{1},\calHK_{\ovr{\CC}})\,\in\,\CAT
\end{flalign}
of pseudo-natural transformations from the constant $2$-functor $\Delta\mathbf{1} : \Loc^\op\to\Pr^R$
(which is a stack with respect to the Grothendieck topology given by all causally convex open covers)
to $\calHK_{\ovr{\CC}}: \Loc^\op\to\Pr^R$ and their modifications. 
\begin{propo}\label{prop:calHKpoints}
Suppose that the $2$-functor $\ovr{\CC(-)} : \Loc\to \Cat^\perp$ is a
net domain and that Assumption \ref{assu:fastrestricts} holds true. Then there exists
an equivalence
\begin{flalign}
\calHK_{\ovr{\CC}}(\mathrm{pt}) \,\simeq\, \HK_{\ovr{\CC}}(\mathrm{pt})^{\mathrm{desc}}
\end{flalign}
between the category of points \eqref{eqn:calHKpoints} of the 
Haag-Kastler-style stack $\calHK_{\ovr{\CC}}$ and the full subcategory 
$\HK_{\ovr{\CC}}(\mathrm{pt})^{\mathrm{desc}}\subseteq \HK_{\ovr{\CC}}(\mathrm{pt})$
of the category of points of the Haag-Kastler-style $2$-functor $\HK_{\ovr{\CC}}$
consisting of all objects $(\{\AAA_M\},\{\alpha_f\})\in \HK_{\ovr{\CC}}(\mathrm{pt})$ 
(see also Remark \ref{rem:HKpoints}) such that $\AAA_M\in \calHK_{\ovr{\CC}}(M)\subseteq \HK_{\ovr{\CC}}(M)$
satisfies the descent conditions from Definition \ref{def:improvedHK}, for all $M\in\Loc$.
\end{propo}
\begin{proof}
This follows immediately by using Assumption \ref{assu:fastrestricts}
to present the Haag-Kastler-style stack as a $2$-subfunctor 
$\calHK_{\ovr{\CC}} \subseteq \HK_{\ovr{\CC}}$ and the fact that 
$\calHK_{\ovr{\CC}}(M)\subseteq \HK_{\ovr{\CC}}(M)$
is a full subcategory, for all $M\in\Loc$.
\end{proof}

It remains to verify that the above results apply to at least some of our examples.
As already anticipated above, it is currently not clear to us 
if the Haag-Kastler $2$-functor $\HK$ from Definition \ref{def:HK2functor}
satisfies Assumption \ref{assu:fastrestricts}. However, we have 
the following positive result for the relatively compact Haag-Kastler $2$-functor $\HK^{\mathrm{rc}}$
from Definition \ref{def:RCHK2functor}.
\begin{theo}\label{theo:RCHKstack}
The relatively compact Haag-Kastler $2$-functor $\HK^{\mathrm{rc}}$
from Definition \ref{def:RCHK2functor} satisfies the requirements of Assumption \ref{assu:fastrestricts}.
Hence, as a consequence of Theorem \ref{theo:HKstacks}, the improved
relatively compact Haag-Kastler pseudo-functor
$\calHK^{\mathrm{rc}}$ associated to the net domain $\ovr{\RC(-)}$
is a
stack with respect to the Grothendieck topology given by all causally convex open covers. 
\end{theo}
\begin{proof}
We have to prove that, given any $\Loc$-morphism $f:M\to N$ and any object $\AAA\in\calHK^{\mathrm{rc}}(N)$
which satisfies the descent conditions from Definition \ref{def:improvedHK} for all causally convex
open covers of $N$, the pullback $f^\ast(\AAA)\in \HK^{\mathrm{rc}}(M)$ satisfies these 
descent conditions for all causally convex open covers of $M$.
Recalling the model for the operadic left Kan extension from Appendix \ref{app:operadicLKE},
this means that we have to show that, for every causally convex open cover $\U=\{U_i\subseteq M\}$ of $M$, 
the canonical map
\begin{flalign}\label{eqn:HKstackdesccondition}
((\epsilon_\U)_{f^\ast(\AAA)})_U\,:\,\colim\Big(\O_{j_{\U}}^\otimes/U \longrightarrow 
\O_{\ovr{\RC(\U)}}^\otimes \stackrel{\O_{j_{\U}}^\otimes}{\longrightarrow}\O_{\ovr{\RC(M)}}^\otimes
\stackrel{\O_{f}^\otimes}{\longrightarrow} \O_{\ovr{\RC(N)}}^\otimes
\stackrel{\AAA^\otimes}{\longrightarrow}\TT\Big)~\longrightarrow~\AAA\big(f(U)\big)
\end{flalign}
is an isomorphism, for all $U\in\RC(M)$. Our proof strategy is to construct,
for every fixed $U\in\RC(M)$, a causally convex open cover $\V$ of $N$ such that the descent conditions
for $\AAA$ in this cover imply that \eqref{eqn:HKstackdesccondition} is an isomorphism.
For this we use that $U\subseteq M$ is a relatively compact
causally convex open subset, hence the image of its closure
$f(\mathrm{cl}(U))\subseteq N$ is a compact subset of $N$. 
We choose any open cover $\mathcal{W} = \{W_j\subseteq N\setminus f(\mathrm{cl}(U))\}$
of the complement such that each $W_j\subseteq N$ is causally convex in $N$.
(Such cover exists
because $N$ is globally hyperbolic and hence strongly causal,
so the open set $N\setminus f(\mathrm{cl}(U))\subseteq N$ contains a causally convex open neighborhood of each of its points.)
From this we define the causally convex open cover 
$\V := f(\U)\cup \mathcal{W} = \{f(U_i)\subseteq N\}\cup \{W_j\subseteq N\}$ of $N$.
It is important to observe that, by construction, the restriction $f^{-1}\V\vert_U = \U\vert_U $
to $U\subseteq M$ of the pullback cover agrees with the restriction
of the given cover $\U$. This implies that we have an isomorphism
\begin{flalign}
f\,:\, \O_{j_{\U}}^\otimes/U~\stackrel{\cong}{\longrightarrow}~\O_{j_{\V}}^\otimes/f(U)
\end{flalign}
between the comma categories by taking images under $f$. Moreover, by direct inspection
one verifies that the diagram
\begin{flalign}
\begin{gathered}
\xymatrix{
\ar[d]^-{\cong}_-{f} \O_{j_{\U}}^\otimes/U \ar[r] ~&~ \O_{\ovr{\RC(\U)}}^\otimes \ar[r]^-{\O_{j_{\U}}^\otimes} ~&~ \O_{\ovr{\RC(M)}}^\otimes\ar[d]^-{\O_{f}^\otimes}\\
\O_{j_{\V}}^\otimes/f(U) \ar[r]~&~ \O_{\ovr{\RC(\V)}}^\otimes \ar[r]_-{\O_{j_{\V}}^\otimes}~&~ \O_{\ovr{\RC(N)}}^\otimes
}
\end{gathered}
\end{flalign}
commutes. This allows us to identify the canonical map \eqref{eqn:HKstackdesccondition}
with the canonical map
\begin{flalign}
((\epsilon_\V)_{\AAA})_{f(U)}\,:\,\colim\Big(\O_{j_{\V}}^\otimes/f(U) 
\longrightarrow  \O_{\ovr{\RC(\V)}}^\otimes \stackrel{\O_{j_{\V}}^\otimes}{\longrightarrow}\O_{\ovr{\RC(N)}}^\otimes
\stackrel{\AAA^\otimes}{\longrightarrow}\TT\Big)~\longrightarrow~\AAA\big(f(U)\big)\quad,
\end{flalign}
which is an isomorphism because $\AAA\in\calHK^{\mathrm{rc}}(N)$
satisfies the descent conditions from Definition \ref{def:improvedHK} for all causally convex open covers
$\V$ of $N$.
\end{proof}

\begin{rem}\label{rem:RCHKstack}
Our proof of Theorem \ref{theo:RCHKstack} does not
generalize in any evident way to the Haag-Kastler $2$-functor $\HK$ from Definition \ref{def:HK2functor}
because in the construction of the extended cover 
$\V := f(\U)\cup \mathcal{W} = \{f(U_i)\subseteq N\}\cup \{W_j\subseteq N\}$ of $N$,
which has the crucial restriction property $f^{-1}\V\vert_U = \U\vert_U$,
it was essential to assume that $U\subseteq M$ is relatively compact. It is therefore
currently not clear to us if the improved Haag-Kastler pseudo-functor $\calHK$ is a stack too.
\end{rem}

\begin{rem}\label{rem:stackification}
An interesting question is whether or not there exists any relationship 
between our AQFT-inspired `descent improvement' construction presented 
in this section and an abstract stackification construction.
In the context of $\CAT$-valued stacks, the stackification
of a pseudo-functor $X: \Loc^\op\to \CAT$ can be described explicitly by forming 
suitable bicolimits in $\CAT$, see e.g.\ \cite{Street}, but it 
is unlikely that the same kind of construction provides a stackification
in the context of $\Pr^R$-valued stacks because bicolimits 
in $\Pr^R$ behave very differently to those in $\CAT$, see Construction \ref{constr:computingbicolimits}.
To the best of our knowledge, stackification in the context of $\Pr^R$-valued stacks 
has not been studied in the literature and it is even unclear to us if it exists.
\sk

Since stackification in the context of $\Pr^R$-valued stacks (if it exists) 
should be defined by a universal property, namely the property
of being left adjoint (in the bicategorical sense) to the forgetful $2$-functor from stacks to pseudo-functors, one can
reformulate the question of whether or not our AQFT-inspired construction describes a stackification
into the following problem: Under our hypotheses, the Haag-Kastler-style stack
$\calHK_{\ovr{\CC}} \subseteq \HK_{\ovr{\CC}}$ is a $2$-subfunctor of the Haag-Kastler-style
$2$-functor with all full subcategory inclusions being left adjoint functors. Passing over to 
the corresponding right adjoints defines a pseudo-natural transformation
$\pi : \HK_{\ovr{\CC}} \Rightarrow \calHK_{\ovr{\CC}}$ of pseudo-functors from $\Loc^\op$
to $\Pr^R$ whose components are the coreflectors from Proposition \ref{propo:bilimovercover}.
Then $\pi$ exhibits $\calHK_{\ovr{\CC}}$ as a stackification of $\HK_{\ovr{\CC}}$ if and only if,
for every stack $X:\Loc^\op\to \Pr^R$, there exists an equivalence
\begin{flalign}
\mathrm{HOM}\big(\calHK_{\ovr{\CC}},X\big)~\simeq~\mathrm{HOM}\big(\HK_{\ovr{\CC}},X\big)
\end{flalign}
between the categories of pseudo-natural transformations and modifications,
which is pseudo-natural in $X$. We are currently not aware of any 
suitable strategies to verify or disprove this universal property, 
hence the question of whether or not the Haag-Kastler-style stack $\calHK_{\ovr{\CC}}$ 
arises as a stackification of the Haag-Kastler-style
$2$-functor $\HK_{\ovr{\CC}}$ unfortunately remains open.
\end{rem}

\subsubsection{\label{subsubsec:descent:timeslice}The case of time-slice}
Throughout this subsection, let us assume that  $\ovr{\CC(-)} : \Loc\to\Cat^\perp$ is a localized
net domain. As a consequence of the similar formal properties of localized net domains 
from Definition \ref{def:Wnice2functor} and net domains from Definition \ref{def:nice2functor}, 
with the only difference given by the additional assumption of Cauchy development stability
in the localized case, all constructions and most of the results from Subsection \ref{subsubsec:notimeslicedescent}
directly generalize to the present case if one consistently replaces general causally convex open covers
by $D$-stable causally convex open covers.
We shall briefly collect the relevant definitions
and results in the present case, without repeating the proofs. 
\sk

The analogue of Definition \ref{def:improvedHK}
in the present case is given as follows.
\begin{defi}\label{def:WimprovedHK}
Let $\ovr{\CC(-)} : \Loc\to\Cat^\perp$ be a localized net domain.
For every object $M\in\Loc$, we denote by
\begin{subequations}\label{eqn:WimprovedHKdescentcondition}
\begin{flalign}
\calHK_{\ovr{\CC}}(M)\,\subseteq\,\HK_{\ovr{\CC}}(M)
\end{flalign}
the full subcategory consisting of all objects $\AAA\in\HK_{\ovr{\CC}}(M)$ 
which satisfy the following descent conditions: For every
$D$-stable causally convex open cover $\U=\{U_i\subseteq M\}$, 
the $\AAA$-component of the counit
\begin{flalign}
(\epsilon_{\U})_\AAA \,:\, j_{\U\,!}\,j_{\U}^\ast(\AAA)~
\stackrel{\cong}{\Longrightarrow}~\AAA
\end{flalign}
\end{subequations}
of the adjunction
$j_{\U\,!} : \HK_{\ovr{\CC}}(\U) \rightleftarrows \HK_{\ovr{\CC}}(M) : j_\U^\ast$
of \eqref{eqn:descentadjunction} is an isomorphism in $\HK_{\ovr{\CC}}(M)$.
\end{defi}

By the same arguments as in the proof of Proposition \ref{propo:bilimovercover},
one can show that the category from Definition \ref{def:WimprovedHK}
arises as the bilimit of a pseudo-functor $\HK_{\ovr{\CC}}^\dagger:\mathbf{Dcov}(M)\to \Pr^L$,
which in the present case is defined on the full subcategory $\mathbf{Dcov}(M)\subseteq \mathbf{cov}(M)$
of $D$-stable causally convex open covers. (Note that this subcategory has a terminal
object, given by the coarsest cover $\{M\subseteq M\}$.) From this one concludes that
$\calHK_{\ovr{\CC}}(M) \subseteq \HK_{\ovr{\CC}}(M)$ is a locally presentable category
which is embedded as a coreflective full subcategory into $\HK_{\ovr{\CC}}(M)$.
Leveraging pseudo-functoriality of the bilimits over the categories of
$D$-stable causally convex open covers, one obtain the pseudo-functor 
\begin{subequations}\label{eqn:adjointimprovedW}
\begin{flalign}
\calHK_{\ovr{\CC}}^\dagger\,:\,\Loc~\longrightarrow~\Pr^L
\end{flalign} 
which assigns to each object $M\in\Loc$ the locally presentable category
$\calHK_{\ovr{\CC}}^\dagger(M):=\calHK_{\ovr{\CC}}(M)\in\Pr^L$ from Definition \ref{def:WimprovedHK}
and to each $\Loc$-morphism $f:M\to N$ the restriction
\begin{flalign}
\calHK_{\ovr{\CC}}^\dagger(f)\,:=\,f_! \,:\, \calHK_{\ovr{\CC}}(M)~\longrightarrow~\calHK_{\ovr{\CC}}(N)
\end{flalign}
\end{subequations}
of the left adjoint $f_! : \HK_{\ovr{\CC}}(M)\to \HK_{\ovr{\CC}}(N)$ to the full subcategories
$\calHK_{\ovr{\CC}}(M)\subseteq \HK_{\ovr{\CC}}(M)$ and $\calHK_{\ovr{\CC}}(N)\subseteq \HK_{\ovr{\CC}}(N)$.
The analogue of Definition \ref{def:improvedHKpseudofunctor} in 
the present case is then as follows.
\begin{defi}\label{def:WimprovedHKpseudofunctor}
Let $\ovr{\CC(-)} : \Loc\to\Cat^\perp$ be a localized net domain.
The \textit{improved Haag-Kastler-style pseudo-functor}
\begin{flalign}
\calHK_{\ovr{\CC}}\,:=\, \calHK_{\ovr{\CC}}^{\dagger\dagger}\,:\, \Loc^\op~\longrightarrow~\Pr^R
\end{flalign}
is defined as the adjoint via \eqref{eqn:LtoR} 
of the pseudo-functor $\calHK_{\ovr{\CC}}^\dagger$ in \eqref{eqn:adjointimprovedW}.
\end{defi}

As in the previous subsection, this pseudo-functor
is difficult to work with, which is why we introduce
an analogue of Assumption 
\ref{assu:fastrestricts}, but weaker since it 
only applies to $\Loc$-morphisms with $D$-stable image.
\begin{assu}\label{assu:Wfastrestricts}
We assume that, for every $\Loc$-morphism $f:M\to N$ whose image 
$f(M)\subseteq N$ is $D$-stable, i.e.\ $D_N(f(M))=f(M)$, the pullback
functor $f^\ast :\HK_{\ovr{\CC}}(N)\to \HK_{\ovr{\CC}}(M)$
restricts to a functor $f^\ast :\calHK_{\ovr{\CC}}(N)\to \calHK_{\ovr{\CC}}(M)$
between the full subcategories $\calHK_{\ovr{\CC}}(N)\subseteq \HK_{\ovr{\CC}}(N)$
and $\calHK_{\ovr{\CC}}(M)\subseteq \HK_{\ovr{\CC}}(M)$ from Definition \ref{def:WimprovedHK}.
\end{assu}

Provided that Assumption \ref{assu:Wfastrestricts} holds true,
one can choose a model for the improved Haag-Kastler-style pseudo-functor
such that $\calHK_{\ovr{\CC}}(f) = f^\ast : \calHK_{\ovr{\CC}}(N)\to \calHK_{\ovr{\CC}}(M)$
is the restriction of the pullback functor, for all $\Loc$-morphisms $f:M\to N$
with $D$-stable image. (For $\Loc$-morphisms whose image is not $D$-stable,
the pseudo-functorial structure is more complicated because one has to use 
coreflectors as in \eqref{eqn:fpullcoreflected}.)
This partially simplified description of the pseudo-functor $\calHK_{\ovr{\CC}}$ however
suffices to conclude that, for every object 
$M\in\Loc$ and every $D$-stable causally convex open cover $\U=\{U_i\subseteq M\}$, 
the descent category of $\calHK_{\ovr{\CC}}$ is given by the full subcategory
\begin{flalign}\label{eqn:caldescentcatW}
\calHK_{\ovr{\CC}}(\U)\,\subseteq\,\HK_{\ovr{\CC}}(\U)
\end{flalign}
consisting of all objects $(\{\AAA_i\},\{\varphi_{ij}\})\in \HK_{\ovr{\CC}}(\U)$
such that $\AAA_i\in \calHK_{\ovr{\CC}}(U_i)\subseteq \HK_{\ovr{\CC}}(U_i)$
lies in the full subcategory of objects satisfying the descent conditions 
from Definition \ref{def:WimprovedHK}, for all $i$. The canonical functor 
\begin{flalign}\label{eqn:caldescentmapW}
j_{\U}^\ast\,:\,\calHK_{\ovr{\CC}}(M)~\longrightarrow~ \calHK_{\ovr{\CC}}(\U)
\end{flalign}
to the descent category is given by restricting the right adjoint 
$j_{\U}^\ast:\HK_{\ovr{\CC}}(M) \to \HK_{\ovr{\CC}}(\U)$ from \eqref{eqn:descentadjunction}
to the full subcategories $\calHK_{\ovr{\CC}}(M)\subseteq \HK_{\ovr{\CC}}(M)$
and $\calHK_{\ovr{\CC}}(\U)\subseteq \HK_{\ovr{\CC}}(\U)$.
With the same proof as in Proposition \ref{prop:caldescentmapadjoint},
one then shows the following result.
\begin{propo}\label{prop:Wcaldescentmapadjoint}
Suppose that the $2$-functor $\ovr{\CC(-)} : \Loc\to \Cat^\perp$ is a
localized net domain and that Assumption \ref{assu:Wfastrestricts} holds true.
Then, for every object $M\in\Loc$ and every $D$-stable causally convex open cover $\U=\{U_i\subseteq M\}$, 
the left adjoint $j_{\U\,!} :\HK_{\ovr{\CC}}(\U) \to \HK_{\ovr{\CC}}(M)$ from \eqref{eqn:descentadjunction}
restricts to the full subcategories $\calHK_{\ovr{\CC}}(\U)\subseteq \HK_{\ovr{\CC}}(\U)$
and $\calHK_{\ovr{\CC}}(M)\subseteq \HK_{\ovr{\CC}}(M)$, and thereby defines
a left adjoint $j_{\U\,!} :\calHK_{\ovr{\CC}}(\U) \to \calHK_{\ovr{\CC}}(M)$
for the functor \eqref{eqn:caldescentmapW}.
\end{propo}

The main result of the present subsection is then similar to Theorem \ref{theo:HKstacks}.
\begin{theo}\label{theo:WHKstacks}
Suppose that the $2$-functor $\ovr{\CC(-)} : \Loc\to \Cat^\perp$ is a localized
net domain in the sense of Definition \ref{def:Wnice2functor} and that 
Assumption \ref{assu:Wfastrestricts} holds true.
Then the improved Haag-Kastler-style pseudo-functor $\calHK_{\ovr{\CC}} : \Loc^\op\to \Pr^R$
from Definition \ref{def:WimprovedHKpseudofunctor} is a stack with respect to the Grothendieck
topology given by all $D$-stable causally convex open covers. 
\end{theo}
\begin{proof}
By the same arguments as in the proof
of Theorem \ref{theo:HKstacks}, this 
follows directly from Proposition \ref{prop:Wcaldescentmapadjoint}
and Theorem \ref{theo:Wprecostack}.
\end{proof}

The result from Proposition \ref{prop:calHKpoints}
about the category of points of the Haag-Kastler-style 
stack does not generalize to the present case because
Assumption \ref{assu:Wfastrestricts} is too weak to imply that
$\calHK_{\ovr{\CC}}$ can be presented 
as $2$-subfunctor of $\HK_{\ovr{\CC}}$. However, we have
the following weaker result which does not rely on any additional
assumptions but characterizes only a full subcategory of the category of points $\calHK_{\ovr{\CC}}(\mathrm{pt})$.
\begin{propo}\label{prop:WcalHKpoints}
Suppose that the $2$-functor $\ovr{\CC(-)} : \Loc\to \Cat^\perp$ is a localized net domain.
Then there exists a fully faithful functor
\begin{flalign}\label{eqn:WcalHKpoints}
\HK_{\ovr{\CC}}(\mathrm{pt})^{\mathrm{desc}}\,\longrightarrow\, \calHK_{\ovr{\CC}}(\mathrm{pt})
\end{flalign}
from the full subcategory $\HK_{\ovr{\CC}}(\mathrm{pt})^{\mathrm{desc}}\subseteq \HK_{\ovr{\CC}}(\mathrm{pt})$
of the category of points of the Haag-Kastler-style $2$-functor $\HK_{\ovr{\CC}}$
consisting of all objects $(\{\AAA_M\},\{\alpha_f\})\in \HK_{\ovr{\CC}}(\mathrm{pt})$ 
(see also Remark \ref{rem:HKpoints}) such that $\AAA_M\in \calHK_{\ovr{\CC}}(M)\subseteq \HK_{\ovr{\CC}}(M)$
satisfies the descent conditions from Definition \ref{def:WimprovedHK}, for all $M\in\Loc$,
to the category of points \eqref{eqn:calHKpoints} of the improved Haag-Kastler-style pseudo-functor $\calHK_{\ovr{\CC}}$.
\end{propo}
\begin{proof}
We provide a direct construction of the fully faithful functor \eqref{eqn:WcalHKpoints}.
For this we use the model for the Haag-Kastler-style pseudo-functor $\calHK_{\ovr{\CC}}$ which is
given by the coreflected pullback functors in \eqref{eqn:fpullcoreflected}, i.e.\ 
$\calHK_{\ovr{\CC}}(f) = \pi_M \, f^\ast\,\iota_N $ for all $\Loc$-morphisms $f:M\to N$. 
Note that this model is only pseudo-functorial with coherences 
\begin{subequations}
\begin{flalign}\label{eqn:tmpdiagram}
\resizebox{.9 \textwidth}{!}{$
\begin{gathered}
\xymatrix@C=3.75em@R=2em{
\ar@{=}[d]\calHK_{\ovr{\CC}}(g\,f) \ar@{=>}[rrr]^-{\cong} ~&~ ~&~ ~&~ 
\calHK_{\ovr{\CC}}(f)\,\calHK_{\ovr{\CC}}(g) \ar@{=}[d]\\
 \pi_M f^\ast g^\ast \iota_O \ar@{=>}[r]_-{\cong} ~&~  
 f^\ast \pi_N g^\ast \iota_O \ar@{=>}[r]_-{f^\ast\eta_N\pi_N g^\ast\iota_O}^-{\cong} ~&~
 f^\ast \pi_N \iota_N \pi_N g^\ast\iota_O \ar@{=>}[r]_-{\cong}~&~ \pi_M f^\ast \iota_N \pi_N g^\ast\iota_O
}
\end{gathered}
$}
\end{flalign}
and 
\begin{flalign}
\xymatrix{
\eta_M \,:\, \id_{\calHK_{\ovr{\CC}}(M)} \ar@{=>}[r]^-{\cong} ~&~ \pi_M\,\iota_M\,=\, \calHK_{\ovr{\CC}}(\id_M)
}
\end{flalign}
\end{subequations}
given by the units of the coreflection adjunctions 
$\iota_M : \calHK_{\ovr{\CC}}(M) \rightleftarrows \HK_{\ovr{\CC}}(M): \pi_M$, for all $M\in\Loc$.
The unlabeled isomorphisms in the bottom row of \eqref{eqn:tmpdiagram} are
the natural isomorphisms $\pi_M\,f^\ast \cong f^\ast\, \pi_N$ which are a consequence of
the commutativity property $\iota_N \,f_! = f_!\,\iota_M$ of the left adjoint functors.
In the following we shall suppress all $\iota$ because these are just full subcategory inclusions.
\sk

For every object $(\{\AAA_M\},\{\alpha_f\})\in \HK_{\ovr{\CC}}(\mathrm{pt})^{\mathrm{desc}}$,
the fact that $\AAA_M\in \calHK_{\ovr{\CC}}(M)\subseteq \HK_{\ovr{\CC}}(M)$
lies in the full subcategory and the isomorphism $\alpha_f : \AAA_M\stackrel{\cong}{\Longrightarrow} f^\ast(\AAA_N)$
in $\HK_{\ovr{\CC}}(M)$ imply that $f^\ast(\AAA_N)\in \calHK_{\ovr{\CC}}(M)\subseteq \HK_{\ovr{\CC}}(M)$
lies in the full subcategory, for every $\Loc$-morphism $f:M\to N$. We define
the functor \eqref{eqn:WcalHKpoints} on objects by sending
$(\{\AAA_M\},\{\alpha_f\})\in \HK_{\ovr{\CC}}(\mathrm{pt})^{\mathrm{desc}}$ to the tuple
\begin{flalign}
\Big(\Big\{\AAA_M\in \calHK_{\ovr{\CC}}(M)\Big\}, \,
\Big\{\tilde{\alpha}_f : \AAA_M\stackrel{\alpha_f}{\Longrightarrow}f^\ast(\AAA_N)
\stackrel{\eta_{f^\ast(\AAA_N)}}{\Longrightarrow}   \pi_M f^\ast(\AAA_N)\Big\}\Big)\quad.
\end{flalign}
One directly checks that this tuple satisfies the conditions in 
Remark \ref{rem:HKpoints}, hence it defines an object $(\{\AAA_M\},\{\tilde{\alpha}_f\})\in 
\calHK_{\ovr{\CC}}(\mathrm{pt})$. The action of the functor \eqref{eqn:WcalHKpoints} 
on morphisms $\{\zeta_M\} : (\{\AAA_M\},\{\alpha_f\})\Rightarrow (\{\BBB_M\},\{\beta_f\})$
in $\HK_{\ovr{\CC}}(\mathrm{pt})^{\mathrm{desc}}$ is given by the same tuple
of maps  $\{\zeta_M\} : (\{\AAA_M\},\{\tilde{\alpha}_f\})\Rightarrow (\{\BBB_M\},\{\tilde{\beta}_f\})$.
One directly checks that this tuple satisfies the conditions in Remark \ref{rem:HKpoints}.
Fully faithfulness then follows immediately from the fact that
$\calHK_{\ovr{\CC}}(M)\subseteq \HK_{\ovr{\CC}}(M)$ is a full subcategory,
for all $M\in\Loc$.
\end{proof}

Analogously to the case where no time-slice axiom is implemented,
see Theorem \ref{theo:RCHKstack} and Remark \ref{rem:RCHKstack},
we can confirm the hypotheses of
Theorem \ref{theo:WHKstacks}
only for the
time-sliced relatively compact Haag-Kastler $2$-functor $\HK^{\mathrm{rc},W}$
from Definition \ref{def:RCHK2functortimeslice}.
We currently
do not know if the improvement construction from Definition \ref{def:WimprovedHKpseudofunctor}
applied to the (non-relatively compact) time-sliced Haag-Kastler $2$-functor $\HK^W$
from Definition \ref{def:HK2functortimeslice} defines a stack.
\begin{theo}\label{theo:WRCHKstack}
The time-sliced relatively compact Haag-Kastler $2$-functor $\HK^{\mathrm{rc},W}$
from Definition \ref{def:RCHK2functortimeslice} 
satisfies the requirements of Assumption \ref{assu:Wfastrestricts}.
Hence, as a consequence of Theorem \ref{theo:WHKstacks}, the improved
time-sliced relatively compact Haag-Kastler pseudo-functor $\calHK^{\mathrm{rc},W}$
associated to the localized net domain $\ovr{\RC(-)}[W_{\mathrm{rc},(-)}^{-1}]$
is a
stack with respect to the Grothendieck topology given by all $D$-stable
causally convex open covers. 
\end{theo}
\begin{proof}
We use the same proof strategy as in Theorem \ref{theo:RCHKstack},
but there are additional technical aspects (treated in Appendix \ref{app:Lorentz_geometry_details}) 
arising from 1.)~the $D$-stability requirement for causally convex open covers, and 
2.)~the fact that morphisms $U\to U^\prime$ 
in the localized orthogonal categories $\ovr{\RC(M)}[W_{\mathrm{rc},M}^{-1}]$
from Example \ref{ex:localizations} exist whenever $U\subseteq D_M(U^\prime)$.
This requires us to construct, for every $\Loc$-morphism $f: M\to N$ with 
$D$-stable image $f(M)\subseteq N$, every $D$-stable causally convex open cover 
$\U=\{U_i\subseteq M\}$ of $M$ and every relatively compact causally 
convex open subset $U\subseteq M$, a $D$-stable causally convex open cover 
$\V$ of $N$ which satisfies the restriction property $f^{-1}\V\vert_{D_M(U)} = \U\vert_{D_M(U)}$
on the Cauchy development $D_M(U)\subseteq M$.
As an immediate corollary of Proposition \ref{prop:Cauchy_development:small_D-stable_neighbourhoods},
one can find an open cover $\mathcal{W} = \{W_j\subseteq N\setminus \mathrm{cl}(f(D_M(U)))\}$
of the complement of the closure $\mathrm{cl}(f(D_M(U)))\subseteq N$
such that each $W_j\subseteq N$ is causally convex and $D$-stable in $N$.
Using further that $\mathrm{cl}(f(D_M(U))) \subseteq \mathrm{cl}(D_N(f(U))) \subseteq f(M)$ 
is contained in the image of $f$
by Proposition \ref{prop:Cauchy_development:closure_in_D_stable_image},
we then obtain a $D$-stable causally convex open cover
$\V := f(\U) \cup \mathcal{W} = \{f(U_i)\subseteq N\}\cup\{W_j\subseteq N\}$ of $N$
which, by construction, satisfies the desired restriction property 
$f^{-1}\V\vert_{D_M(U)} = \U\vert_{D_M(U)}$. The remaining steps in the proof
are then identical to Theorem \ref{theo:RCHKstack}.
\end{proof}

\subsection{\label{subsec:KGexample}Exhibiting examples of points}
The aim of this subsection is to exhibit points
of the relatively compact Haag-Kastler stack $\calHK^{\mathrm{rc}}$
from Theorem \ref{theo:RCHKstack} and of the time-sliced relatively compact
Haag-Kastler stack $\calHK^{\mathrm{rc},W}$ from Theorem \ref{theo:WRCHKstack}.
These points correspond to AQFTs which are 
presented by generators and relations that satisfy
certain simpler descent conditions. Explicit examples include 
free (i.e.\ non-interacting) AQFTs such as the free Klein-Gordon quantum field.
\sk

We start with some technical preparations for these results.
Let $\ovr{\CC(-)} : \Loc\to \Cat^\perp$ be a net domain 
(see Definition \ref{def:nice2functor}) 
or a localized net domain (see Definition \ref{def:Wnice2functor}).
Given any causally convex open cover $\U = \{U_i\subseteq M\}$ of any object $M\in \Loc$, 
which we assume to be $D$-stable in the localized case, there exists a diagram of adjunctions
\begin{flalign}\label{eqn:bigadjunctiondiagram}
\begin{gathered}
\xymatrix@R=3em@C=3em{
\ar@<1.2ex>[d]^-{\subseteq} \HK_{\ovr{\CC}}(\U) \ar@<1.2ex>[r]_-{\perp}^-{j_{\U\,!}}~&~
\ar@<1.2ex>[l]^-{j_\U^\ast} \HK_{\ovr{\CC}}(M) \ar@<1.2ex>[d]^-{\subseteq}\\
\ar@<1.2ex>[u]_-{\dashv}^-{p_{\U}} \ar@<1.2ex>[d]^-{\mathsf{U}_{\U}} \Fun\big(\CC(\U),\Alg_{\mathrm{uAs}}(\TT)\big) \ar@<1.2ex>[r]_-{\perp}^-{\Lan_{j_\U}}~&~
\ar@<1.2ex>[u]_-{\dashv}^-{p_{M}} \ar@<1.2ex>[l]^-{j_\U^\ast} \Fun\big(\CC(M),\Alg_{\mathrm{uAs}}(\TT)\big) \ar@<1.2ex>[d]^-{\mathsf{U}_{M}}\\
\ar@<1.2ex>[u]_-{\dashv}^-{\mathsf{F}_{\U}}\Fun(\CC(\U),\TT) \ar@<1.2ex>[r]_-{\perp}^-{\lan_{j_\U}}~&~
\ar@<1.2ex>[u]_-{\dashv}^-{\mathsf{F}_{M}} \ar@<1.2ex>[l]^-{j_\U^\ast} \Fun(\CC(M),\TT)
}
\end{gathered}\qquad.
\end{flalign}
The top horizontal adjunction $j_{\U\,!}\dashv j_{\U}^\ast$
is the one determining the descent conditions from Definitions
\ref{def:improvedHK} and \ref{def:WimprovedHK} for the improved Haag-Kastler-style
pseudo-functor $\calHK_{\ovr{\CC}}$, while the middle and bottom horizontal
adjunctions $\Lan_{j_\U}\dashv j_{\U}^\ast$ and $\lan_{j_\U}\dashv j_{\U}^\ast$
are given by left Kan extensions of functors along the fully faithful functors
$j_\U : \CC(\U)\to\CC(M)$. All horizontal left adjoints are fully faithful functors.
The top vertical adjunctions $p_\U\dashv {\subseteq}$ and $p_M\dashv {\subseteq}$
describe the reflectors for the full subcategory inclusions $\HK_{\ovr{\CC}}(\U)\subseteq 
\Fun\big(\CC(\U),\Alg_{\mathrm{uAs}}(\TT)\big)$ and $\HK_{\ovr{\CC}}(M)\subseteq 
\Fun\big(\CC(M),\Alg_{\mathrm{uAs}}(\TT)\big)$. The left adjoints of these 
adjunctions enforce the $\perp$-commutativity axiom and they appeared before
in \cite[Section 4.1]{BSWoperad} under the name $\perp$-abelianizations.
The bottom vertical adjunctions $\mathsf{F}_\U\dashv \mathsf{U}_{\U}$
and $\mathsf{F}_M\dashv \mathsf{U}_{M}$ are the free-forget adjunctions for unital associative algebras
in the functor categories $\Fun(\CC(\U),\TT)$ and $\Fun(\CC(M),\TT)$. Concretely,
given any object $X\in \Fun(\CC(M),\TT)$, then $\mathsf{F}_M(X) \in \Fun\big(\CC(M),\Alg_{\mathrm{uAs}}(\TT)\big)$
is defined object-wise $\mathsf{F}_M(X)(U) := \mathsf{F}_{\mathsf{uAs}}(X(U))$ by taking
the free unital associative algebra over $X(U)\in \TT$, for all $U\in\CC(M)$.
(The functor $\mathsf{F}_\U$ is defined similarly by taking object-wise free unital associative algebras.)
Let us record the following commutativity properties of the diagram \eqref{eqn:bigadjunctiondiagram} of adjunctions:
\begin{itemize}
\item[(1)] The top square of right adjoints commutes $\subseteq\,j_{\U}^\ast = j_{\U}^\ast\,\subseteq$.
Hence, also the top square of left adjoints commutes up to a natural isomorphism 
$j_{\U\,!}\,p_\U \cong p_M\,\Lan_{j_\U}$.

\item[(2)] The bottom square of right adjoints commutes $j_{\U}^\ast\,\mathsf{U}_M = \mathsf{U}_{\U} \,j_{\U}^\ast$.
Hence, also the bottom square of left adjoints commutes up to a natural isomorphism
$\Lan_{j_\U}\,\mathsf{F}_{\U} \cong \mathsf{F}_M\,\lan_{j_\U}$.

\item[(3)] In the bottom square we have also
$\mathsf{F}_{\U}\,j_\U^\ast = j_{\U}^\ast \,\mathsf{F}_M$ 
because the free algebra functors $F_{\U}$ and $F_{M}$ are defined object-wise,
hence they commute with the pullback functors along $j_\U : \CC(\U)\to\CC(M)$.
\end{itemize}
Let us consider two objects $\LLL_M,\RRR_M\in \Fun(\CC(M),\TT)$
and two parallel morphisms $\tilde{r}_1^M, \tilde{r}_2^M: \RRR_M\Rightarrow \mathsf{U}_M\mathsf{F}_M(\LLL_M)$
in $\Fun(\CC(M),\TT)$. We define the object
\begin{flalign}\label{eqn:genrelAQFT}
\AAA_M\,:=\,\colim_{\Alg}^{}\Big(\xymatrix@C=3.5em{
\mathsf{F}_M(\RRR_M) \ar@{=>}@<0.75ex>[r]^-{r^M_1}\ar@{=>}@<-0.75ex>[r]_-{r^M_2}~&~ \mathsf{F}_M(\LLL_M)
}\Big)\,\in\, \Fun\big(\CC(M),\Alg_{\mathrm{uAs}}(\TT)\big)
\end{flalign}
by taking the colimit (i.e.\ coequalizer) in $\Fun\big(\CC(M),\Alg_{\mathrm{uAs}}(\TT)\big)$,
where $r^M_1,r^M_2$ denote the adjuncts of $\tilde{r}_1^M,\tilde{r}_2^M$ 
with respect to the free-forget adjunction $\mathsf{F}_M\dashv \mathsf{U}_M$.
We assume that $\AAA_M$ satisfies the $\perp$-commutativity axiom from Definition \ref{def:AQFT},
hence we obtain an object $\AAA_M\in \HK_{\ovr{\CC}}(M)\subseteq \Fun\big(\CC(M),\Alg_{\mathrm{uAs}}(\TT)\big)$
which we interpret as an AQFT that is presented by the generators $\LLL_M$
and the relations $\tilde{r}_1^M,\tilde{r}_2^M$ .
\sk

Our goal is to establish criteria for this object to satisfy
the descent conditions from Definitions \ref{def:improvedHK} and \ref{def:WimprovedHK},
i.e.\ criteria such that the counit component 
$(\epsilon_{\U})_{\AAA_M} : j_{\U\,!}\,j_{\U}^\ast(\AAA_M)\Rightarrow \AAA_M$ 
is an isomorphism in $\HK_{\ovr{\CC}}(M)$. Using 
commutativity of the top square of right adjoints in \eqref{eqn:bigadjunctiondiagram},
we can compute the pullback $j_\U^\ast(\AAA_M)\in \HK_{\ovr{\CC}}(\U)$ by 
\begin{flalign}
\nn j_\U^\ast(\AAA_M)\,&=\, j_\U^{\ast}\Big(\colim_{\Alg}^{}\Big(\xymatrix@C=3.5em{
\mathsf{F}_M(\RRR_M) \ar@{=>}@<0.75ex>[r]^-{r^M_1}\ar@{=>}@<-0.75ex>[r]_-{r^M_2}~&~ \mathsf{F}_M(\LLL_M)
}\Big)\Big)\\
\,&=\,\colim_{\Alg}^{}\Big(\xymatrix@C=3.5em{
j_\U^{\ast}\mathsf{F}_M(\RRR_M) \ar@{=>}@<0.75ex>[r]^-{j_\U^{\ast}(r^M_1)}\ar@{=>}@<-0.75ex>[r]_-{j_\U^{\ast}(r^M_2)}~&~ j_\U^{\ast}\mathsf{F}_M(\LLL_M)
}\Big)\quad,
\end{flalign}
where in the second step we used that colimits in functor categories are computed object-wise,
hence they commute with pullback functors. Since the top vertical adjunctions
in \eqref{eqn:bigadjunctiondiagram} exhibit reflective full subcategory inclusions,
we can model the top horizontal left adjoint by $j_{\U\,!} := p_M\,\Lan_{j_\U}\,\subseteq$,
see also \cite[Proposition 4.3]{BSWoperad} for more details on this point.
We then compute
\begin{flalign}
\nn j_{\U\,!}j_\U^\ast(\AAA_M)\,&=\, p_M\,\Lan_{j_\U}\Big(\colim_{\Alg}^{}\Big(\xymatrix@C=5em{
j_\U^{\ast}\mathsf{F}_M(\RRR_M) \ar@{=>}@<0.75ex>[r]^-{j_\U^{\ast}(r^M_1)}\ar@{=>}@<-0.75ex>[r]_-{j_\U^{\ast}(r^M_2)}~&~ j_\U^{\ast}\mathsf{F}_M(\LLL_M)
}\Big)\Big)\\
\,&\cong\, p_M\Big(\colim_{\Alg}^{}\Big(\xymatrix@C=5em{
\Lan_{j_\U} j_\U^{\ast}\mathsf{F}_M(\RRR_M) \ar@{=>}@<0.75ex>[r]^-{\Lan_{j_\U}j_\U^{\ast}(r^M_1)}\ar@{=>}@<-0.75ex>[r]_-{\Lan_{j_\U}j_\U^{\ast}(r^M_2)}~&~ \Lan_{j_\U}j_\U^{\ast}\mathsf{F}_M(\LLL_M)
}\Big)\Big)\quad.\label{eqn:jpulljastgenrel}
\end{flalign}
In this model, the counit component $(\epsilon_{\U})_{\AAA_M} : 
j_{\U\,!}\,j_{\U}^\ast(\AAA_M)\Rightarrow \AAA_M \cong p_M(\AAA_M)$ in $\HK_{\ovr{\CC}}(M)$
is given by applying the reflector $p_M$ to the map between colimits
which is induced by the map of diagrams
\begin{subequations}\label{eqn:jpulljastgenreldiagram}
\begin{flalign}
\begin{gathered}
\xymatrix@R=3em@C=5em{
\ar@{=>}[d]_-{(\epsilon^{\Lan}_\U)_{\mathsf{F}_M(\RRR_M)}} 
\Lan_{j_\U} j_\U^{\ast}\mathsf{F}_M(\RRR_M) \ar@{=>}@<0.75ex>[r]^-{\Lan_{j_\U}j_\U^{\ast}(r^M_1)}\ar@{=>}@<-0.75ex>[r]_-{\Lan_{j_\U}j_\U^{\ast}(r^M_2)}~&~ \Lan_{j_\U}j_\U^{\ast}\mathsf{F}_M(\LLL_M)
\ar@{=>}[d]^-{(\epsilon^{\Lan}_\U)_{\mathsf{F}_M(\LLL_M)}} \\
\mathsf{F}_M(\RRR_M) \ar@{=>}@<0.75ex>[r]^-{r^M_1}\ar@{=>}@<-0.75ex>[r]_-{r^M_2}~&~ \mathsf{F}_M(\LLL_M)
}
\end{gathered}\qquad,
\end{flalign}
where $\epsilon^{\Lan}_\U$ denotes the counit of the adjunction $\Lan_{j_\U}\dashv j_{\U}^\ast$.
Using the commutativity properties in the diagram of adjunctions \eqref{eqn:bigadjunctiondiagram},
we can rewrite this map of diagrams in the following more convenient form
\begin{flalign}
\begin{gathered}
\xymatrix@R=3em@C=3.5em{
\ar@{=>}[d]_-{\mathsf{F}_M(\epsilon^{\lan}_\U)_{\RRR_M}} 
\mathsf{F}_M\lan_{j_\U} j_\U^{\ast}(\RRR_M) \ar@{=>}@<0.75ex>[r]^-{s^M_1}\ar@{=>}@<-0.75ex>[r]_-{s^M_2}~&~ \mathsf{F}_M\lan_{j_\U}j_\U^{\ast}(\LLL_M)
\ar@{=>}[d]^-{\mathsf{F}_M(\epsilon^{\lan}_\U)_{\LLL_M}} \\
\mathsf{F}_M(\RRR_M) \ar@{=>}@<0.75ex>[r]^-{r^M_1}\ar@{=>}@<-0.75ex>[r]_-{r^M_2}~&~ \mathsf{F}_M(\LLL_M)
}
\end{gathered}\qquad,
\end{flalign}
\end{subequations}
where $\epsilon^{\lan}_\U$ denotes the counit of the adjunction $\lan_{j_\U}\dashv j_{\U}^\ast$
and the relations $s^M_1,s^M_2$ are defined implicitly through this identification. 
(In what follows we do not need explicit expressions for $s^M_1,s^M_2$.) This allows us
to formulate a useful criterion for $\AAA_M\in\HK_{\ovr{\CC}}(M)$ to 
satisfy the descent conditions from Definitions \ref{def:improvedHK} and \ref{def:WimprovedHK}.
\begin{propo}\label{prop:descentcriterion}
Suppose that the counit component 
$(\epsilon^{\lan}_\U)_{\LLL_M} : \lan_{j_\U}\,j_\U^\ast(\LLL_M)\Rightarrow \LLL_M$
of the generators $\LLL_M\in \Fun(\CC(M),\TT)$ is an isomorphism in $\Fun(\CC(M),\TT)$. 
Then the counit component
$(\epsilon_{\U})_{\AAA_M} : j_{\U\,!}\,j_{\U}^\ast(\AAA_M)\Rightarrow \AAA_M$ 
of the AQFT $\AAA_M\in\HK_{\ovr{\CC}}(M)$ from \eqref{eqn:genrelAQFT}
is an isomorphism in $\HK_{\ovr{\CC}}(M)$ if and only if the diagram
\begin{flalign}\label{eqn:descentcriterionHK}
\begin{gathered}
\xymatrix@R=3em@C=3.5em{ 
p_M\mathsf{F}_M\lan_{j_\U} j_\U^{\ast}(\RRR_M) 
\ar@{=>}@<0.75ex>[rr]^-{p_M(r_1^M\circ \mathsf{F}_M(\epsilon^{\lan}_\U)_{\RRR_M})}
\ar@{=>}@<-0.75ex>[rr]_-{p_M(r_2^M\circ \mathsf{F}_M(\epsilon^{\lan}_\U)_{\RRR_M})}~&~ ~&~
p_M\mathsf{F}_M(\LLL_M)\ar@{=>}[r]~&~ \AAA_M
}
\end{gathered}
\end{flalign}
is a coequalizer in $\HK_{\ovr{\CC}}(M)$. A sufficient condition 
for this is that the diagram
\begin{flalign}\label{eqn:descentcriterion}
\begin{gathered}
\xymatrix@R=3em@C=3em{ 
\mathsf{F}_M\lan_{j_\U} j_\U^{\ast}(\RRR_M) 
\ar@{=>}@<0.75ex>[rr]^-{r_1^M\circ \mathsf{F}_M(\epsilon^{\lan}_\U)_{\RRR_M}}
\ar@{=>}@<-0.75ex>[rr]_-{r_2^M\circ \mathsf{F}_M(\epsilon^{\lan}_\U)_{\RRR_M}}~&~ ~&~
p_M\mathsf{F}_M(\LLL_M)\ar@{=>}[r]~&~ \AAA_M
}
\end{gathered}
\end{flalign}
is a coequalizer in the functor category $\Fun\big(\CC(M),\Alg_{\mathsf{uAs}}(\TT)\big)$,
where the unit of the adjunction $p_M \dashv {\subseteq}$ is left implicit.
\end{propo}
\begin{proof}
Using the hypothesis that $(\epsilon^{\lan}_\U)_{\LLL_M}$ is an isomorphism
and the diagrams \eqref{eqn:jpulljastgenreldiagram}, we obtain an isomorphic
presentation for \eqref{eqn:jpulljastgenrel} which is given by
\begin{flalign}
\nn j_{\U\,!}j_\U^\ast(\AAA_M)\,&\cong\, p_M\Big(\colim_{\Alg}^{}\Big(\xymatrix@C=6em{
\mathsf{F}_M\lan_{j_\U} j_\U^{\ast}(\RRR_M)
\ar@{=>}@<0.75ex>[r]^-{r_1^M\circ \mathsf{F}_M(\epsilon^{\lan}_\U)_{\RRR_M}}
\ar@{=>}@<-0.75ex>[r]_-{r_2^M\circ \mathsf{F}_M(\epsilon^{\lan}_\U)_{\RRR_M}}
~&~ \mathsf{F}_M(\LLL_M)
}\Big)\Big)\\
\,&\cong\, \colim_{\HK}^{}\Big(\xymatrix@C=7.5em{
p_M\mathsf{F}_M\lan_{j_\U} j_\U^{\ast}(\RRR_M)
\ar@{=>}@<0.75ex>[r]^-{p_M(r_1^M\circ \mathsf{F}_M(\epsilon^{\lan}_\U)_{\RRR_M})}
\ar@{=>}@<-0.75ex>[r]_-{p_M(r_2^M\circ \mathsf{F}_M(\epsilon^{\lan}_\U)_{\RRR_M})}
~&~ p_M\mathsf{F}_M(\LLL_M)
}\Big)\quad,
\end{flalign}
where $\colim_{\HK}$ denotes the colimit in $\HK_{\ovr{\CC}}(M)$.
Then $(\epsilon_{\U})_{\AAA_M} $ is an isomorphism if and only if 
the diagram \eqref{eqn:descentcriterionHK} is a coequalizer in $\HK_{\ovr{\CC}}(M)$.
To prove the second statement, we apply the left adjoint functor $p_M$
to the coequalizer \eqref{eqn:descentcriterion}, which yields a coequalizer in
$\HK_{\ovr{\CC}}(M)$ that is isomorphic to \eqref{eqn:descentcriterionHK} 
since $p_M\dashv {\subseteq}$ exhibits a reflective full subcategory, 
hence $p_Mp_M\mathsf{F}_M(\LLL_M)\cong p_M\mathsf{F}_M(\LLL_M)$ and 
$p_M(\AAA_M)\cong \AAA_M$.
\end{proof}

\begin{ex}\label{ex:KGdescent}
We will show that the free Klein-Gordon quantum field on $M\in\Loc$,
formulated in terms of a relatively compact Haag-Kastler-style AQFT 
$\AAA_M^{\KG}\in\HK^{\mathrm{rc}}(M)$ neglecting the time-slice axiom, 
satisfies the descent conditions from Definition \ref{def:improvedHK}
for every causally convex open cover $\U=\{U_i\subseteq M\}$.
For this we shall test the criterion from Proposition \ref{prop:descentcriterion}.
In this example we choose as target $\TT=\Vec_\bbC$ the closed symmetric monoidal category of complex vector spaces.

\sk
Recalling the standard construction of the free Klein-Gordon quantum field, 
see e.g.\ \cite{BDHreview,BDreview} for reviews, the object $\AAA_M^{\KG}\in\HK^{\mathrm{rc}}(M)$ 
may be presented by the generators
\begin{flalign}\label{eqn:KGgenerators}
 \LLL^{\KG}_M\,:=\,\tfrac{C^\infty_\cc(-)}{P_MC^\infty_\cc(-)} \,:\, \RC(M)~&\longrightarrow~\Vec_\bbC\quad,
\end{flalign}
where $P_M:= -\Box_M+m^2$ denotes the Klein-Gordon operator. More explicitly,
this functor assigns to an object $U\in \RC(M)$ the quotient vector space 
$\LLL^{\KG}_M(U)=\tfrac{C^\infty_\cc(U)}{P_MC^\infty_\cc(U)}$ and to a morphism $U\subseteq U^\prime$
in $\RC(M)$ the pushforward (extension by zero) map, 
which we shall simply write as $\LLL^{\KG}_M(U)\to \LLL^{\KG}_M(U^\prime)\,,~[\varphi]\mapsto[\varphi]$.
(Recall that these pushforward maps are injective, for all morphisms $U\subseteq U^\prime$ in $\RC(M)$.)
The relations are given by the canonical commutation relations (CCR), i.e.\
\begin{subequations} \label{eqn:CCR}
\begin{flalign}
\RRR^{\KG}_M \,:=\,  \LLL^{\KG}_M\otimes  \LLL^{\KG}_M \,:\, \RC(M)~\longrightarrow~ \Vec_\bbC
\end{flalign} 
is the object-wise tensor product of generators,
\begin{flalign}
\tilde{r}^M_1 \,:=\ [-,-]\,:\, \RRR^{\KG}_M ~\Longrightarrow~\mathsf{U}_M\mathsf{F}_M(\LLL^\KG_M)
\end{flalign}
is the commutator in the free algebra, and
\begin{flalign}
\tilde{r}^M_2 \,:=\ \ii\hbar\,\tau_M(-,-)\,\oone\,:\, \RRR^{\KG}_M ~\Longrightarrow~\mathsf{U}_M\mathsf{F}_M(\LLL^\KG_M)
\end{flalign}
\end{subequations}
is given by weighting the unit $\oone$ of the free algebra with $\ii\hbar $ times the usual Poisson structure
$\tau_M(-,-):= \int_M (-)\,G_M(-)\,\vol_M$ which is constructed out of the retarded-minus-advanced
Green's operator $G_M:= G_M^{+}- G_M^-$ for the Klein-Gordon operator $P_M$. 
It is well-known \cite{BDHreview,BDreview} that the AQFT $\AAA_M^{\KG}$ constructed above satisfies the time-slice axiom, 
i.e.\ it sends Cauchy morphisms $U \subseteq U^\prime$ in $\RC(M)$ to isomorphisms.
This fact will be helpful for some of the computations performed below.
\sk

Computing explicitly the left Kan extension $\lan_{j_\U}\,j_{\U}^\ast(\LLL^\KG_M)$
for any causally convex open cover $\U= \{U_i\subseteq M\}$ 
by the usual object-wise colimit formula \cite[Chapter 6.2]{Riehl},
one finds that $(\epsilon^{\lan}_\U)_{\LLL^\KG_M}$ is an isomorphism
if and only if
\begin{flalign}\label{eqn:KGcoequalizer}
\xymatrix{
\bigoplus\limits_{i,j} \LLL^\KG_M(U_{ij}\cap U) \ar@<0.5ex>[r]\ar@<-0.5ex>[r]~&~ 
\bigoplus\limits_{i} \LLL^\KG_M(U_i\cap U) \ar[r]~&~\LLL^\KG_M(U)
}
\end{flalign}
is a coequalizer in $\Vec_{\bbC}$, for all $U\in \RC(M)$.
The unlabeled maps are given by applying 
the functor $\LLL^\KG_M$ to the inclusions $U_{ij}\cap U\subseteq U_i\cap U$,
$U_{ij}\cap U\subseteq U_j\cap U$ and $U_i\cap U\subseteq U$. Picking a partition of unity
$\{\chi_i\}$ subordinate to the open cover $\{U_i\cap U\subseteq U\}$ of $U$,
one shows that the map $\bigoplus_{i} \LLL^\KG_M(U_i\cap U)\to \LLL^\KG_M(U)$
in \eqref{eqn:KGcoequalizer} is surjective. Indeed, given any element $[\varphi]\in 
\LLL^\KG_M(U)$, then $\bigoplus_i[\chi_i\varphi]\in \bigoplus_{i} \LLL^\KG_M(U_i\cap U)$
maps to $\sum_i [\chi_i \varphi] = \big[\sum_i\chi_i \varphi\big] = [\varphi]$. It remains
to show that every element $\bigoplus_i[\varphi_i]\in \bigoplus_{i} \LLL^\KG_M(U_i\cap U)$
which is mapped to zero, i.e.\ $\sum_i[\varphi_i] = \big[\sum_i\varphi_i\big]=0$
or equivalently $\sum_i\varphi_i = P_M(\psi)$ for some $\psi\in C^\infty_\cc(U)$,
is of the form $\bigoplus_i[\varphi_i]=\bigoplus_{i}\sum_j\big([\varphi_{ji}] - [\varphi_{ij}]\big)$
for some $\bigoplus_{i,j}[\varphi_{ij}]\in\bigoplus_{i,j} \LLL^\KG_M(U_{ij}\cap U)$.
This can be achieved by setting $\varphi_{ij}:= \chi_i\varphi_j - \chi_i\,P_M(\chi_j\psi)$,
for all $i,j$, since
\begin{flalign}
\nn \sum_j\big([\varphi_{ji}] - [\varphi_{ij}]\big)\,&=\,\sum_j\big[ \chi_j\varphi_i - \chi_j\,P_M(\chi_i\psi)
- \chi_i\varphi_j + \chi_i\,P_M(\chi_j\psi)\big]\\
\,&=\,\big[\varphi_i-P_M(\chi_i\psi)\big] - \Big[\chi_i\sum_j\varphi_j - \chi_i P_{M}(\psi)\Big] \,=\,[\varphi_i]\quad,
\end{flalign}
for all $i$. Hence, \eqref{eqn:KGcoequalizer} is a coequalizer and the hypothesis
of Proposition \ref{prop:descentcriterion} holds true.
\sk

It remains to investigate the diagram 
\eqref{eqn:descentcriterion} from Proposition \ref{prop:descentcriterion}.
Computing the left Kan extension
$\lan_{j_\U}\,j_{\U}^\ast(\RRR^\KG_M)$
again via the object-wise colimit formula,
one finds that \eqref{eqn:descentcriterion} is a coequalizer if and only if the object-wise diagrams
\begin{flalign}\label{eqn:objectwiseKGcoeq}
\begin{gathered}
\xymatrix@R=3em@C=2em{ 
\mathsf{F}_{\mathsf{uAs}}\Big(\bigoplus\limits_{i}\big(\LLL_{M}^{\KG}(U_i\cap U)
\otimes \LLL_{M}^{\KG}(U_i\cap U)\big)\Big)
\ar@<0.75ex>[r]^-{(r_1^M)_U}
\ar@<-0.75ex>[r]_-{(r_2^M)_U}~&~ 
\big(p_M\mathsf{F}_M(\LLL^\KG_M)\big)(U)\ar[r]~&~ \AAA_M^{\KG}(U)
}
\end{gathered}
\end{flalign}
are coequalizers in $\Alg_{\mathsf{uAs}}(\Vec_\bbC)$, for all $U\in\RC(M)$, where the relations are obtained by
restricting the CCR relations \eqref{eqn:CCR} along the map
\begin{flalign}
\bigoplus\limits_i \Big(\LLL^\KG_M(U_i\cap U)\otimes \LLL^\KG_M(U_i\cap U) \Big)~\longrightarrow~
\LLL^\KG_M(U)\otimes \LLL^\KG_M(U) \,=\,\RRR_M^\KG(U)
\end{flalign}
determined by the applying the functor $\LLL^\KG_M$ to the inclusions $U_i\cap U\subseteq U$.
The algebra $\big(p_M\mathsf{F}_M(\LLL^\KG_M)\big)(U)\in \Alg_{\mathsf{uAs}}(\Vec_\bbC)$ 
in this expression is generated freely by all $[\varphi]\in \LLL^\KG_M(U)$, modulo the 
minimal $\perp$-commutativity relations demanding vanishing of the commutator
\begin{flalign}\label{eqn:pMrelations}
\big[[\varphi_1^\perp], [\varphi_2^\perp]\big]\,=\,0\quad,
\end{flalign}
for all $[\varphi_1^\perp],[\varphi_2^\perp]\in \LLL^\KG_M(U)$ such that 
$[\varphi_a^\perp]$ comes from extension by zero along morphisms 
$V_a \subseteq U$ in $\RC(M)$ with $V_1 \perp V_2$ causally disjoint. 
It thus remains to show that \eqref{eqn:pMrelations} and the restricted
CCR relations
\begin{flalign}\label{eqn:restrictedCCRrelations}
\big[[\varphi_1^i],[\varphi_2^i]\big] \,=\,\ii\hbar\,\tau_M\big([\varphi^i_1],[\varphi^i_2]\big)\,\oone\quad,
\end{flalign}
for all $[\varphi^i_1],[\varphi^i_2]\in \LLL^\KG_M(U_i\cap U)$ and all $i$,
imply together the general CCR relations in $\AAA_M^{\KG}(U)$ for all 
$[\varphi_1],[\varphi_2]\in \LLL^\KG_M(U)$. This can be done by an argument
which is similar to the one in \cite[Lemma 3.2 and Proposition 3.2]{DappiaggiLang}.\footnote{We 
refer to the published version, which differs from the preprint available on arXiv.}
For this we pick any spacelike Cauchy surface $\Sigma$ of $U$ and a family of causally convex open subsets
$\{V_\alpha\subseteq U\}$ in $U$ that covers $\Sigma$, i.e.\ $\Sigma\subseteq \bigcup_\alpha V_\alpha \subseteq U$,
but does not necessarily cover $U$,
and satisfies the following property:
If $V_\alpha$ and $V_\beta$ are not causally disjoint then there 
exists an $i$ such that $V_\alpha \cup V_\beta \subseteq U_i$.
(Such a cover $\{V_\alpha\}$ may be found as follows:
For each point $p$ in the Riemannian manifold $\Sigma$,
pick an index $i_p$ and a radius $r_p$
such that the open ball in $\Sigma$ centered at $p$ with radius $3 r_p$ is contained in $U_{i_p}$. (In particular, $p \in U_{i_p}$.) 
The collection of smaller open balls of radius $r_p$ around each $p$ covers $\Sigma$ and has the property that,
whenever two such balls around $p, q \in \Sigma$ intersect,
their union is contained in $U_{i_p}$, where without loss of generality we assume $r_q \leq r_p$.
The family $\{V_\alpha\}$ can be taken as Cauchy developments
of this cover of $\Sigma$ by small balls,
after possibly shrinking the balls further to ensure $V_\alpha \subseteq U_i \cap U$.)
Let us also choose any partition of unity $\{\chi_\alpha\}$ subordinate to the cover $\{V_\alpha\}$.
Using the time-slice axiom and the Cauchy morphism $V \subseteq U$ given by the inclusion in $U$ 
of a causally convex open neighborhood $V \subseteq \bigcup_\alpha V_\alpha$
of $\Sigma$,  we can and will choose representatives 
for $[\varphi_1],[\varphi_2]\in \LLL^\KG_M(U)$ with support in $V\subseteq U$, 
and hence also in $\bigcup_\alpha V_\alpha\subseteq U$. This allows us to compute
\begin{flalign}
\nn \big[[\varphi_1],[\varphi_2]\big]
\,&=\, \sum_{\alpha,\beta} \big[[\chi_\alpha\varphi_1],[\chi_\beta\varphi_2]\big]
\,=\,\sum_{\substack{\alpha,\beta\text{~s.t.}\\ \text{not } V_\alpha \perp V_\beta}} \big[[\chi_\alpha\varphi_1],[\chi_\beta\varphi_2]\big]+ \sum_{\substack{\alpha,\beta\text{~s.t.}\\ V_\alpha \perp V_\beta}} \big[[\chi_\alpha\varphi_1],[\chi_\beta\varphi_2]\big]\\[4pt]
\nn \,&=\,
\sum_{\substack{\alpha,\beta\text{~s.t.}\\ \text{not } V_\alpha \perp V_\beta}}\,\ii\hbar\,\tau_M\big([\chi_\alpha\varphi_1],[\chi_\beta\varphi_2]\big) \,\oone+ 0\,=\,
\sum_{\alpha,\beta}\,\ii\hbar\,\tau_M\big([\chi_\alpha\varphi_1],[\chi_\beta\varphi_2]\big)\,\oone\\[4pt]
\,&=\,\ii \hbar\, \tau_M\big([\varphi_1],[\varphi_2]\big)\,\oone\quad,
\end{flalign}
where in the third step we used the relations \eqref{eqn:restrictedCCRrelations} and \eqref{eqn:pMrelations},
together with the specified property of the cover $\{V_\alpha\}$, and in the fourth step
we used that the Poisson structure vanishes between causally disjoint generators.
\end{ex}

\begin{ex}\label{ex:KGdescentW}
We will now show that the free Klein-Gordon quantum field on $M\in\Loc$,
formulated in terms of a time-sliced relatively compact Haag-Kastler-style AQFT 
$\AAA_M^{\KG,W}\in\HK^{\mathrm{rc},W}(M)$,
satisfies the descent conditions from Definition \ref{def:WimprovedHK}
for every $D$-stable causally convex open cover $\U=\{U_i\subseteq M\}$.
For this we shall test the criterion from Proposition \ref{prop:descentcriterion}.
In this example we choose again as target $\TT=\Vec_\bbC$ the closed symmetric monoidal category of complex vector spaces.
\sk

To streamline our calculations, it will be convenient to 
describe $\AAA_M^{\KG,W}\in\HK^{\mathrm{rc},W}(M)$ as the 
pullback along the full orthogonal subcategory inclusion 
$\ovr{\RC(M)}[W_{\mathrm{rc},M}^{-1}] \subseteq \ovr{\COpen(M)}[W_M^{-1}]$ from Example \ref{ex:localizations} 
of a time-sliced Haag-Kastler-style AQFT $\AAA_M^{\KG,W}\in\HK^{W}(M)$ 
(denoted with abuse of notation by the same symbol) defined on all causally convex open subsets, including those that are not relatively compact.
For the generators of $\AAA_M^{\KG,W}\in\HK^{W}(M)$ we take
\begin{subequations}\label{eqn:KGgeneratorsW}
\begin{flalign}
\LLL^{\KG,W}_M\,:=\,\tfrac{C^\infty_\cc(-)}{P_MC^\infty_\cc(-)} \,:\, 
\COpen(M)[W_{M}^{-1}]~&\longrightarrow~\Vec_\bbC
\end{flalign}
with the functorial structure on morphisms $U\to U^\prime$
in the localized category $\COpen(M)[W_M^{-1}]$ 
from Example \ref{ex:localizations} given by
\begin{flalign}
\LLL^{\KG,W}_M(U)~\longrightarrow~\LLL^{\KG,W}_M(U^\prime)~~,\quad
[\varphi]~\longmapsto~
\pm \big[P_M\big(\chi_\pm\,G_M(\varphi)\big)\big]
\quad,
\end{flalign}
\end{subequations}
where $\{\chi_+,\chi_-\}$ is any choice of partition of unity
subordinate to the open cover $\{I^+_M(\Sigma^-),I^-_{M}(\Sigma^+)\}$
of $J_M(D_M(U^\prime))= J^+_M(D_M(U^\prime))\cup J^-_M(D_M(U^\prime))\subseteq M$
which is associated with any choice of two spacelike Cauchy surfaces $\Sigma^+,\Sigma^-$
of $U^\prime$ such that $\Sigma^+\subseteq I_M^+(\Sigma^-)$ lies in the chronological future of $\Sigma^-$.
This concrete description of the functorial structure follows by 
applying\footnote{Our opposite overall sign compared to \cite{BMS} is a consequence
of the fact that the latter considers shifted cochain complexes.}
\cite[Theorem 3.13]{BMS}
to construct inverses for morphisms which are 
obtained by applying the functor \eqref{eqn:KGgenerators} to Cauchy morphisms.
Note that for morphisms $U\to U^\prime$ which correspond to subset inclusions $U\subseteq U^\prime$,
this functorial structure simplifies to the extension by zero map 
$\LLL^{\KG,W}_M(U)\to \LLL^{\KG,W}_M(U^\prime)\,,~[\varphi]\mapsto 
\pm \big[P_M\big(\chi_\pm\,G_M(\varphi)\big)\big] = [\varphi]$
because $\pm P_M\big(\chi_\pm\, G_M (\varphi)\big) = \varphi + P_M (\psi)$
with $\psi := - \chi_- \, G_M^+ (\varphi) - \chi_+ \, G_M^-(\varphi) \in C^\infty_\cc(U^\prime)$
for all $\varphi\in C^\infty_\cc(U)$ with support in $U\subseteq U^\prime$.
This observation will be very useful in our calculations below.
\sk

The relations for $\AAA_M^{\KG,W}\in\HK^{W}(M)$ are given again by the CCR relations, i.e.\
\begin{subequations}\label{eqn:CCRW}
\begin{flalign}
\RRR^{\KG,W}_M \,:=\,  \LLL^{\KG,W}_M\otimes  \LLL^{\KG,W}_M \,:\, \COpen(M)[W_{M}^{-1}]
~\longrightarrow~ \Vec_\bbC
\end{flalign} 
is the tensor product of generators,
\begin{flalign}
\tilde{r}^{M,W}_1 \,:=\ [-,-]\,:\, \RRR^{\KG,W}_M ~\Longrightarrow~\mathsf{U}_M\mathsf{F}_M(\LLL^{\KG,W}_M)
\end{flalign}
is the commutator in the free algebra, and
\begin{flalign}
\tilde{r}^{M,W}_2 \,:=\ \ii\hbar\,\tau_M(-,-)\,\oone \,:\, \RRR^{{\KG,W}}_M ~\Longrightarrow~
\mathsf{U}_M\mathsf{F}_M(\LLL^{\KG,W}_M)
\end{flalign}
\end{subequations}
is given by weighting the unit $\oone$ of the free algebra with $\ii\hbar $ times the usual Poisson structure
$\tau_M(-,-):= \int_M (-)\,G_M(-)\,\vol_M$.
\sk

Let us consider now the pullback of the above generators 
and relations along $\ovr{\RC(M)}[W_{\mathrm{rc},M}^{-1}] \subseteq \ovr{\COpen(M)}[W_M^{-1}]$.
Computing explicitly the left Kan extension $\lan_{j_\U}\,j_{\U}^\ast(\LLL^{\KG,W}_M)$
for any $D$-stable causally convex open cover $\U= \{U_i\subseteq M\}$ 
by the usual object-wise colimit formula \cite[Chapter 6.2]{Riehl}, one
finds that $(\epsilon^{\lan}_\U)_{\LLL^{\KG,W}_M}$ is an isomorphism if and only if
\begin{flalign}\label{eqn:KGcoequalizerW}
\xymatrix{
\bigoplus\limits_{i,j} \LLL^{\KG,W}_M\big(U_{ij}\cap D_M(U)\big) \ar@<0.5ex>[r]\ar@<-0.5ex>[r]~&~ 
\bigoplus\limits_{i} \LLL^{\KG,W}_M\big(U_i\cap D_M(U)\big) \ar[r]~&~\LLL^{\KG,W}_M(U)
}
\end{flalign}
is a coequalizer in $\Vec_{\bbC}$, for all $U\in \RC(M)[W_{\mathrm{rc},M}^{-1}]$.
Note that $U_{i}\cap D_M(U)\subseteq M$ and $U_{ij}\cap D_M(U)\subseteq M$ are not necessarily relatively compact,
which is why it was convenient to introduce the generators \eqref{eqn:KGgeneratorsW}
on all of $\COpen(M)[W_M^{-1}]$ instead of only on the full subcategory 
$\RC(M)[W_{\mathrm{rc},M}^{-1}]\subseteq \COpen(M)[W_M^{-1}]$. We can postcompose
\eqref{eqn:KGcoequalizerW} with the isomorphism $\LLL^{\KG,W}_M(U)\stackrel{\cong}{\longrightarrow} 
\LLL^{\KG,W}_M\big(D_M(U)\big)$, which yields an equivalent diagram where all maps
are determined by subset inclusions, hence the functorial structure \eqref{eqn:KGgeneratorsW}
simplifies to the usual extension by zero maps. Then the same calculations
as in Example \ref{ex:KGdescent} show that the diagram \eqref{eqn:KGcoequalizerW} is indeed a coequalizer for all
$U\in \RC(M)[W_{\mathrm{rc},M}^{-1}]$,
so the hypothesis of Proposition \ref{prop:descentcriterion} is satisfied.

\sk
It remains to investigate the diagram \eqref{eqn:descentcriterion} from Proposition \ref{prop:descentcriterion}.
Again via the colimit formula for the left Kan extension $\lan_{j_\U}\,j_{\U}^\ast(\RRR^{\KG,W}_M)$,
one finds that \eqref{eqn:descentcriterion} is a coequalizer if and only if the object-wise diagrams
\begin{flalign}\label{eqn:objectwiseKGcoeqW}
\begin{gathered}
\resizebox{.9\hsize}{!}{$
\xymatrix@R=3em@C=2em{ 
\mathsf{F}_{\mathsf{uAs}}\Big(\bigoplus\limits_{i}\big(\LLL_{M}^{{\KG,W}}(U_i\cap U)
\otimes \LLL_{M}^{\KG,W}(U_i\cap U)\big)\Big)
\ar@<0.75ex>[r]^-{(r_1^M)_U}
\ar@<-0.75ex>[r]_-{(r_2^M)_U}~&~ 
\big(p_M\mathsf{F}_M(\LLL^{\KG,W}_M)\big)(U)\ar[r]~&~ \AAA_M^{\KG,W}(U)
}
$}
\end{gathered}
\end{flalign}
are coequalizers in $\Alg_{\mathsf{uAs}}(\Vec_\bbC)$, for all $U\in\RC(M)[W_{\mathrm{rc},M}^{-1}]$.
To understand this claim, it is crucial to observe that the restriction of the CCR relations
\eqref{eqn:CCRW} along the map
\begin{flalign}
\resizebox{.9\hsize}{!}{$
\bigoplus\limits_i \Big(\LLL^{\KG,W}_M\big(U_i\cap D_M(U)\big)\otimes 
\LLL^{\KG,W}_M\big(U_i\cap D_M(U)\big) \Big)~\longrightarrow~
\LLL^{\KG,W}_M(U)\otimes \LLL^{\KG,W}_M(U) \,=\,\RRR_M^{\KG,W}(U)
$}
\end{flalign}
can be restricted further along the isomorphisms $\LLL^{\KG,W}_M(U_i\cap U)\stackrel{\cong}{\longrightarrow}
\LLL^{\KG,W}_M\big(U_i\cap D_M(U)\big)$, leading to an equivalent description of the relations.
Then the same calculations as in Example \ref{ex:KGdescent} show that the diagram \eqref{eqn:objectwiseKGcoeqW}
is indeed a coequalizer, for all $U\in \RC(M)[W_{\mathrm{rc},M}^{-1}]$.
\end{ex}

To state and prove the main result of this subsection, let us recall
that the Klein-Gordon quantum field can be constructed also
as a locally covariant AQFT $\AAA^\KG\in \AQFT(\ovr{\Loc})$, see e.g.\ \cite{BDHreview,BDreview}.
It is well-known and easy to verify that this locally covariant AQFT
satisfies both the time-slice axiom and the additivity property \cite{BPS}, i.e.\ we even 
obtain an object $\AAA^\KG\in \AQFT(\ovr{\Loc})^{\mathrm{add},W}$.
The restriction of $\AAA^\KG\in \AQFT(\ovr{\Loc})$ along the orthogonal functor
$k_M : \ovr{\RC(M)}\to\ovr{\Loc}$ from item (3) of Example \ref{ex:OCat} then
coincides with the relatively compact Haag-Kastler-style AQFT $\AAA^\KG_M\in\HK^{\mathrm{rc}}(M)$
from Example \ref{ex:KGdescent}, for all $M\in\Loc$. Furthermore, passing also to
the orthogonal localization $L_{\mathrm{rc},M} : \ovr{\RC(M)}\to \ovr{\RC(M)}[W_{\mathrm{rc},M}^{-1}]$ 
from Example \ref{ex:localizations},
one obtains the time-sliced relatively compact Haag-Kastler-style AQFT 
$\AAA^{\KG,W}_M\in\HK^{\mathrm{rc},W}(M)$ from Example \ref{ex:KGdescentW}, for all $M\in\Loc$.
\begin{theo}\label{theo:KGpoints}
Let $\AAA^\KG\in \AQFT(\ovr{\Loc})^{\mathrm{add},W}\subseteq \AQFT(\ovr{\Loc})$ 
be the free Klein-Gordon quantum field formulated as a locally covariant AQFT.
\begin{itemize}
\item[(1)] The point of the relatively compact Haag-Kastler $2$-functor $\HK^{\mathrm{rc}}$
which is obtained by forgetting the time-slice axiom $\AAA^\KG\in \AQFT(\ovr{\Loc})^{\mathrm{add}}$ 
and applying the fully faithful 
functor from Corollary \ref{cor:additiveLCQFT} defines via Proposition \ref{prop:calHKpoints}
a point of the relatively compact Haag-Kastler stack $\calHK^{\mathrm{rc}}$ from Theorem \ref{theo:RCHKstack}.

\item[(2)] The point of the time-sliced relatively compact Haag-Kastler $2$-functor $\HK^{\mathrm{rc},W}$
which is obtained by applying the fully faithful functor from Corollary \ref{cor:WadditiveLCQFT} to 
$\AAA^\KG\in \AQFT(\ovr{\Loc})^{\mathrm{add},W}$ defines via Proposition \ref{prop:WcalHKpoints}
a point of the time-sliced relatively compact Haag-Kastler stack $\calHK^{\mathrm{rc},W}$ 
from Theorem \ref{theo:WRCHKstack}.
\end{itemize}
\end{theo}
\begin{proof}
These claims hold true because the
descent conditions required for Propositions \ref{prop:calHKpoints}
and \ref{prop:WcalHKpoints} have already been verified in Examples \ref{ex:KGdescent}
and \ref{ex:KGdescentW} above.
\end{proof}

%%%%%%%%%%%%%%%%%%%%%%%%%%%%%%%%%%%%%%%%%%%%%%%%
%%%%%%%%%%%%%%%%%%%%%%%%%%%%%%%%%%%%%%%%%%%%%%%%

\section*{Acknowledgments}
We would like to thank the anonymous referee
for their useful comments and questions which
helped us to improve this paper.
We would like to thank Angelos Anastopoulos for agreeing to share 
a draft of \cite{AB}. 
The work of M.B.\ is supported in part by the MUR Excellence 
Department Project awarded to Dipartimento di Matematica, 
Universit{\`a} di Genova (CUP D33C23001110001) and it is fostered by 
the National Group of Mathematical Physics (GNFM-INdAM (IT)). 
A.G-S.\ was supported by Royal Society Enhancement Grants (RF\textbackslash ERE\textbackslash 210053
and RF\textbackslash ERE\textbackslash 231077).
A.S.\ gratefully acknowledges the support of 
the Royal Society (UK) through a Royal Society University 
Research Fellowship (URF\textbackslash R\textbackslash 211015)
and Enhancement Grants (RF\textbackslash ERE\textbackslash 210053 and 
RF\textbackslash ERE\textbackslash 231077).

%\section*{Data availability statement}
%All data generated or analyzed during this study are contained in this document. 
%
%
%\section*{Conflict of interest statement}
%The authors have no conflict of interest to declare that are relevant to the content of this article. 

%%%%%%%%%%%%%%%%%%%%%%%%%%%%%%%%%%%%%%%%%%%%%%%%
%%%%%%%%%%%%%%%%%%%%%%%%%%%%%%%%%%%%%%%%%%%%%%%%

\appendix

\section{\label{app:operadicLKE}Operadic left Kan extensions}
In this appendix we present a concrete model
for the operadic left Kan extensions from Proposition \ref{propo:operadicLKE}.
For this we have to recall the explicit description of the AQFT operads
$\O_{\ovr{\CC}}^{}$ from \cite{BSWoperad}, see also \cite{BSreview}. 
\begin{defi}\label{def:AQFToperad}
Let $\ovr{\CC} = (\CC,\perp_{\CC}^{})$ be an orthogonal category.
The associated {\em AQFT operad} $\O_{\ovr{\CC}}$ is the colored operad
which is defined by the following data:
\begin{itemize}
\item[(i)] the objects are the objects of $\CC$;

\item[(ii)] the set of operations from $\und{M} := (M_1,\dots, M_n)$
to $N$ is the quotient set
\begin{flalign}
\O_{\ovr{\CC}}\big(\substack{N \\ \und{M}}\big)\,:=\,\bigg(\Sigma_n \times \prod_{i=1}^n \Hom_{\CC}(M_i,N)\bigg)\Big/\!\sim_{\perp_{\CC}^{}}^{}\quad,
\end{flalign}
where $\Hom_{\CC}(M_i,N)$ denotes the set of $\CC$-morphisms from $M_i$ to $N$, 
$\Sigma_n$ denotes the permutation group on $n$ letters,
and the equivalence relation is defined as follows: $(\sigma,\und{f}) 
\sim_{\perp_{\CC}^{}}^{} (\sigma^\prime,\und{f}^\prime)$ if and only if 
$\und{f} := (f_1, \ldots , f_n) = (f_1^\prime , \ldots, f_n^\prime) =: \und{f}^\prime$ and the right permutation
$\sigma\sigma^{\prime -1} :
\und{f}\sigma^{-1} := \left( f_{\sigma^{-1}(1)}, \ldots , f_{\sigma^{-1}(n)} \right)
\to \und{f}\sigma^{\prime -1} := \left( f_{\sigma^{\prime -1}(1)}, \ldots , f_{\sigma^{\prime -1}(n)} \right)$
is generated by transpositions of adjacent orthogonal pairs;

\item[(iii)] the composition of $[\sigma,\und{f}] : \und{M}\to N$ 
with $[\sigma_i,\und{g}_i] : \und{K}_i\to M_i$, for $i=1,\dots,n$, is
\begin{subequations}
\begin{flalign}
[\sigma , \und{f}]\, [ \und{\sigma}, \und{\und{g}}]\,:=\,
\big[\sigma(\sigma_1,\dots,\sigma_n), \und{f}\,\und{\und{g}}\,\big]\,:\,\und{\und{K}}\,\longrightarrow\,N\quad,
\end{flalign}
where $\sigma(\sigma_1,\dots,\sigma_n)$ denotes
the composition in the unital associative operad and
\begin{flalign}
 \und{f}\,\und{\und{g}} \,:=\, \big(f_1\,g_{11},\dots, f_1\,g_{1k_1},\dots, f_n\,g_{n1},\dots,f_{n}\,g_{n k_n}\big)
\end{flalign}
\end{subequations}
is given by compositions in the category $\CC$;

\item[(iv)] the unit elements are $\oone := [e,\id_N^{}] : N\to N$, where $e\in\Sigma_1$ is the identity permutation;

\item[(v)] the permutation action of $\sigma^\prime\in\Sigma_n$ on $[\sigma,\und{f}] : \und{M}\to N$ is
\begin{flalign}
[\sigma,\und{f}]\cdot \sigma^\prime \,:=\, [\sigma\sigma^\prime,\und{f}\sigma^\prime] : \und{M}\sigma^{\prime}~\longrightarrow~
 N\quad,
\end{flalign}
where $\und{f}\sigma^\prime = (f_{\sigma^\prime(1)},\dots,f_{\sigma^\prime(n)})$ 
and $\und{M}\sigma^{\prime}= (M_{\sigma^\prime(1)},\dots, M_{\sigma^{\prime}(n)})$ denote 
the permuted tuples and $\sigma\sigma^\prime$ is given by the composition of permutations in $\Sigma_n$.
\end{itemize}
The key property of the AQFT operad $\O_{\ovr{\CC}}$ is that the category of
$\TT$-valued AQFTs over $\ovr{\CC}$ from Definition \ref{def:AQFT} is isomorphic to
the category of $\O_{\ovr{\CC}}$-algebras with values in $\TT$, 
i.e.\ $\AQFT(\ovr{\CC})\cong \Alg_{\O_{\ovr{\CC}}}(\TT)$.
\end{defi}

Given any orthogonal functor $F : \ovr{\CC}\to\ovr{\DD}$, one defines
an operad morphism $\O_F : \O_{\ovr{\CC}}\to \O_{\ovr{\DD}}$
by sending each object $M\in\CC$ to $F(M)\in \DD$
and each operation $[\sigma,\und{f}] : \und{M}\to N$ in $\O_{\ovr{\CC}}$
to the operation $\O_F([\sigma,\und{f}]) := [\sigma,F(\und{f})] : F(\und{M})\to F(N)$ in $\O_{\ovr{\DD}}$,
where $F(\und{f}) := (F(f_1),\dots, F(f_n))$ and $F(\und{M}) := (F(M_1),\dots,F(M_n))$
denote the actions of $F$ on tuples.
The pullback functor $F^\ast = \O_{F}^\ast$ 
from \eqref{eqn:Fpullbackfunctor} then coincides with the pullback functor of operad algebras.
\sk

A model for the operadic left Kan extension $F_! : \AQFT(\ovr{\CC})\to \AQFT(\ovr{\DD})$ 
can be given in terms of an ordinary left Kan extension along the 
induced functor $\O^\otimes_{F} : \O_{\ovr{\CC}}^\otimes\to \O_{\ovr{\DD}}^\otimes$
between the monoidal envelopes of the AQFT operads. For details about these concepts,
we refer the reader to \cite[Section 1.1]{Horel} and also to
\cite[Section 6]{BPSW} where the notations are closer to our present ones.
More explicitly, given any $\AAA\in\AQFT(\ovr{\CC})$, the operadic
left Kan extension $F_!(\AAA)\in \AQFT(\ovr{\DD})$ is defined 
object-wise by the colimit
\begin{flalign}
F_!(\AAA)(K)\,:=\,\colim\Big(\O_{F}^\otimes/K\longrightarrow \O^\otimes_{\ovr{\CC}}\stackrel{\AAA^\otimes}{\longrightarrow}\TT\Big)\quad,
\end{flalign}
for all $K\in\DD$, together with its canonically defined $\O_{\ovr{\DD}}$-algebra structure.
Let us briefly describe the building blocks of this colimit in more detail:
\begin{itemize}
\item The monoidal envelope $\O^\otimes_{\ovr{\CC}}$ of the AQFT operad $\O_{\ovr{\CC}}$ 
from Definition \ref{def:AQFToperad} is the category whose objects are all (possibly empty)
tuples $\und{M} := (M_1,\dots,M_n)$ of objects in $\O_{\ovr{\CC}}$. 
A morphism $\und{M}\to \und{N}$ in $\O^\otimes_{\ovr{\CC}}$ 
from $\und{M}=(M_1,\dots,M_n)$ to $\und{N} = (N_1,\dots,N_p)$
is a pair $\big(\alpha,\und{[\sigma,\und{f}]}\big)$ consisting of a map of sets
$\alpha : \{1,\dots,n\}\to \{1,\dots, p\}$ and a tuple
$\und{[\sigma,\und{f}]} := \big([\sigma_1,\und{f}_1],\dots,[\sigma_p,\und{f}_p]\big)$
of operations $[\sigma_j,\und{f}_j] : \und{M}_{\alpha^{-1}(j)}\to N_j$ in $\O_{\ovr{\CC}}$,
for all $j\in\{1,\dots,p\}$. Concatenation of tuples endows the category $\O^\otimes_{\ovr{\CC}}$ with
a symmetric monoidal structure.

\item The symmetric monoidal functor $\AAA^\otimes : \O^\otimes_{\ovr{\CC}}\to \TT$
is canonically defined from the operad algebra $\AAA\in\AQFT(\ovr{\CC}) \cong \Alg_{\O_{\ovr{\CC}}}(\TT)$
as follows: To each object $\und{M} = (M_1,\dots,M_n)\in\O^\otimes_{\ovr{\CC}}$ it assigns
the tensor product $\AAA^\otimes(\und{M}) := \bigotimes_{i=1}^n\AAA(M_i)\in\TT$, and to each morphism
$\big(\alpha,\und{[\sigma,\und{f}]}\big) : \und{M}\to\und{N}$ in $\O^\otimes_{\ovr{\CC}}$ it assigns the $\TT$-morphism
\begin{flalign}
\xymatrix{
\AAA^\otimes\big(\alpha,\und{[\sigma,\und{f}]}\big)\,:\,
\AAA^\otimes(\und{M}) \ar[r]^-{\text{permute}}~&~\bigotimes_{j=1}^p \AAA^\otimes\big(\und{M}_{\alpha^{-1}(j)}\big)
\ar[rr]^-{\Motimes_{j=1}^p \AAA([\sigma_j,\und{f}_j])}~&~~&~ \AAA^\otimes(\und{N})
}\quad,
\end{flalign}
where the permutation of tensor factors is via the symmetric braiding of $\TT$.

\item The category $\O_{F}^\otimes/K$ is the comma category of the functor 
$\O^\otimes_{F} : \O_{\ovr{\CC}}^\otimes\to \O_{\ovr{\DD}}^\otimes$ over the length one
tuple $K\in \O_{\ovr{\DD}}^\otimes$ consisting of the given object $K\in\DD$.
Hence, an object in $\O_{F}^\otimes/K$ is a pair $\big(\und{M},[\rho,\und{g}]\big)$ consisting
of an object $\und{M}\in \O_{\ovr{\CC}}^\otimes$ and an operation
$[\rho,\und{g}] : F(\und{M}) \to K$ in $\O_{\ovr{\DD}}$. A morphism
$\big(\und{M},[\rho,\und{g}]\big)\to \big(\und{N},[\tau,\und{h}]\big)$
in $\O_{F}^\otimes/K$ is a morphism $\big(\alpha,\und{[\sigma,\und{f}]}\big) : \und{M}\to \und{N}$
in $\O_{\ovr{\CC}}^\otimes$ such that the triangle 
\begin{flalign}
\begin{gathered}
\xymatrix{
\ar[dr]_-{[\rho,\und{g}]} F(\und{M}) \ar[rr]^-{(\alpha,\und{[\sigma,F(\und{f})]})}~&~ ~&~ F(\und{N}) \ar[dl]^-{[\tau,\und{h}]}\\
~&~ K ~&~
}
\end{gathered}
\end{flalign}
in $\O_{\ovr{\DD}}^\otimes$ commutes.
\end{itemize}

\section{\label{app:localization}Orthogonal localizations}
We develop via the calculus of fractions \cite{GabrielZisman}
the explicit model from Example \ref{ex:localizations}
for the orthogonal localization $\ovr{\COpen(M)}[W_{M}^{-1}]$ 
of $\ovr{\COpen(M)}$ at all Cauchy morphisms $W_M$.
The same construction will apply to give the 
explicit model from Example \ref{ex:localizations} for the orthogonal localization 
$\ovr{\RC(M)}[W_{\mathrm{rc},M}^{-1}]$.
Our conventions are those of \cite[Chapter 7]{KashiwaraSchapira}.
\begin{lem} \label{lem:localized:rightmultiplicative}
The set of Cauchy morphisms $W_M$ in the category $\COpen(M)$ is a right multiplicative system.
\end{lem}
\begin{proof}
We must verify that $W_M$ satisfies the four properties
(S1)--(S4) listed in \cite[Definition 7.1.5]{KashiwaraSchapira}.
\begin{enumerate}[(S1)]
\item All isomorphisms in $\COpen(M)$ are identities, hence Cauchy morphisms.

\item Any Cauchy morphism $U \subseteq U^\prime$ exhibits Cauchy surfaces in $U$ as Cauchy surfaces in $U^\prime$.
It follows that Cauchy morphisms are closed under composition.

\item Given any Cauchy morphism $U \subseteq U^\prime$ and any morphism $U \subseteq V$,
we argue that the union $U^\prime \cup V\subseteq M$ is causally convex,
and moreover that the inclusion $V \subseteq U^\prime \cup V$ is a Cauchy morphism.
This gives a (necessarily commuting) square
\begin{flalign}\label{eqn:localized:rightmultiplicative:S3}
\begin{gathered}
\xymatrix{
\ar[d]_-{\mathrm{Cauchy}}U \ar[r]~&~ V\ar[d]^-{\mathrm{Cauchy}}\\
U^\prime \ar[r]~&~ U^\prime\cup V
}
\end{gathered}
\end{flalign}
in $\COpen(M)$. To show causal convexity of $U^\prime \cup V\subseteq M$,
one verifies first that $J_M^\pm(U^\prime) \subseteq U^\prime \cup J_M^\pm(V)$ 
because $U \subseteq U^\prime$ is a Cauchy morphism.
It then suffices to consider any causal curve $\gamma : [0,1] \to M$ with $\gamma(0) \in U^\prime$ 
and $\gamma(1) \in V$. If $\gamma$ is future-directed, then for any $t \in [0,1]$ it holds that
\begin{flalign}
\nn \gamma(t) \,&\in\, J_M^+(U^\prime) \cap J_M^-(V)\\
\nn \,& \subseteq\, \big( U^\prime \cup J_M^+ (V)\big) \cap J_M^-(V)\\
\nn \,& =\, \big(U^\prime \cap J_M^-(V)\big) \cup \big( J_M^+ (V) \cap J_M^-(V) \big) \\
\, & \subseteq\, U^\prime \cup V \quad ,
\end{flalign}
where the last step involves also the causal convexity of $V$.
Thus $\gamma$ does not exit $U^\prime \cup V$.
A similar argument holds when $\gamma$ is past-directed.
To see that $V \subseteq U^\prime \cup V$ is a Cauchy morphism,
observe that any inextendable causal curve in $M$
which intersects $U^\prime \cup V$ must necessarily intersect $V$.
Indeed, if it intersects $U^\prime$ then, because $U \subseteq U^\prime$ is a Cauchy morphism, 
it also intersects $U$, which is contained in $V$.

\item The property (S4) is trivially satisfied since $\COpen(M)$ is thin.
\qedhere
\end{enumerate}
\end{proof}

As a consequence of this lemma, the calculus of fractions applies
to give a model for the localization
$L_M : \COpen(M) \to \COpen(M)[W_M^{-1}]$, see \cite[Theorem 7.1.16]{KashiwaraSchapira}.
Furthermore, orthogonal localization endows the localized category $\COpen(M)[W_M^{-1}]$ with
the orthogonality relation pushed forward along the localization functor $L_M$, see 
\cite[Proposition 2.11]{BCStimeslice}.
The resulting model for the orthogonal localization 
$L_M : \ovr{\COpen(M)} \to \ovr{\COpen(M)}[W_M^{-1}]$ is then given as follows:
\begin{itemize}
\item The category $\COpen(M)[W_M^{-1}]$ has the same objects as $\COpen(M)$,
and	its morphisms $[X] : U \to V$ are equivalence classes of objects $X \in \COpen(M)$
with $U \subseteq X \supseteq V$ such that $(V \subseteq X) \in W_M$ is a Cauchy morphism.
Two such $X, X^\prime \in \COpen(M)$ are equivalent if there exists
a third $X^{\prime\prime} \in \COpen(M)$ with $X \subseteq X^{\prime\prime} \supseteq X^\prime$
such that $(V \subseteq X^{\prime\prime})\in W_M$ is a Cauchy morphism.
The composite of $[X] : U \to V$ and $[Y] : V \to W$ is given by $[Y] \circ [X] = [X \cup Y]$,
using \eqref{eqn:localized:rightmultiplicative:S3}.

\item The orthogonality relation on $\ovr{\COpen(M)}[W_M^{-1}]$ is characterized as follows:
$([X_1] : U_1 \rightarrow V) \perp ([X_2] : U_2 \to V)$ if and only if
$(U_1 \subseteq X_1 \cup X_2) \perp (U_2 \subseteq X_1 \cup X_2)$ in $\ovr{\COpen(M)}$,
or equivalently $U_1$ and $U_2$ are causally disjoint in $M$.

\item The orthogonal localization functor
$L_M : \overline{\COpen(M)} \to \ovr{\COpen(M)}[W_M^{-1}]$
acts as identity on objects 
and sends a morphism $U \subseteq V$ in $\COpen(M)$ to $[V] : U \to V$.
\end{itemize}

This model can be simplified significantly.
Let us recall that the Cauchy development $D_M(S)\subseteq M$ of a subset
$S \subseteq M$ is the set of points $p \in M$ such that every inextendable causal curve 
through $p$ also intersects $S$.
An inclusion $U \subseteq V$ of causally convex open subsets of $M$ is a Cauchy morphism
if and only if $D_M(U) = D_M(V)$. Other useful properties of Cauchy
developments include that $D_M(D_M(S)) = D_M(S)$
and $f (D_M(S)) \subseteq D_N(f(S))$, for every subset $S \subseteq M$ and 
every $\Loc$-morphism $f : M \to N$, see also Lemma \ref{lem:Cauchy_development:Loc-morphisms}.
\begin{propo} \label{prop:localized:thin}
The category $\COpen(M)[W_M^{-1}]$ is thin,
i.e.\ there exists at most one morphism between every two objects.
Moreover, the unique morphism $U \to V$ exists if and only if
$U\subseteq D_M(V)$ is contained in the Cauchy development of $V$ in $M$.
\end{propo}
\begin{proof}
For any two parallel morphisms $[X], [X^\prime] : U \to V$,
one verifies using \eqref{eqn:localized:rightmultiplicative:S3} 
that $[X] = [X \cup X^\prime] = [X^\prime]$ since both $V \subseteq X$ 
and $V \subseteq X^\prime$ are Cauchy morphisms.
\sk

For the second statement, suppose that a morphism $[X] : U \to V$ exists.
Then we have a morphism $U \subseteq X$ and a Cauchy morphism $V \subseteq X$ in $\COpen(M)$,
hence $U \subseteq X \subseteq D_M(X) = D_M(V)$.
Conversely, suppose that $U \subseteq D_M(V)$.
We show that $[J_M^{+\cap -} (U \cup V)] : U \to V$ is a valid
morphism\footnote{It is also true that $[D_M(V)] : U \to V$ is
a valid morphism of $\COpen(M)[W_M^{-1}]$.
However, for $U, V  \in \RC(M)$, $D_M(V) \in \COpen(M)$ may fail to be relatively compact
and hence may not define a morphism $U \to V$ in $\RC(M)[W_{\mathrm{rc},M}^{-1}]$.
We prefer to present a proof that adapts straightforwardly to the relatively compact case.}
in $\COpen(M)[W_M^{-1}]$,
where $J_M^{+\cap -}(S) := J_M^+(S) \cap J_M^-(S)$ is the causally convex hull of $S \subseteq M$.
There are clearly inclusions
$U \subseteq J_M^{+\cap -}(U \cup V)$ and $V \subseteq J_M^{+\cap -}(U \cup V)$.
Using the hypothesis $U \subseteq D_M(V)$, one has
that $J_M^\pm (U) \subseteq J_M^\pm(D_M(V))$.
Then also
$J_M^\pm(U \cup V) = J_M^\pm(U) \cup J_M^\pm (V) \subseteq J_M^\pm (D_M(V))$,
so that $J_M^{+ \cap -} (U \cup V) \subseteq J_M^{+ \cap -} (D_M (V)) = D_M(V)$,
where the last equality expresses causal convexity of $D_M(V)$.
It follows that $D_M (J_M^{+ \cap -}(U \cup V)) = D_M(V)$,
i.e.\ the inclusion $V \subseteq J_M^{+ \cap -}(U \cup V)$ is a Cauchy morphism.
\end{proof}

For every $\Loc$-morphism $f : M \to N$,
the orthogonal functor 
$f : \ovr{\COpen(M)} \to \ovr{\COpen(N)}$ from \eqref{eqn:fpullback} 
preserves Cauchy morphisms, i.e.\ it maps $W_M$ to $W_N$, 
hence it induces an orthogonal functor
$f_W : \ovr{\COpen(M)}[W_M^{-1}] \to \ovr{\COpen(N)}[W_N^{-1}]$
such that $L_N \circ f = f_W \circ L_M$.
On objects, $f_W(U) := f(U)$ takes images under $f$.
Note that, by Proposition \ref{prop:localized:thin}, this defines a valid functor since
$U \subseteq D_M(V)$ implies $f(U) \subseteq f(D_M(V)) \subseteq D_N(f(V))$.
\begin{propo} \label{prop:localized:functor}
The orthogonal functor
$f_W : \ovr{\COpen(M)}[W_M^{-1}] \to \ovr{\COpen(N)}[W_N^{-1}]$
associated to any $\Loc$-morphism $f : M \to N$
is fully faithful and reflects orthogonality.
\end{propo}
\begin{proof}
The functor $f_W$ both preserves and reflects orthogonality because
$U_1, U_2 \subseteq M$ are causally disjoint if and only if
$f(U_1), f(U_2) \subseteq N$ are causally disjoint.
From Proposition \ref{prop:localized:thin}, fully faithfulness of $f_W$ means
equivalently that $f(U) \subseteq D_N(f(V))$ if and only if $U \subseteq D_M(V)$,
for all causally convex opens $U,V \in \COpen(M)$.
Suppose that $U \subseteq D_M(V)$.
Then $f(U) \subseteq f(D_M(V))\subseteq D_N(f(V))$
by the properties of Cauchy development noted above Proposition \ref{prop:localized:thin}.
Now suppose
that $f(U) \subseteq D_N(f(V))$ and take any point $p \in U$ and any
inextendable causal curve $\gamma : (-1,1) \to M$ with $\gamma(0) = p$.
Pick an extension $\widetilde{\gamma}$ of $f\circ \gamma : (-1,1)\to N$,
i.e.\ an inextendable causal curve $\widetilde{\gamma} : (a,b) \to N$
with $(-1,1) \subseteq (a,b)$ such that $\widetilde{\gamma}\vert_{(-1,1)} = f \circ \gamma$.
Because $\widetilde{\gamma}(0) = f(\gamma(0)) \in f(U) \subseteq D_N(f(V))$,
there exists $t \in (a,b)$ with $\widetilde{\gamma}(t) \in f(V)$.
Since $\gamma$ is inextendable in $M$,
$\widetilde{\gamma}(-1)$ and $\widetilde{\gamma}(1)$ (if defined) do not lie in $f(M)$.
Thus $t \in (-1,1)$ because $f(M)$ is causally convex.
We therefore have $f(\gamma(t)) = \widetilde{\gamma}(t) \in f(V)$,
so $\gamma (t) \in V$ because $f$ is injective.
This demonstrates that $p \in D_M(V)$.
\end{proof}

The above results specialize to the full orthogonal subcategory $\ovr{\RC(M)}\subseteq \ovr{\COpen(M)}$
of all relatively compact causally convex open subsets of $M$
because the constructions used in their proofs preserve relative compactness.
\begin{lem} \label{lem:hullofrelativelycompact}
Let $M\in \Loc$ be any object and $S \subseteq M$ a relatively compact subset.
Then the causally convex hull $J_M^{+ \cap -}(S)\subseteq M$ is also relatively compact.
\end{lem}
\begin{proof}
Because the closure $\mathrm{cl}(S)\subseteq M$ is by hypothesis compact,
it is a consequence of global hyperbolicity of $M$ \cite[Definition 4.117 and Theorem 4.12]{Minguzzi}
that $J_M^{+ \cap -}(\mathrm{cl}(S))\subseteq M$ is compact
and that $J_M^+(\mathrm{cl}(S))\subseteq M$ and $J_M^-(\mathrm{cl}(S))\subseteq M$ are closed.
It follows that $\mathrm{cl}\big(J_M^{+ \cap -} (S)\big) \subseteq
\mathrm{cl}\big(J_M^+(S)\big) \cap \mathrm{cl}\big(J_M^-(S)\big) 
\subseteq J_M^+(\mathrm{cl}(S)) \cap J_M^-(\mathrm{cl}(S))$
is a closed subset of a compact set, and hence is compact.
\end{proof}

\begin{cor}
The orthogonal localization $\ovr{\RC(M)}[W_{\mathrm{rc},M}^{-1}]$
of\, $\ovr{\RC(M)}$ at all Cauchy morphisms $W_{\mathrm{rc},M}$
is thin and admits the following explicit description:
Its objects are all objects $U\in \ovr{\RC(M)}$, a morphism $U \to V$
exists uniquely if and only if $U \subseteq D_M(V)$,
and $(U_1 \to V) \perp (U_2 \to V)$ are orthogonal
if and only if $U_1$ and $U_2$ are causally disjoint in $M$.
The orthogonal localization functor
$L_{\mathrm{rc},M}:\ovr{\RC(M)}\to \ovr{\RC(M)}[W_{\mathrm{rc},M}^{-1}]$
acts as the identity on objects and sends the morphism $U\subseteq V$ in $\RC(M)$
to $U\to V$.
Furthermore, the orthogonal functor
$f_{W} : \ovr{\RC(M)}[W_{\mathrm{rc},M}^{-1}] \to \ovr{\RC(N)}[W_{\mathrm{rc},N}^{-1}]$
associated to any $\Loc$-morphism $f : M \to N$
is fully faithful and reflects orthogonality.
\end{cor}
\begin{proof}
The proofs of Lemma \ref{lem:localized:rightmultiplicative}
and Propositions \ref{prop:localized:thin} and \ref{prop:localized:functor}
hold without alteration for $\RC(M)$ in place of $\COpen(M)$,
since unions and causally convex hulls of relatively compact subsets are relatively compact.
\end{proof}

\section{\label{app:bilimit}Technical details for Proposition \ref{propo:bilimovercover}}
In this appendix we supply the technical details which are needed to prove
Proposition \ref{propo:bilimovercover}. To simplify notation,
let us denote the pseudo-functor whose bilimit we wish to compute by
\begin{flalign}
X\,:\,\DD~\longrightarrow~\Pr^L\quad,
\end{flalign}
where $\DD$ is a small $1$-category with terminal object $\ast\in\DD$.
As in the context of Proposition \ref{propo:bilimovercover}, we assume
that the left adjoint functor $X(g): X(d)\to X(d^\prime)$ is fully faithful,
for all $\DD$-morphisms $g:d\to d^\prime$, and we denote its right adjoint (coreflector)
by $X^\dagger (g): X(d^\prime)\to X(d)$.
\sk

Following Construction \ref{constr:computingbilimits}, we compute the bilimit of $X$
by starting from the explicit model
\begin{flalign}
\bilim(X) \,=\, \Hom(\Delta\mathbf{1},X)
\end{flalign}
given by the category of pseudo-natural transformations and modifications.
Spelling out these data, one finds the following explicit description:
\begin{itemize}
\item An object in $\bilim(X)$ is a tuple $(\{x_d\},\{\varphi_g \})$
consisting of a family of objects $x_d\in X(d)$, for all $d\in\DD$,
and a family of isomorphisms $\varphi_g : X(g)(x_d)\to x_{d^\prime}$ in $X(d^\prime)$,
for all $\DD$-morphisms $g:d\to d^\prime$. These data have to satisfy the
conditions that
\begin{subequations}\label{eqn:bilimcocycleAPP}
\begin{flalign}
\begin{gathered}
\xymatrix@C=4em{
\ar[d]_-{\cong}X(g^\prime)X(g)(x_d)\ar[r]^-{X(g^\prime)(\varphi_g)}~&~X(g^\prime)(x_{d^\prime})\ar[d]^-{\varphi_{g^\prime}}\\
X(g^\prime g)(x_d)\ar[r]_-{\varphi_{g^\prime g}}~&~x_{d^{\prime\prime}}
}
\end{gathered}
\end{flalign}
commutes in $X(d^{\prime\prime})$, for all composable $\DD$-morphisms $g: d\to d^\prime$ and 
$g^\prime : d^\prime\to d^{\prime\prime}$, and that
\begin{flalign}
\begin{gathered}
\xymatrix@C=4em{
X(\id_d)(x_d)\ar[r]^-{\varphi_{\id_d}}~&~x_d\\
\ar[u]^-{\cong}x_d\ar[ru]_-{\id_{x_d}}~&~
}
\end{gathered}
\end{flalign}
\end{subequations}
commutes in $X(d)$, for all $d\in\DD$.

\item A morphism in $\bilim(X)$ is a tuple $\{\psi_d\} : (\{x_d\},\{\varphi_g \})\to
\big(\{x^\prime_d\},\{\varphi^\prime_g \}\big)$ consisting of a family
of morphisms $\psi_d : x_d\to x^\prime_d$ in $X(d)$, for all $d\in \DD$, such that
\begin{flalign}\label{eqn:bilimmorAPP}
\begin{gathered}
\xymatrix@C=4em{
\ar[d]_-{X(g)(\psi_d)} X(g)(x_d) \ar[r]^-{\varphi_g} ~&~ x_{d^\prime} \ar[d]^-{\psi_{d^\prime}}\\
X(g)(x^\prime_d) \ar[r]_-{\varphi^\prime_{g}}~&~ x^{\prime}_{d^\prime}
}
\end{gathered}
\end{flalign}
commutes in $X(d^{\prime})$, for all $\DD$-morphisms $g:d\to d^\prime$.
\end{itemize}

We will now simplify this description by using our hypotheses on the category $\DD$ and the 
pseudo-functor $X$. Consider the canonical projection functor
\begin{flalign}\label{eqn:bilimtoterminalAPP}
\bilim(X)~\longrightarrow~X(\ast)
\end{flalign}
from the bilimit to the value of $X$ on the terminal object.
Explicitly, this functor assigns to an object $(\{x_d\},\{\varphi_g \})\in\bilim(X)$
the component $x_\ast\in X(\ast)$ at the terminal object,
and to a morphism $\{\psi_d\}: (\{x_d\},\{\varphi_g \}) \to (\{x^\prime_d\},\{\varphi^\prime_g \})$ 
in $\bilim(X)$ the component $\psi_\ast : x_\ast\to x^\prime_\ast$ in $X(\ast)$ at the terminal object.
\begin{lem}\label{lem:FffAllAPP}
The functor $\bilim (X) \to X(\ast)$ of \eqref{eqn:bilimtoterminalAPP} is fully faithful.
\end{lem}
\begin{proof}
Specializing the commutative diagrams \eqref{eqn:bilimmorAPP}
to the terminal $\DD$-morphisms $t_d : d\to\ast$, we obtain commutative diagrams
\begin{flalign}\label{eqn:APPtmp}
\begin{gathered}
\xymatrix@C=3em{
\ar[d]_-{X(t_d)(\psi_d)} X(t_d)(x_d) \ar[r]^-{\varphi_{t_d}}_-{\cong} ~&~ x_{\ast} \ar[d]^-{\psi_{\ast}}\\
X(t_d)(x^\prime_d) \ar[r]_-{\varphi^\prime_{t_d}}^-{\cong}~&~ x^{\prime}_{\ast}
}
\end{gathered}
\end{flalign}
which, together with the fact that $X(t_d)$ is fully faithful by our hypotheses, 
imply that all components $\psi_d$ of a morphism
$\{\psi_d\}: (\{x_d\},\{\varphi_g \}) \to (\{x^\prime_d\},\{\varphi^\prime_g \})$ in $\bilim(X)$
are uniquely determined by $\psi_\ast$. This shows faithfulness.
To prove fullness, we have to show that, given any $\psi_\ast$, the components
$\psi_d$ defined uniquely by \eqref{eqn:APPtmp} make the diagrams \eqref{eqn:bilimmorAPP} commute,
for all $\DD$-morphisms $g:d\to d^\prime$. Note that this is equivalent to verifying that the diagrams
obtained by applying the fully faithful functor $X(t_{d^\prime})$ to \eqref{eqn:bilimmorAPP} commute.
The latter can be expanded as
\begin{flalign}
\begin{gathered}
\xymatrix@C=3.5em@R=2.5em{
\ar[ddd]_-{X(t_d)(\psi_d)}X(t_{d})(x_d) \ar[rrr]_-{\cong}^-{\varphi_{t_d}}~&~ ~&~ ~&~ \ar[ddd]^-{\psi_\ast} x_\ast\\
~&~ \ar[ul]_-{\cong}\ar[d]_-{X(t_{d^\prime})X(g)(\psi_d)} X(t_{d^\prime})X(g)(x_d) \ar[r]_-{\cong}^-{X(t_{d^\prime})(\varphi_g)} ~&~ X(t_{d^\prime})(x_{d^\prime}) \ar[d]^-{X(t_{d^\prime})(\psi_{d^\prime})}\ar[ru]_-{\cong}^-{\varphi_{t_{d^\prime}}}~&~\\
~&~ \ar[dl]^-{\cong}X(t_{d^\prime})X(g)(x^\prime_d) \ar[r]^-{\cong}_-{X(t_{d^\prime})(\varphi^\prime_{g})}~&~ X(t_{d^\prime})(x^{\prime}_{d^\prime})\ar[rd]^-{\cong}_-{\varphi^\prime_{t_{d^\prime}}}~&~\\
X(t_{d})(x^\prime_d)\ar[rrr]^-{\cong}_-{\varphi^\prime_{t_{d}}}~&~ ~& ~&~ x^\prime_\ast
}
\end{gathered} \qquad .
\end{flalign}
The bottom and top square commute as a consequence of \eqref{eqn:bilimcocycleAPP},
the left square commutes because $X$ is a pseudo-functor,
and the right and outer squares commute as a consequence of \eqref{eqn:APPtmp}.
Hence, the middle square commutes too, which completes the proof.
\end{proof}

In order to characterize the bilimit of our pseudo-functor $X$, it remains to
compute the essential image $\mathcal{X}\subseteq X(\ast)$ 
of the fully faithful functor \eqref{eqn:bilimtoterminalAPP}.
This then yields a factorization
\begin{flalign}
\bilim(X)~\stackrel{\simeq}{\longrightarrow}~\mathcal{X}~ \stackrel{\subseteq}{\longrightarrow} X(\ast)
\end{flalign}
which provides a model for the bilimit as a coreflective full subcategory of $X(\ast)$, i.e.\ there exists
a right adjoint coreflector $X(\ast)\to \mathcal{X}$.
\begin{lem}
The essential image of the functor \eqref{eqn:bilimtoterminalAPP}
is the full subcategory $\mathcal{X}\subseteq X(\ast)$  
consisting of all objects $x_\ast\in X(\ast)$ which satisfy the following conditions:
For every object $d\in\DD$, the $x_\ast$-component of the counit
\begin{flalign}
(\epsilon_{t_d})_{x_\ast}\,:\, X(t_d)\,X^\dagger(t_d)(x_\ast)~\stackrel{\cong}{\longrightarrow}~ x_\ast
\end{flalign}
of the adjunction $X(t_d)\dashv X^\dagger(t_d)$ associated with the terminal $\DD$-morphism
$t_d:d\to \ast$ is an isomorphism.
\end{lem}
\begin{proof}
To show that the essential image is contained in $\mathcal{X}$,
consider any object $(\{x_d\},\{\varphi_g\})\in\bilim(X)$ and observe that 
we have isomorphisms $\varphi_{t_d} : X(t_d)(x_d) \to x_\ast$,
for every object $d\in\DD$. The counit condition then follows from the commutative diagram
\begin{flalign}
\begin{gathered}
\xymatrix@C=4em{
X(t_d)\,X^\dagger(t_d)(x_\ast) \ar[r]^-{(\epsilon_{t_d})_{x_\ast}}~&~x_\ast\\
\ar[u]_-{\cong}^-{X(t_d)\,X^\dagger(t_d)(\varphi_{t_d} )}X(t_d)\,X^\dagger(t_d) X(t_d)(x_d)\ar[r]^-{(\epsilon_{t_d})_{X(t_d)(x_d)}}~&~X(t_d)(x_d)\ar[u]^-{\cong}_-{\varphi_{t_d}}\\
\ar[u]_-{\cong}^-{X(t_d)(\eta_{t_d})_{x_d}}X(t_d)(x_d) \ar[ru]^-{\cong}_-{~~~~\id_{X(t_d)(x_d)}}~&~
}
\end{gathered}
\end{flalign}
where we recall that the unit $\eta_{t_d}$ is a natural isomorphism
because $X(t_d)$ is by hypothesis fully faithful.
\sk

To show that each object $x_\ast\in\mathcal{X}$ lies in the essential image,
let us define the tuple of objects
\begin{subequations}
\begin{flalign}
\big\{x_d \,:=\, X^\dagger(t_d)(x_\ast)\,\in\,X(d) \big\}
\end{flalign}
and the tuple of morphisms
\begin{flalign}
\left\{\begin{gathered}
\xymatrix@C=2em{
\ar@{=}[d] X(g)x_d \ar[rrr]^-{\varphi_g}~&~ ~&~ ~&~ x_{d^\prime}\ar@{=}[d]\\
 X(g)X^\dagger(t_d)(x_\ast) \ar[r]_-{\cong}~&~  X(g)X^\dagger(g)X^\dagger(t_{d^\prime})(x_\ast)\ar[rr]_-{(\epsilon_g)_{X^\dagger(t_{d^\prime})(x_\ast)}}~&~~&~X^\dagger(t_{d^\prime})(x_\ast)
}
\end{gathered}
\right\}\quad,
\end{flalign}
\end{subequations}
where in the unnamed isomorphism we used the unique 
factorization $t_d = t_{d^\prime}\, g$ of the terminal morphism. 
Note that the morphisms $\varphi_g$ are isomorphisms because, applying the fully faithful
functor $X(t_{d^\prime})$
to the relevant component of $\epsilon_g$,
one obtains the commutative diagram
\begin{flalign}
\begin{gathered}
\xymatrix@C=8em{
\ar[d]_-{\cong} X(t_{d^\prime})X(g)X^\dagger(g)X^\dagger(t_{d^\prime})(x_\ast)\ar[r]^-{X(t_{d^\prime})(\epsilon_g)_{X^\dagger(t_{d^\prime})(x_\ast)}}~&~X(t_{d^\prime})X^\dagger(t_{d^\prime})(x_\ast)\ar[d]_-{\cong}^-{(\epsilon_{t_{d^\prime}})_{x_\ast}}\\
X(t_{d})X^\dagger(t_d)(x_\ast)\ar[r]^-{\cong}_-{(\epsilon_{t_d})_{x_\ast}}~&~x_\ast
}
\end{gathered}
\end{flalign}
where the bottom horizontal and right vertical morphisms are isomorphisms by definition of the full subcategory
$\mathcal{X}\subseteq X(\ast)$. One directly checks that the tuple $(\{x_d\},\{\varphi_g\})$
introduced above defines an object in $\bilim(X)$. This object maps
under the functor \eqref{eqn:bilimtoterminalAPP} to $X^\dagger(t_\ast)(x_\ast) \cong x_\ast\in \mathcal{X}$, which 
completes the proof of essential surjectivity.
\end{proof}

\section{\label{app:Lorentz_geometry_details}Technical details for Theorem \ref{theo:WRCHKstack}}
Our proof of Theorem \ref{theo:WRCHKstack} requires the following facts 
about Lorentzian geometry, which we prove in this appendix.
\begin{propo} \label{prop:Cauchy_development:small_D-stable_neighbourhoods}
Let $p \in M$ be any point in any object $M \in \Loc$ 
and $U \subseteq M$ any neighborhood of $p$.
Then there exists a $D$-stable causally convex open subset $V \subseteq M$ such that $p \in V \subseteq U$.
\end{propo}

\begin{propo} \label{prop:Cauchy_development:closure_in_D_stable_image}
Let $f : M \to N$ be any $\Loc$-morphism with $D$-stable image,
i.e.\ $D_N(f(M)) = f(M)$, and $U \subseteq M$ any relatively compact subset.
Then one has $\cl(D_N(f(U))) \subseteq f(M)$, where the closure is taken inside $N$.
\end{propo}

\begin{rem} \label{rem:Cauchy_development:fine_D-stable_covers}
In Subsections \ref{subsubsec:precostack:timeslice} and \ref{subsubsec:descent:timeslice}
we make use of the Grothendieck topology on $\Loc$
given by all $D$-stable causally convex open covers.
Proposition \ref{prop:Cauchy_development:small_D-stable_neighbourhoods}
shows that any $M \in \Loc$ has arbitrarily small $D$-stable causally convex 
open neighborhoods around each of its points.
In other words, the aforementioned Grothendieck topology 
contains arbitrarily fine refinements.
\end{rem}

The proofs of Propositions \ref{prop:Cauchy_development:small_D-stable_neighbourhoods}
and \ref{prop:Cauchy_development:closure_in_D_stable_image} below 
will require statements of a Lorentz geometric nature concerning the properties of Cauchy developments.
We refer to \cite{Minguzzi} for a comprehensive review of Lorentzian causality theory,
and cite this review rather than original sources to give a unified resource for the reader.
For the remainder of this appendix
a \textit{spacetime} will mean a time-oriented Lorentzian manifold.
(Objects of $\Loc$ are thus the oriented globally hyperbolic spacetimes of a chosen dimension.)
\begin{lem} \label{lem:Cauchy_development:Loc-morphisms}
Let $M$ and $N$ be any spacetimes of the same dimension,
$f : M \to N$ a time-orientation preserving isometric embedding with causally convex image,
and $U \subseteq M$ any subset. Then $f(D_M (U)) = D_N(f(U)) \cap f(M)$.
Moreover, if $M$ is globally hyperbolic and
$D_M(U) \subseteq M$ is a relatively compact subset, then $D_N(f(U)) \subseteq f(M)$.
\end{lem}
\begin{proof}
The inclusion $f(D_M(U)) \subseteq D_N(f(U)) \cap f(M)$ 
follows from the fact that, under the above assumptions, 
each inextendable causal curve in $N$ admits a unique restriction 
(along $f$) to an inextendable causal curve in $M$. 
For the inclusion $f(D_M(U)) \supseteq D_N(f(U)) \cap f(M)$, 
take $p \in M$ such that $f(p) \in D_N(f(U))$ and consider 
an inextendable causal curve $\gamma$ in $M$ through $p$. 
Under the above assumptions, there exists an inextendable causal 
curve $\hat{\gamma}$ in $N$ that restricts (along $f$) to $\gamma$. 
By construction, $\hat{\gamma}$ goes through $f(p) \in D_N(f(U))$. 
This entails that $\hat{\gamma}$ hits $f(U)$ and therefore $\gamma$ hits $U$.
This shows that $p \in D_M(U)$. 
\sk

It remains to show that $D_N(f(U)) \subseteq f(M)$ when $D_M(U) \subseteq M$ 
is a relatively compact subset and $M$ is globally hyperbolic.
We will in fact not need global hyperbolicity of $M$,
but merely a weaker causal property which is implied by it:
Every globally hyperbolic spacetime $M$ is also non-partially imprisoning \cite[Definition 4.68]{Minguzzi},
i.e.\ inextendable causal curves in $M$ are proper maps.
Take $p \in D_N(f(U))$,
and let $\gamma : \mathbb{R} \to N$ be any future-directed inextendable causal curve with $\gamma(0) = p$.
Define $(a,b) := \gamma^{-1}(f(M))$, which is an interval because $f(M)\subseteq N$ is causally convex.
We show that $a < 0 < b$,
from which it follows that $p = \gamma(0) \in f(M)$.
\sk

Assume that $b \leq 0$.
Let $\hat{\gamma} : (a,b) \to M$ be the unique restriction of $\gamma$ along $f$,
i.e.\ $f \circ \hat{\gamma} = \gamma\vert_{(a,b)}$.
Then $\hat{\gamma}$ is a future-directed inextendable causal curve in $M$.
Our hypotheses imply that $\cl(D_M(U)) \subseteq M$ is compact and $\hat{\gamma}$ is a proper map,
so $\hat{\gamma}^{-1}(\cl(D_M(U))) \subseteq (a,b)$ is compact.
Thus there exists some $t_0 \in (a,b)$ with $\hat{\gamma}(t) \not \in D_M(U)$ for all $t_0 \leq t < b$.
It follows that there exists a future-directed past-inextendable causal curve
$\eta$ in $M$ with future-endpoint $\hat{\gamma}(t_0)$ which does not intersect $U$.
The concatenation of $\gamma\vert_{[t_0, \infty)}$ after $f\circ \eta$
is thus a future-inextendable causal curve in $N$ which does not intersect $f(U)$.
Let $\delta$ be any future-directed inextendable causal curve in $N$
extending the above concatenated curve.
It follows by past-inextendability of $\eta$ in $M$ and causal convexity of $f(M) \subseteq N$ 
that no past extension of $f \circ \eta$ in $N$ intersects $f(U) \subseteq f(M)$.
So, $\delta$ never intersects $f(U)$.
But our assumption $b \leq 0$ gives that
$\delta$ passes through $p = \delta(0) = \gamma\vert_{[t_0, \infty)}(0)$,
which is a contradiction with $p \in D_N(f(U))$.
We conclude that $b > 0$.
A similar argument shows that $a < 0$.
\end{proof}

Our proof of Proposition \ref{prop:Cauchy_development:small_D-stable_neighbourhoods}
will make key use of strictly convex normal neighborhoods \cite[Definition 2.3]{Minguzzi},
via the Lemma \ref{lem:Cauchy_development:diamonds_in_convex_normal} below.
It is generically true that, under the exponential map
$\exp_p : U_p \subseteq T_p M \to U \subseteq M$ at a point $p \in M$,
the image of the future (past) light-cone in $U_p \subseteq T_p M$ is included in
the causal future $J_U^+(p)$ (causal past $J_U^-(p)$, respectively) of $p$.
Similar inclusions hold for the time-cone and chronological future/past $I_U^\pm(p)$.
Strictly convex normal neighborhoods $U \subseteq M$ have the special property
that these inclusions are equalities \cite[Corollary 2.10]{Minguzzi}.
\begin{lem} \label{lem:Cauchy_development:diamonds_in_convex_normal}
Let $M$ be a spacetime,
$U \subseteq M$ a strictly convex normal, globally hyperbolic, 
causally convex open subset and $p_1, p_2 \in U$.
Then $I_U^+(p_1) \cap I_U^-(p_2) \subseteq M$ is a $D$-stable relatively compact 
open subset.
\end{lem}
\begin{proof}
Relative compactness both in $U$ and in $M$ follows from the inclusion
$I_U^+(p_1) \cap I_U^-(p_2) \subseteq J_U^+(p_1) \cap J_U^-(p_2)$ 
into a compact (by global hyperbolicity of $U$) subset.
We will show that, for $p \in U$, $I_U^\pm(p) \subseteq U$ is $D$-stable.
This entails $D$-stability of $I_U^+(p_1) \cap I_U^-(p_2) \subseteq U$ via
\begin{flalign}
D_U\left(I_U^+(p_1) \cap I_U^-(p_2)\right) \, \subseteq\, 
D_U\left(I_U^+(p_1)\right) \cap D_U\left(I_U^-(p_2)\right) \,=\, I_U^+(p_1) \cap I_U^-(p_2)\quad .
\end{flalign}
Then, because $U$ is globally hyperbolic,
Lemma \ref{lem:Cauchy_development:Loc-morphisms} applies to
the inclusion $U \hookrightarrow M$ and the subset 
$I_U^+(p_1) \cap I_U^-(p_2) \subseteq U$ to give that
$D_U(I_U^+(p_1) \cap I_U^-(p_2))
= D_M(I_U^+(p_1) \cap I_U^-(p_2)) \cap U$ 
and $D_M(I_U^+(p_1) \cap I_U^-(p_2)) \subseteq U$, hence 
\begin{flalign}
D_M\left(I_U^+(p_1) \cap I_U^-(p_2)\right) \,=\, D_U\left(I_U^+(p_1) \cap I_U^-(p_2)\right) \,=\, 
I_U^+(p_1) \cap I_U^-(p_2)\quad .
\end{flalign}

Given $p \in U$, let us now show that $I_U^+(p) \subseteq U$ is $D$-stable.
The proof for $I_U^-(p)$ is similar.
Let $q \in D_U(I_U^+(p))$, 
so any inextendable future-directed causal curve $\gamma : \mathbb{R} \to U$
with $\gamma(0) = q$ intersects $I_U^+(p)$ at some $t \in \mathbb{R}$.
Then $\gamma$ itself exhibits $\gamma(t^\prime) \in J_U^+(I_U^+ (p)) = I_U^+ (p)$
for all $t^\prime \geq t$, where we have used the
``push-up lemma'' \cite[Theorem 2.24]{Minguzzi} in the last equality.
Because $I_U^+(p) \subseteq U$ is open, it follows that
there exists $t_0 \in \mathbb{R}$ such that
$\gamma (t) \in I_U^+(p)$ for all $t > t_0$
and $\gamma(t) \not \in I_U^+(p)$ for all $t \leq t_0$.
In particular, $\gamma(t_0) \in \partial I_U^+(p)$ is a boundary point.
We show that $t_0 < 0$, from which $q = \gamma(0) \in I_U^+(p)$ follows.
\sk

Assume $t_0 \geq 0$.
Because $U$ is globally hyperbolic and hence $J_U^+(p) \subset U$ is closed,
we have $\gamma(t_0) \in \partial I_U^+(p) = \partial J_U^+(p) = J_U^+(p) \setminus I_U^+(p)$.
It follows from convex normality of $U \subseteq M$
\cite[Corollary 2.10]{Minguzzi} that either $\gamma(t_0) = p$
or there exists a unique future-directed null vector $v \in T_pM$
such that $\gamma(t_0) = \exp_p(v)$.
In the latter case, $t \mapsto \eta(t) := \exp_p((t + 1) v)$
describes an inextendible future-directed null geodesic in $U$ through $p$,
with $\eta(0) = \gamma(t_0)$ and
$\eta(t) \in J_U^+(p) \setminus I_U^+(p)$ for all $t \geq 0$.
In the former case of $\gamma(t_0) = p$, pick any future-directed null vector $v \in T_p M$.
Then $t \mapsto \eta(t) := \exp_p(t v)$ gives again an
inextendable future-directed null geodesic with $\eta(0) = p = \gamma(t_0)$
and $\eta(t) \in J_U^+(p) \setminus I_U^+(p)$ for all $t \geq 0$.
In either case, concatenating $\eta\vert_{[0,\infty)}$ after $\gamma\vert_{(-\infty,t_0]}$
gives an inextendable causal curve in $U$ through
$q = \gamma(0) \in D_U(I_U^+(p))$ that does not hit $I_U^+(p)$, a contradiction.
\end{proof}

The preceding Lemma \ref{lem:Cauchy_development:diamonds_in_convex_normal} 
allows us to construct the arbitrarily small $D$-stable neighborhoods 
stipulated in Proposition \ref{prop:Cauchy_development:small_D-stable_neighbourhoods}.
\begin{proof}[Proof of Proposition \ref{prop:Cauchy_development:small_D-stable_neighbourhoods}]
It is a standard result \cite[Theorems 1.35 and 2.7]{Minguzzi} that,
because the spacetime $M$ is globally hyperbolic and hence strongly causal,
each point $p \in M$ has a nested neighborhood basis $\{V_k \}_{k \in \mathbb{N}}$ consisting of
strictly convex normal, globally hyperbolic, causally convex,
relatively compact and open neighborhoods $V_k \subseteq M$.
Thus for sufficiently large $k$, $V_k \subseteq U$.
Pick $p_1 \in I_{V_k}^-(p)$ and $p_2 \in I_{V_k}^+(p)$.
By Lemma \ref{lem:Cauchy_development:diamonds_in_convex_normal},
$V := I_{V_k}^+(p_1) \cap I_{V_k}^-(p_2) \subseteq M$ is a $D$-stable
relatively compact open subset. 
Furthermore, by construction $p \in V \subseteq V_k \subseteq U$.
That the chronological diamond is also causally convex in $V_k$ 
and hence in $M$ is a straightforward consequence of the
``push-up lemma'' \cite[Theorem 2.24]{Minguzzi}.
\end{proof}

Our proof of Proposition \ref{prop:Cauchy_development:closure_in_D_stable_image}
will make use of certain subsets of a spacetime $M$ which bound above or below the 
Cauchy development of $U \subseteq M$.
These subsets are constructed by taking the double causal complement
(i.e.\ the causal complement of the causal complement) of $U$.
Recall that, for $U \subseteq M$,
the \textit{causal complement} $U^\prime$ of $U$ in $M$ is the
maximal subset of $M$ that is causally disjoint from $U$, i.e.\
\begin{flalign}
U^\prime \,:=\,  M \setminus J_M(U) \,=\, M \setminus \big(J_M^+(U) \cup J_M^-(U)\big)\,\subseteq\, M\quad.
\end{flalign}
\begin{lem} \label{lem:Cauchy_development:in_double_complement}
Let $M$ be any spacetime. Then $D_M(U) \subseteq U^{\prime\prime}$ 
for any subset $U \subseteq M$.
\end{lem}
\begin{proof}
Let us check the equivalent complementary inclusion
$M \setminus D_M(U) \supseteq J_M(M \setminus J_M(U))$.
Take $p \in J_M(M \setminus J_M(U))$. Then there exists an inextendible
future-directed causal curve $\gamma$ through $p$
that hits some $q \in M \setminus J_M(U)$. Such $\gamma$ does not meet $U$
(otherwise $q \in J_M(U)$, a contradiction), hence
$p \in M \setminus D_M(U)$ is not in the Cauchy development of $U$.
\end{proof}

\begin{lem} \label{lem:Cauchy_development:double_complement_in_hull}
Let $M$ be a globally hyperbolic spacetime.
Then $U^{\prime\prime} \subseteq D_M (J^{+\cap -}_M(U))$
for any relatively compact open subset $U \subseteq M$.
\end{lem}
\begin{proof}
Take $p \in U^{\prime\prime}$ and any inextendable future-directed causal curve $\gamma$ in $M$ through $p$.
Observe first that $\gamma$ lies inside $J_M(U) = M \setminus U^\prime$ because $U^\prime$ and $U^{\prime\prime}$ are causally disjoint.
Using that $U\subseteq M$ is relatively compact and $M$ is globally hyperbolic,
there exist Cauchy surfaces $\Sigma^+$ 
and $\Sigma^- \subseteq I_M^-(\Sigma^+)$ of $M$ 
which lie in the future and past of $U$ respectively, i.e.\
$\Sigma^\pm \cap J_M^\mp(U) = \emptyset$.
Then $\gamma$ intersects $\Sigma^\pm$ at a unique point, call it $\gamma(t^\pm) \in \Sigma^\pm$. 
It follows that
$\gamma(t^\pm) \in \Sigma^\pm \cap J_M(U) = \Sigma^\pm \cap J_M^\pm(U)$.
Note that $t^+ > t^-$ because $\gamma$ is future-directed and $\Sigma^- \subseteq I_M^-(\Sigma^+)$
and recall that $\gamma(t) \in J_M(U)$ for all $t \in [t^-, t^+]$.
Because $J_M^\pm(U) \subseteq M$ are open subsets,
the curve $\gamma$ is continuous,
and the endpoints of the interval are such that $\gamma(t^\pm) \in J_M^\pm(U)$,
there exists $t \in (t^-, t^+)$ with $\gamma(t) \in J_M^{+ \cap -}(U)$.
This shows that $p \in D_M(J_M^{+ \cap -}(U))$. 
\end{proof}

The next statement is an immediate consequence 
of Lemmas \ref{lem:Cauchy_development:in_double_complement}
and \ref{lem:Cauchy_development:double_complement_in_hull}.
\begin{cor}\label{cor:D=''}
Let $M$ be a globally hyperbolic spacetime.
Then $U^{\prime\prime} = D_M(U)$ for any relatively compact causally convex open subset $U \subseteq M$.
\end{cor}

\begin{proof}[Proof of Proposition \ref{prop:Cauchy_development:closure_in_D_stable_image}]
From Lemma \ref{lem:Cauchy_development:in_double_complement}, we have
$D_N(f(U)) \subseteq D_N(f(\cl(U))) \subseteq f(\cl(U))^{\prime\prime}$.
Because $\cl(U) \subseteq M$ is a compact subset 
and $N$ is a globally hyperbolic spacetime,
it follows that $J_N(f(\cl(U))) \subseteq N$ is closed
\cite[Theorem 4.12]{Minguzzi}, i.e.\ $f(\cl(U))^{\prime} \subseteq N$ is open.
Therefore $f(\cl(U))^{\prime\prime} \subseteq N$ is closed, and hence
\begin{flalign} \label{eq:Cauchy_development:closure_in_D_stable_image:1}
\cl (D_N(f(U))) \,\subseteq\, f(\cl(U))^{\prime\prime} \quad .
\end{flalign}

We claim that there exists a relatively compact causally convex open subset 
$V \subseteq N$ such that $f(\cl(U)) \subseteq V \subseteq f(M)$. 
We have
\begin{flalign} \label{eq:Cauchy_development:closure_in_D_stable_image:2}
V^{\prime\prime} \,=\, D_N(V) \,\subseteq\,  D_N(f(M)) \,=\, f(M) \quad,
\end{flalign}
where the first equality follows from Corollary \ref{cor:D=''},
while the last equality uses the hypothesis that $f(M) \subseteq N$ is $D$-stable.
Then the inclusion $\cl (D_N(f(U))) \subseteq f(M)$ is obtained 
combining \eqref{eq:Cauchy_development:closure_in_D_stable_image:1}
and \eqref{eq:Cauchy_development:closure_in_D_stable_image:2}
with the inclusion $f(\cl(U))^{\prime\prime} \subseteq V^{\prime\prime}$, 
which follows from $f(\cl(U)) \subseteq V$.
\sk

To conclude the proof, let us construct a relatively compact causally convex open subset 
$V \subseteq N$ such that $f(\cl(U)) \subseteq V \subseteq f(M)$. 
For each $p \in f(\cl(U))$,
take a relatively compact, causally convex and open neighborhood 
$V_p \subseteq f(M)$ of $p$, which exists 
because $N$ is globally hyperbolic and hence strongly causal \cite[Theorem 1.35]{Minguzzi}.
Since $\{V_p\}$ is an open cover of the compact subset $f(\cl(U)) \subseteq N$, 
one finds a finite subcover $\{V_{p_1}, \ldots , V_{p_n} \}$.
Recalling also Lemma \ref{lem:hullofrelativelycompact},  
$V = J_N^{+\cap -} \left(\cup_{i=1}^n V_{p_i}\right) \subseteq N$ 
is a relatively compact causally convex open subset. Furthermore, 
$f(\cl(U)) \subseteq V \subseteq f(M)$ holds by construction 
and because $f(M) \subseteq N$ is causally convex. 
\end{proof}

%%%%%%%%%%%%%%%%%%%%%%%%

\end{document}